\documentclass{article}
\usepackage{amsmath,amssymb,amsfonts}

\newcommand{\beq}{\begin{equation}}
\newcommand{\eeq}{\end{equation}}
\newcommand{\ben}{\begin{eqnarray}}
\newcommand{\een}{\end{eqnarray}}
\newcommand{\be}{\begin{eqnarray*}}
\newcommand{\ee}{\end{eqnarray*}}

\newcounter{eqalph}
\newcounter{equationa}

\let\ssection=\section
\renewcommand{\section}{\setcounter{equation}{0}\ssection}

\newcounter{example}[section]
\newcounter{remark}[section]
\newcounter{theorem}[section]
\newcounter{proposition}[section]
\newcounter{lemma}[section]
\newcounter{corollary}[section]
\newcounter{definition}[section]

\def\theremark{\arabic{section}.\arabic{remark}}
\def\thetheorem{\arabic{section}.\arabic{theorem}}

\def\thedefinition{\arabic{section}.\arabic{definition}}

\newenvironment{proof}{\noindent {\textit{Proof:}}
}{$\Box$\medskip}
\newenvironment{example}{\refstepcounter{remark}\medskip\noindent{\bf
Example \theremark:} }{$\Box$ \medskip}
\newenvironment{remark}{\refstepcounter{remark}\medskip\noindent{\bf
Remark \theremark:} }{$\Box$\medskip}
\newenvironment{theorem}{\refstepcounter{theorem}
\medskip\noindent{\sc Theorem \thetheorem}:}{$\Box$\medskip}
\newenvironment{proposition}{\refstepcounter{theorem}\medskip\noindent{\sc
Proposition \thetheorem}:}{$\Box$\medskip}

\newenvironment{definition}{\refstepcounter{definition}\medskip\noindent{\sc
Definition \thedefinition}:}{$\Box$\medskip}

\def\op#1{\mathop{{\it\fam0} #1}\limits}

\newcommand{\wh}{\widehat}

\newcommand{\ol}{\overline}
\newcommand{\ot}{\otimes}

\newcommand{\ar}{\op\longrightarrow}
\newcommand{\id}{{\mathrm{Id}\,}}

\newcommand{\hm}{{\mathrm {Hom}\,}}
\newcommand{\Ker}{\mathrm{Ker}\,}
\newcommand{\im}{{\mathrm {Im}\,}}
\newcommand{\dif}{{\mathrm{Diff}}}

\newcommand{\cP}{{\mathcal P}}
\newcommand{\cO}{{\mathcal O}}
\newcommand{\cA}{{\mathcal A}}
\newcommand{\cV}{{\mathcal V}}
\newcommand{\cT}{{\mathcal T}}
\newcommand{\cG}{{\mathcal G}}
\newcommand{\cR}{{\mathcal R}}
\newcommand{\cK}{{\mathcal K}}
\newcommand{\cZ}{{\mathcal Z}}

\newcommand{\cS}{{\mathcal S}}
\newcommand{\cI}{{\mathcal I}}
\newcommand{\cJ}{{\mathcal J}}
\newcommand{\cW}{{\mathcal W}}

\newcommand{\bp}{{\mathbf p}}

\newcommand{\gf}{{\mathfrak f}}
\newcommand{\gd}{{\mathfrak d}}
\newcommand{\gR}{{\mathfrak R}}
\newcommand{\gj}{{\mathfrak J}}
\newcommand{\gP}{{\mathfrak P}}
\newcommand{\gU}{{\mathfrak U}}
\newcommand{\gA}{{\mathfrak A}}

\newcommand{\lla}{\op\longleftarrow}
\newcommand{\bll}{\bullet}
\newcommand{\sU}{{\{U\}}}
\newcommand{\nw}[1]{[{#1}]}

\newcommand{\f}{\phi}

\newcommand{\si}{\sigma}
\newcommand{\vf}{\varphi}
\newcommand{\ve}{\varepsilon}
\newcommand{\e}{\epsilon}

\newcommand{\Om}{\Omega}
\newcommand{\vr}{\varrho}
\newcommand{\al}{\alpha}

\newcommand{\bt}{\beta}
\newcommand{\g}{\gamma}
\newcommand{\G}{\Gamma}
\newcommand{\m}{\mu}
\newcommand{\la}{\lambda}
\newcommand{\La}{\Lambda}
\newcommand{\dl}{\delta}

\newcommand{\bb}{{\mathbf 1}}
\newcommand{\dr}{\partial}
\newcommand{\w}{\wedge}
\newcommand{\bL}{{\mathbf L}}

\newcommand{\mar}[1]{}

\begin{document}

\hbox{}

\begin{center}

{\Large\bf Differential calculus over $\mathbb N$-graded
commutative rings}
\bigskip
\bigskip

{\sc G. SARDANASHVILY, W. WACHOWSKI }
\bigskip

Department of Theoretical Physics, Moscow State University, Russia

\bigskip
\bigskip

\textbf{Abstract}
\end{center}

\noindent The Chevalley--Eilenberg differential calculus and
differential operators over $\mathbb N$-graded commutative rings
are constructed. This is a straightforward generalization of the
differential calculus over commutative rings, and it is the most
general case of the differential calculus over rings that is not
the non-commutative geometry. Since any $\mathbb N$-graded ring
possesses the associated $\mathbb Z_2$-graded structure, this also
is the case of the graded differential calculus over Grassmann
algebras and the supergeometry and field theory on graded
manifolds.

\bigskip

\tableofcontents

\addcontentsline{toc}{section}{Introduction}

\section{Introduction}

This work addresses the differential calculus over $\mathbb
N$-graded commutative rings.

This is a straightforward generalization of the differential
calculus over commutative rings (Section 2), and it is the most
general case of the differential calculus over rings that is not
the non-commutative geometry (Section 9).

Throughout the work, all algebras are associative, unless they are
Lie and graded Lie algebras. By a ring is meant a unital algebra,
i.e., it contains a unit element $\bb\neq 0$. Hereafter, $\cK$
denotes a commutative ring without a divisor of zero (i.e., an
integral domain), and $\mathbb N$ is a set of natural numbers,
including 0.

\begin{definition} \label{nn90} \mar{nn90}
A direct sum of $\cK$-modules
\mar{nn91}\beq
P=\op\oplus_{i\in\mathbb N} P^i \label{nn91}
\eeq
is called the $\mathbb N$-graded $\cK$-module (the internally
graded module in the terminology of \cite{mcl}). Its element $p$
is said to be homogeneous of degree $|p|$ if $p\in P^{|p|}$.
\end{definition}

\begin{definition} \label{nn40} \mar{nn40}
A $\cK$-ring $\Om$ is said to be $\mathbb N$-graded if it is an
$\mathbb N$-graded $\cK$-module
\mar{nn41}\beq
\Om=\Om^*= \op\oplus_k \Om^k, \qquad k\in\mathbb N, \label{nn41}
\eeq
so that, for homogeneous elements $\al\in \Om^{|\al|}$ of degree
$|\al|$, their product is a homogeneous element $\al\cdot\al'\in
\Om^{|\al|+|\al'|}$ of degree $|\al|+|\al'|$. In particular, it
follows that $\Om^0$ is a $\cK$-ring, while $\Om^{k>0}$ are
$\Om^0$-bimodules and, accordingly, $\Om^*$ is well.
\end{definition}

It should be emphasized that a $\cK$-ring $\Om$ can admit
different $\mathbb N$-graded structures $\Om^*$ (\ref{nn41})
(Theorem \ref{nn205}).

For instance, any ring $\cA$ is the $\mathbb N$-graded ring
$\cA^*$ where $\cA^0=\cA$ and $\cA^{>0}=0$.

We mainly consider $\mathbb N$-graded commutative rings.

\begin{definition} \label{nn44} \mar{nn44}
An $\mathbb N$-graded ring  is said to be graded commutative if
\mar{nn92}\beq
\al\cdot\bt=(-1)^{|\al||\bt|}\bt\cdot \al, \qquad \al,\bt\in
\Om^*. \label{nn92}
\eeq
In this case, $\Om^0$ is a commutative $\cK$-ring, and $\Om^*$ is
an $\Om^0$-ring.
\end{definition}

Any commutative $\mathbb N$-graded ring $\cA^*$ can be regarded as
an even $\mathbb N$-graded commutative ring $\La^*$ such that
$\La^{2i}=\cA^i$, $\La^{2i+1}=0$. Obviously, an even $\mathbb
N$-graded commutative ring is a commutative $\cK$-ring. Therefore,
studying the differential calculus over $\mathbb N$-graded rings
(Section 6), we refer to that over commutative rings (Sections 2
-- 4)

\begin{definition} \label{nn211} \mar{nn211}
An $\mathbb N$-graded commutative $\cK$-ring $\Om^*$ is said to be
generated in degree 2 if $\Om^0=\cK$,  $\Om^1=0$, and
 $\Om^*=\cP[\Om^2]$ is a polynomial $\cK$-ring of a
$\cK$-module $\Om^2$ (Example \ref{ws40}). It is an even $\mathbb
N$-graded commutative ring $\Om^*$ where $\Om^{2i+1}=0$, $i\in
\mathbb N$.
\end{definition}

A polynomial $\cK$-ring $\cP[Q]$ of a $\cK$-module $Q$ in Example
\ref{ws40} exemplifies a commutative $\mathbb N$-graded ring. If
$\cK$ is a field and $Q$ is a free $\cK$-module of finite rank, a
polynomial $\cK$-ring $\cP[Q]$ is finitely generated in degree 2
by virtue of Definition \ref{nn210}. If $\cK$ is a field, all
$\mathbb N$-graded structures of this ring are mutually isomorphic
in accordance with Theorem \ref{nn205}.

An $\mathbb N$-graded commutative ring $\Om^*$ possesses an
associated $\mathbb Z_2$-graded commutative structure
\mar{nn104,5}\ben
&& \Om=\Om_0\oplus \Om_1, \qquad \Om_0= \op\oplus_k \Om^{2k}, \qquad
\Om_0= \op\oplus_k \Om^{2k+1}, \qquad k\in\mathbb N,
\label{nn104}\\
&& \al\cdot\bt=(-1)^{[\al][\bt]}\bt\cdot \al, \qquad \al,\bt\in
\Om_*, \label{nn105}
\een
where the symbol $[.]$ stands for the $\mathbb Z_2$-degree
(Definition \ref{nn112}).

In view of this fact, we also consider the differential calculus
over $\mathbb Z_2$-graded commutative rings (Section 5.2), but
focus our consideration on Grassmann algebras (Definition
\ref{nn114}). They are the case of $\mathbb N$-graded commutative
rings of the following type.

\begin{definition} \label{nn106} \mar{nn106}
An $\mathbb N$-graded commutative $\cK$-ring $\Om=\Om^*$ is called
the Grassmann-graded $\cK$-ring if it is finitely generated in
degree 1 (Definition \ref{nn210}), i.e., the following hold:

$\bullet$ $\Om^0=\cK$,

$\bullet$ $\Om^1$ is a free $\cK$-module of finite rank,

$\bullet$ $\Om^*$ is generated by $\Om^1$, namely, if $R$ is an
ideal generated by $\Om^1$, then there are $\cK$-module
isomorphism $\Om/R=\cK$, $R/R^2=\Om^1$.
\end{definition}

Let us note that, a Grassmann-graded $\cK$-ring $\Om^*$ seen as a
$\mathbb Z_2$-graded commutative ring $\Om$ can admit different
Grassmann-graded structures $\Om^*$ and $\Om'^*$. However, since
it is finitely generated in degree 1 (Definition \ref{nn210}), all
these structures mutually are isomorphic in accordance with
Theorem \ref{nn205} if $\cK$ is a field.

We follow the conventional technique of the differential calculus
over commutative rings, including formalism of differential
operators on modules over commutative rings, the
Chevalley--Eilenberg differential calculus over commutative rings,
and the apparat of connections on modules and rings (Section 2)
\cite{book,grot,kras,book12}. In particular, this is a case of the
conventional differential calculus on smooth manifolds (Section
4).

One can generalize the Chevalley--Eilenberg differential calculus
to a case of an arbitrary ring (Section 9)
\cite{dub01,book05,book12}. However, an extension of the notion of
a differential operator to modules over a non-commutative ring
meets difficulties \cite{book05,book12}. A key point is that a
multiplication in a non-commutative ring is not a zero-order
differential operator.

One overcomes this difficulty in a case of $\mathbb Z_2$-graded
and $\mathbb N$-graded commutative rings by means of a
reformulation of the notion of differential operators (Remark
\ref{nn121}). As a result, the differential calculus technique has
been extended to $\mathbb Z_2$-graded commutative rings (Section
5.2) \cite{book05,book12,ijgmmp13}.

Since any $\mathbb N$-graded commutative ring $\cA^*$ possesses a
structure of a $\mathbb Z_2$-graded commutative ring $\cA$, and
commutation relations  (\ref{nn92}) of its elements depend of
their $\mathbb Z_2$-gradation degree, but not the $\mathbb
N$-gradation one, the differential calculus on $\mathbb N$-graded
modules over $\mathbb N$-graded commutative rings is defined just
as that over $\mathbb Z_2$-graded commutative rings (Section 5.2).

However, it should be emphasized that an $\mathbb N$-graded
differential operator is an $\mathbb N$-graded $\cK$-module
homomorphism which obeys the conditions (\ref{nn300}), i.e., it is
a sum of homogeneous morphisms of fixed $\mathbb N$-degrees, but
not the $\mathbb Z_2$ ones. Therefore, any $\mathbb N$-graded
differential operator also is a $\mathbb Z_2$-graded differential
operator, but the converse might not be true (Remark \ref{nn301}).

Let us note that the differential calculus over $\mathbb
Z_2$-graded commutative rings, namely, Grassmann algebras
(Definition \ref{nn114}) provides the mathematical formulation of
Lagrangian formalism of even and odd variables on $\mathbb
Z_2$-graded manifolds and bundles
\cite{lmp08,book09,ijgmmp13,book16}.

There are different approaches to formulating graded manifolds
\cite{bart,bru,far,roy,vit}. We follow their definition in terms
of local-ringed space (Section 3.1) and consider local-ringed
spaces whose stalks are Grassmann-graded rings (Section 6). Since
Grassmann-graded rings also are Grassmann algebras, we follow
formalism of $\mathbb Z_2$-graded manifolds in Section 5.3.

Let $Z$ be an $n$-dimensional real smooth manifold. Let $\cA$ be a
real Grassmann-graded ring. By virtue of Theorem \ref{nn117}, it
is isomorphic to the exterior algebra $\w W$ of a real vector
space $W=\cA^1$. Therefore, we come to Definition \ref{nn306} of
an $\mathbb N$-graded manifold as a simple $\mathbb Z_2$-graded
manifold $(Z,\gA_E)$ modelled over some vector bundle $E\to Z$
(Definition \ref{nn173}). In Section 5.3 on $\mathbb Z_2$-graded
manifolds, we have restricted our consideration to simple graded
manifolds, and this Section, in fact, presents formalism of
$\mathbb N$-graded manifolds (Remark \ref{nn301}).

\section{Differential calculus over commutative rings}

Any commutative ring can be seen as an even graded commutative
ring whose components $\Om^{k>0}$ are empty. Conversely, every
even graded commutative ring is commutative.

Therefore we start with the differential calculus over commutative
rings. As was mentioned above, this is a case of the conventional
differential calculus on smooth manifolds (Section 4). This
technique fails to be extended straightforwardly to
non-commutative rings (Section 9), unless this is a case of graded
commutative rings (Sections 5 -- 6).

\subsection{Commutative algebra}

In this Section, the relevant basics on modules over commutative
rings are summarized  \cite{lang,lomb,mcl,massey}.

An algebra $\cA$  is defined to be an additive group which
additionally is provided with distributive multiplication. As was
mentioned above, all algebras throughout are associative, unless
they are Lie and graded Lie algebras. By a ring is meant a  unital
algebra, i.e., it contains a unit element $\bb\neq 0$.
Non-zero elements of a ring $\cA$ form the multiplicative monoid $M\cA\subset \cA$ of $\cA$. If it is a
group, $\cA$  is called the division ring. A field is a
commutative division ring.

A ring $\cA$ is said to be the domain (the integral domain or the
entire ring in commutative algebra) if it has no a divisor of
zero, i.e., $ab=0$, $a,b\in \cA$, implies either $a=0$ or $b=0$.
For instance, a division ring is a domain, and a field is an
integral domain. A polynomial $\cA$-ring over an integral domain
$\cA$ (Example \ref{ws40}) is an integral domain.

A subset $\cI$ of an algebra $\cA$ is said to be the left (resp.
right) ideal if it is a subgroup of an additive group $\cA$ and
$ab\in \cI$ (resp. $ba\in\cI$) for all $a\in \cA$, $b\in \cI$. If
$\cI$ is both a left and right ideal, it is called the two-sided
ideal. For instance, any ideal of a commutative algebra is
two-sided. An ideal is a subalgebra, but a proper ideal (i.e.,
$\cI\neq \cA$) of a ring is not a subring. A proper ideal of a
ring is said to be  maximal if it does not belong to another
proper ideal. A proper ideal $\cI$ of a ring $\cA$ is called completely prime
(prime in commutative algebra) if $ab\in \cI$ implies either $a\in\cI$ or $b\in\cI$.
Any maximal ideal of a commutative ring is prime.

Given a two-sided ideal $\cI\subset \cA$, an additive factor group
$\cA/\cI$ is an algebra, called the  factor algebra. If $\cA$ is a
ring, then $\cA/\cI$ is so.  If $\cA$ is a commutative ring and $\cI$ is its prime ideal, the factor
ring $\cA/\cI$ is an entire ring, and it is a field
if $\cI$ is a maximal ideal.

\begin{definition} \label{nn251} \mar{nn251}
A ring $\cA$ is called local if it has a unique maximal two-sided
ideal. This ideal consists of all non-invertible elements of
$\cA$.
\end{definition}

\begin{remark} \label{nn245} \mar{nn245}
Local rings conventionally are defined in commutative algebra
\cite{lang,lomb}. This notion has been extended to
$\mathbb Z_2$-graded commutative rings too \cite{bart}. Any division ring, in particular, a field is local.
Its unique maximal ideal contains only zero.
Grassmann-graded rings in Definition \ref{nn106} and Grassmann
algebras in Definition \ref{nn114} are local.
\end{remark}

\begin{remark} \label{ws90} \mar{ws90}
One can associate a local ring to any commutative ring $\cA$ as
follows. Let $S\subset M\cA$ be a multiplicative subset
of $\cA$, i.e., a submonoid of the multiplicative monoid $M\cA$ of $\cA$. Let
us say that two pairs $(a,s)$ and $(a',s')$, $a,a'\in\cA$,
$s,s'\in S$, are equivalent if there exists an element $s''\in S$
such that
\be
s''(s'a-sa')=0.
\ee
We abbreviate with $a/s$ the equivalence classes of $(a,s)$. A set
$S^{-1}\cA$ of these equivalence classes is a local commutative
ring with respect to operations
\be
s/a+ s'/a'=(s'a+sa')/(ss'), \qquad (a/s)\cdot (a'/s')=(aa')/(ss').
\ee
There is a homomorphism
\mar{ws91}\beq
\Phi_S: \cA\ni \mapsto a/\bb\in S^{-1}\cA \label{ws91}
\eeq
such that any element of $\Phi_S(S)$ is invertible in $S^{-1}\cA$.
If a ring $\cA$ has no divisor of zero and $S$ does not contain a
zero element, then $\Phi_S$ (\ref{ws91}) is a monomorphism. In
particular, if $S=M\cA$, the ring $S^{-1}\cA$ is a field, called the field of
quotients or the fraction field of $\cA$. If $\cA$ is a field, its
fraction field coincides with $\cK$.
\end{remark}

Given an algebra $\cA$, an additive group $P$ is said to be the
left (resp. right)  $\cA$-module if it is provided with
distributive multiplication $\cA\times P\to P$ by elements of
$\cA$ such that $(ab)p=a(bp)$ (resp. $(ab)p=b(ap)$) for all
$a,b\in\cA$ and $p\in P$.  If $\cA$ is a ring, one additionally
assumes that $\bb p=p=p\bb$ for all $p\in P$. Left and right
module structures are usually written by means of left and right
multiplications $(a,p)\to ap$ and $(a,p)\to pa$, respectively. If
$P$ is both a left module over an algebra $\cA$ and a right module
over an algebra $\cA'$, it is called the $(\cA-\cA')$-bimodule
(the $\cA$-bimodule if $\cA=\cA'$).  If $\cA$ is a commutative
algebra, an $\cA$-bimodule $P$ is said to be  commutative if
$ap=pa$ for all $a\in \cA$ and $p\in P$. Any left or right module
over a commutative algebra $\cA$ can be brought into a commutative
bimodule. Therefore, unless otherwise stated (Section 2.2), any
$\cA$-module over a commutative algebra is a commutative
$\cA$-bimodule, which is called the $\cA$-module if there is no
danger of confusion.

A module over a field is called the  vector space. If an algebra
$\cA$ is a commutative bimodule over a ring $\cK$ (i.e., a
commutative $\cK$-bimodule), it is said to be the $\cK$-algebra.
Any algebra can be seen as a $\mathbb Z$-algebra.

Hereafter, by $\cA$ in this Section is meant a commutative ring.

The following are standard constructions of new modules from the
old ones.

$\bullet$ A direct sum $P_1\oplus P_2$ of $\cA$-modules $P_1$ and
$P_2$ is an additive group $P_1\times P_2$ provided with an
$\cA$-module structure
\be
a(p_1,p_2)=(ap_1,ap_2), \qquad p_{1,2}\in P_{1,2}, \qquad a\in\cA.
\ee
Let $\{P_i\}_{i\in I}$ be a set of $\cA$-modules. Their direct sum
$\oplus P_i$ consists of elements $(\ldots, p_i,\ldots)$ of the
Cartesian product $\prod P_i$ such that $p_i\neq 0$ at most for a
finite number of indices $i\in I$.

$\bullet$ A  tensor product $P\ot Q$ of $\cA$-modules $P$ and $Q$
is an additive group which is generated by elements $p\ot q$,
$p\in P$, $q\in Q$, obeying relations
\be
&& (p+p')\ot q =p\ot q + p'\ot q, \quad p\ot(q+q')=p\ot q+p\ot q', \\
&&  pa\ot q= p\ot aq, \qquad p\in P, \qquad q\in Q, \qquad
a\in\cA.
\ee
It is provided with an $\cA$-module structure
\be
a(p\ot q)=(ap)\ot q=p\ot (qa)=(p\ot q)a.
\ee
If a ring $\cA$ is treated as an $\cA$-module, a tensor product
$\cA\ot_\cA Q$ is canonically isomorphic to $Q$ via the assignment
\be
\cA\ot_\cA Q\ni a\ot q \leftrightarrow aq\in Q.
\ee

\begin{example} \label{ws40} \mar{ws40} Let $Q$ be an $\cA$-module.
We denote $Q^{\ot k}=\op\ot^kQ$. Let us consider an $\mathbb
N$-graded module
\mar{spr620}\beq
\ot Q=\cA\oplus Q\oplus\cdots\oplus Q^{\ot k}\oplus\cdots.
\label{spr620}
\eeq
It is an $\mathbb N$-graded $\cA$-algebra with respect to a tensor
product $\ot$. It is called the tensor algebra of an $\cA$-module
$Q$. Its quotient $\w Q$ with respect to an ideal generated by
elements $q\ot q'+q'\ot q$, $q,q'\in Q$, is an $\mathbb N$-graded
commutative algebra, called the exterior algebra of an
$\cA$-module $Q$. The quotient $\cP[Q]=\vee Q$ of $\ot Q$
(\ref{spr620}) with respect to an ideal generated by elements
$pq\ot q'-q'\ot q$, $q,q'\in Q$, is called the polynomial
$\cA$-ring of an $\cA$-module $Q$. This is an even $\mathbb
N$-graded commutative ring.
\end{example}

$\bullet$ Given a submodule $Q$ of an $\cA$-module $P$, the
quotient $P/Q$ of an additive group $P$ with respect to its
subgroup $Q$ also is provided with an $\cA$-module structure. It
is called the  factor module.

$\bullet$ A set $\hm_\cA(P,Q)$ of $\cA$-linear morphisms of an
$\cA$-module $P$ to an $\cA$-module $Q$ is naturally an
$\cA$-module. An $\cA$-module $P^*=\hm_\cA(P,\cA)$ is called the
dual of an $\cA$-module $P$. There is a natural monomorphism $P\to
P^{**}$.

An $\cA$-module $P$ is called free if it admits a basis, i.e., a
linearly independent subset $I\subset P$ spanning $P$ such that
each element of $P$ has a unique expression as a linear
combination of elements of $I$ with a finite number of non-zero
coefficients from a ring $\cA$.
Any module over a division ring, e.g., a vector space is free. Every
module is isomorphic to a quotient of a free module. A module is
said to be finitely generated (or of finite rank) if it is a
quotient of a free module with a finite basis.

One says that a module $P$ is  projective if it is a direct
summand of a free module, i.e., there exists a module $Q$ such
that $P\oplus Q$ is a free module. A module $P$ is projective iff
$P=\bp S$ where $S$ is a free module and $\bp$ is a projector of
$S$, i.e., $\bp^2=\bp$.

\begin{theorem} \label{nn228} \mar{nn228}
If $P$ is a projective module of finite rank, then its dual $P^*$
is so, and $P^{**}$ is isomorphic to $P$.
\end{theorem}

\begin{theorem} \label{nn1} \mar{nn1}
Any projective module over a local commutative ring is free.
\end{theorem}

The forthcoming constructions are extended to a case of modules
over graded commutative rings (Remark \ref{nn3}).

A composition of module morphisms
\be
P\ar^i Q\ar^j T
\ee
is said to be exact  at $Q$ if $\Ker j=\im i$. A composition of
module morphisms
\mar{spr13}\beq
0\to P\ar^i Q\ar^j T\to 0 \label{spr13}
\eeq
is called the short exact sequence if it is exact at all the terms
$P$, $Q$, and $T$. This condition implies that: (i) $i$ is a
monomorphism,  (ii) $\Ker j=\im i$, and (iii) $j$ is an
epimorphism onto a factor module $T=Q/P$.

\begin{theorem} \label{spr183} \mar{spr183}
Given an exact sequence of modules (\ref{spr13}) and another
$\cA$-module $R$, the sequence of modules
\be
0\to\hm_\cA(T,R)\ar^{j^*} \hm_\cA(Q,R)\ar^{i^*} \hm(P,R)
\ee
is exact at the first and second terms, i.e., $j^*$ is a
monomorphism, but $i^*$ need not be an epimorphism.
\end{theorem}

One says that the exact sequence (\ref{spr13}) is  split
 if there exists a
monomorphism $s:T\to Q$ such that $j\circ s=\id T$ or,
equivalently,
\be
Q=i(P)\oplus s(T) \cong P\oplus T.
\ee

\begin{theorem} \label{nn233} \mar{nn233}
The exact sequence (\ref{spr13}) always is split if $T$ is a
projective module.
\end{theorem}

\begin{remark} \label{nn2} \mar{nn2}
A  directed set $I$ is a set with an order relation $<$ which
satisfies the following three conditions:

(i) $i<i$, for all $i\in I$;

(ii) if $i<j$ and $j< k$, then $i<k$;

(iii) for any $i,j\in I$, there exists $k\in I$ such that $i<k$
and $j<k$.

\noindent It may happen that $i\neq j$, but $i<j$ and $j<i$
simultaneously.
\end{remark}

A family of $\cA$-modules $\{P_i\}_{i\in I}$, indexed by a
directed set $I$, is called the  direct system
 if, for any pair $i<j$, there
exists a morphism $r^i_j:P_i\to P_j$ such that
\be
r^i_i=\id P_i, \qquad r^i_j\circ r^j_k=r^i_k, \qquad i<j<k.
\ee
A direct system of modules admits a  direct limit.
 This is an $\cA$-module $P_\infty$ together with
morphisms $r^i_\infty: P_i\to P_\infty$ such that
$r^i_\infty=r^j_\infty\circ r^i_j$ for all $i<j$. A module
$P_\infty$ consists of elements of a direct sum $\oplus P_i$
modulo the identification of elements of $P_i$ with their images
in $P_j$ for all $i<j$. In particular, a direct system
\be
P_0\ar P_1\cdots\ar P_i\ar^{r^i_{i+1}}\cdots, \qquad I=\mathbb N,
\ee
indexed by $\mathbb N$, is called the direct sequence.

\begin{remark} \label{nn3} \mar{nn3} It should be noted that direct
limits also exist in the categories of commutative and graded
commutative algebras and rings, but not in categories whose
objects are non-commutative groups.
\end{remark}

A morphism of a direct system $\{P_i, r^i_j\}_I$ to a direct
system $\{Q_{i'}, \rho^{i'}_{j'}\}_{I'}$ consists of an order
preserving map $f:I\to I'$ and $\cA$-module morphisms $F_i:P_i\to
Q_{f(i)}$ which obey compatibility conditions
\mar{nn13}\beq
\rho^{f(i)}_{f(j)}\circ F_i=F_j\circ r^i_j. \label{nn13}
\eeq
If $P_\infty$ and $Q_\infty$ are direct limits of these direct
systems, there exists a unique $\cA$-module morphism $F_\infty:
P_\infty\to Q_\infty$ such that
\be
\rho^{f(i)}_\infty\circ F_i=F_\infty\circ r^i_\infty, \qquad i\in
I.
\ee

\begin{proposition} \label{nn17} \mar{nn17}
A construction of a direct limit morphism preserve monomorphisms,
epimorphisms and, consequently, isomorphisms. Namely, if all
$F_i:P_i\to Q_{f(i)}$ are monomorphisms (resp. epimorphisms and
isomorphisms), so is $F_\infty:P_\infty\to Q_\infty$.
\end{proposition}

\begin{example} \label{nn14} \mar{nn14}
In particular, let $\{P_i, r^i_j\}_I$ be a direct system of
$\cA$-modules and $Q$ an $\cA$-module together with $\cA$-module
morphisms $F_i:P_i\to Q$ which obey the compatibility conditions
$F_i=F_j\circ r^i_j$ (\ref{nn13}). Then there exists an
$\cA$-module morphism $F_\infty: P_\infty\to Q$ such that
$F_i=F_\infty\circ r^i_\infty$ for any $i\in I$. If all $F_i$ are
monomorphisms or epimorphisms, so is $F_\infty$.
\end{example}

\begin{theorem} \label{spr170} \mar{spr170}
Direct limits commute with direct sums and tensor products of
modules. Namely, let $\{P_i\}$ and $\{Q_i\}$ be two direct systems
of $\cA$-modules which are indexed by the same directed set $I$,
and let $P_\infty$ and $Q_\infty$ be their direct limits. Then
direct limits of direct systems $\{P_i\oplus Q_i\}$ and $\{P_i\ot
Q_i\}$ are $P_\infty\oplus Q_\infty$ and $P_\infty\ot Q_\infty$,
respectively.
\end{theorem}

\begin{theorem} \label{dlim1} \mar{dlim1}
Let short exact sequences
\mar{spr186}\beq
0\to P_i\ar^{F_i} Q_i\ar^{\Phi_i} T_i\to 0 \label{spr186}
\eeq
for all $i\in I$ define a short exact sequence of direct systems
of modules $\{P_i\}_I$, $\{Q_i\}_I$, and $\{T_i\}_I$ which are
indexed by the same directed set $I$. Then there exists a short
exact sequence of their direct limits
\mar{spr187}\beq
0\to P_\infty\ar^{F_\infty} Q_\infty\ar^{\Phi_\infty} T_\infty\to
0. \label{spr187}
\eeq
\end{theorem}

In particular, a direct limit of factor modules $Q_i/P_i$ is a
factor module $Q_\infty/P_\infty$. By virtue of Theorem
\ref{spr170}, if all the exact sequences (\ref{spr186}) are split,
the exact sequence (\ref{spr187}) is well.

In a case of inverse systems of modules, we restrict our
consideration to inverse sequences
\be
P^0\lla P^1\lla \cdots P^i\op\lla^{\pi^{i+1}_i}\cdots.
\ee
Its  inverse limit  is a module $P^\infty$ together with morphisms
$\pi^\infty_i: P^\infty\to P^i$ so that $\pi^\infty_i=\pi^j_i\circ
\pi^\infty_j$ for all $i<j$. It consists of elements
$(\ldots,p^i,\ldots)$, $p^i\in P^i$, of the Cartesian product
$\prod P^i$ such that $p^i=\pi^j_i(p^j)$ for all $i<j$.

A morphism of an inverse system $\{P_i, \pi^i_j\}$ to an inverse
system $\{Q_i, \vr^i_j\}$ consists of $\cA$-module morphisms
$F_i:P_i\to Q_i$ which obey compatibility conditions
\mar{nn18}\beq
F_j\circ\pi^i_j=\vr^i_j\circ F_i. \label{nn18}
\eeq
If $P_\infty$ and $Q_\infty$ are inverse  limits of these inverse
systems, there exists a unique $\cA$-module morphism $F_\infty:
P_\infty\to Q_\infty$ such that
\be
F_j\circ\pi^\infty_j=\vr^\infty_j\circ F_\infty.
\ee

\begin{proposition} \label{nn19} \mar{nn19}
Inverse limits preserve monomorphisms, but not epimorphisms.
\end{proposition}

\begin{example} \label{nn21} \mar{nn21}
In particular, let $\{P_i, \pi^i_j\}$ be an inverse system of
$\cA$-modules and $Q$ an $\cA$-module together with $\cA$-module
morphisms $F_i:Q\to P_i$ which obey compatibility conditions
$F_j=\pi^i_j\circ F_i$. Then there exists  a unique morphism
$F_\infty: Q\to P_\infty$ such that $F_j=\pi^\infty_j\circ
F_\infty$.
\end{example}

\begin{example} \label{nn22} \mar{nn22}
Let $\{P_i, \pi^i_j\}$ be an inverse system of $\cA$-modules and
$Q$ an $\cA$-module. Given a term $P_r$, let $\Phi_r:P_r\to Q$ be
an $\cA$-module morphism. It yields the pull-back morphisms
\mar{nn25}\beq
\pi^{r+k}_r{}^*\Phi_r=\Phi_r\circ \pi^{r+k}_r:P_{r+k}\to Q
\label{nn25}
\eeq
which obviously obey the compatibility conditions (\ref{nn18}).
Then there exists a unique morphism $\Phi_\infty: P_\infty\to Q$
such that $\Phi_\infty=\Phi_r\circ \pi^\infty_r$.
\end{example}

\begin{theorem} \label{spr3} \mar{spr3}
If a sequence
\be
0\to P^i\ar^{F^i} Q^i\ar^{\Phi^i} T^i, \qquad i\in\mathbb N,
\ee
of inverse systems of $\cA$-modules $\{P^i\}$, $\{Q^i\}$ and
$\{T^i\}$ is exact, so is a sequence of inverse limits
\be
0\to P^\infty\ar^{F^\infty} Q^\infty\ar^{\Phi^\infty} T^\infty.
\ee
\end{theorem}

In contrast with direct limits (Remark \ref{nn3}), the inverse
ones exist in the category of groups which need not be
commutative.

\begin{example} \label{tt1} \mar{tt1}
Let $\{P_i\}$ be a direct sequence of $\cA$-modules. Given an
$\cA$-module $Q$, modules $\hm_\cA(P_i,Q)$ constitute an inverse
sequence such that its inverse limit is isomorphic to
$\hm_\cA(P_\infty,Q)$.
\end{example}

\begin{example} \label{nn23} \mar{nn23}
Let $\{P_i\}$ be an inverse sequence of $\cA$-modules. Given an
$\cA$-module $Q$, modules $\hm_\cA(P_i,Q)$ constitute a direct
sequence such that its direct limit is isomorphic to
$\hm_\cA(P_\infty,Q)$.
\end{example}

\subsection{Differential operators on modules and rings}

This Section addresses the notion of (linear) differential
operators on modules over commutative rings
\cite{book,grot,kras,book12}.

As was mentioned above, $\cK$ throughout is a commutative ring
without a divisor of zero. Let $\cA$ be a commutative $\cK$-ring,
and let $P$ and $Q$ be $\cA$-modules. A $\cK$-module $\hm_\cK
(P,Q)$ of $\cK$-module homomorphisms $\Phi:P\to Q$ can be endowed
with two different $\cA$-module structures
\mar{5.29}\beq
(a\Phi)(p)= a\Phi(p),  \qquad  (\Phi\bll a)(p) = \Phi (a p),\qquad
a\in \cA, \quad p\in P. \label{5.29}
\eeq
For the sake of convenience, we will refer to the second one as an
$\cA^\bll$-module structure. Let us put
\mar{spr172}\beq
\dl_a\Phi= a\Phi -\Phi\bll a, \qquad a\in\cA. \label{spr172}
\eeq

\begin{definition} \label{ws131} \mar{ws131}
An element $\Delta\in\hm_\cK(P,Q)$ is called the $s$-order
$Q$-valued
 differential operator on $P$ if
\be
(\dl_{a_0}\circ\cdots\circ\dl_{a_s})\Delta=0
\ee
for any tuple of $s+1$ elements $a_0,\ldots,a_s$ of $\cA$.
\end{definition}

A set $\dif_s(P,Q)$ of these operators inherits the $\cA$- and
$\cA^\bll$-module structures (\ref{5.29}). Of course, an $s$-order
differential operator also is of $(s+1)$-order.

In particular, zero-order differential operators obey a condition
\be
\dl_a \Delta(p)=a\Delta(p)-\Delta(ap)=0, \qquad a\in\cA, \qquad
p\in P,
\ee
and, consequently, they coincide with $\cA$-module morphisms $P\to
Q$. A first-order differential operator $\Delta$ satisfies a
condition
\mar{ws106}\beq
(\dl_b\circ\dl_a)\Delta(p)= ba\Delta(p) -b\Delta(ap)
-a\Delta(bp)+\Delta(abp) =0, \quad a,b\in\cA. \label{ws106}
\eeq

The following fact reduces the study of $Q$-valued differential
operators on an $\cA$-module $P$ to that of $Q$-valued
differential operators on a ring $\cA$ seen as an $\cA$-module.

\begin{proposition} \label{ws109} \mar{ws109}
Let us consider an $\cA$-module morphism
\mar{n2}\beq
h_s: \dif_s(\cA,Q)\to Q, \qquad h_s(\Delta)=\Delta(\bb).
\label{n2}
\eeq
Any $s$-order $Q$-valued differential operator $\Delta\in
\dif_s(P,Q)$ on $P$ uniquely factorizes as
\be
\Delta:P\ar^{\gf_\Delta} \dif_s(\cA,Q)\ar^{h_s} Q
\ee
through the morphism $h_s$ (\ref{n2}) and some homomorphism
\be
\gf_\Delta: P\to \dif_s(\cA,Q), \qquad (\gf_\Delta
p)(a)=\Delta(ap), \qquad a\in \cA,
\ee
of an $\cA$-module $P$ to an $\cA^\bll$-module $\dif_s(\cA,Q)$
\cite{kras}. The assignment $\Delta\to\gf_\Delta$ defines an
isomorphism
\be
\dif_s(P,Q)=\hm_{\cA-\cA^\bll}(P,\dif_s(\cA,Q)).
\ee
\end{proposition}

Therefore, let $P=\cA$. Any zero-order $Q$-valued differential
operator $\Delta$ on $\cA$ is defined by its value $\Delta(\bb)$.
Then there is an $\cA$-module isomorphism $\dif_0(\cA,Q)=Q$ via
the association
\be
Q\ni q\to \Delta_q\in \dif_0(\cA,Q),
\ee
where $\Delta_q$ is given by an equality $\Delta_q(\bb)=q$.

A first-order $Q$-valued differential operator $\Delta$ on $\cA$
fulfils a condition
\be
\Delta(ab)=b\Delta(a)+ a\Delta(b) -ba \Delta(\bb), \qquad
a,b\in\cA.
\ee

\begin{definition} \label{nn95} \mar{nn95}
It is called the $Q$-valued  derivation of $\cA$ if
$\Delta(\bb)=0$, i.e., the  Leibniz rule
\mar{+a20}\beq
\Delta(ab) = \Delta(a)b + a\Delta(b), \qquad  a,b\in \cA,
\label{+a20}
\eeq
holds.
\end{definition}

One obtains at once that any first-order differential operator on
$\cA$ falls into a sum
\be
\Delta(a)= a\Delta(\bb) +[\Delta(a)-a\Delta(\bb)]
\ee
of a zero-order differential operator $a\Delta(\bb)$ and a
derivation $\Delta(a)-a\Delta(\bb)$. If $\dr$ is a derivation of
$\cA$, then $a\dr$ is well for any $a\in \cA$. Hence, derivations
of $\cA$ constitute an $\cA$-module $\gd(\cA,Q)$, called the
derivation module.  There is an $\cA$-module decomposition
\mar{spr156'}\beq
\dif_1(\cA,Q) = Q \oplus\gd(\cA,Q). \label{spr156'}
\eeq

If $Q=\cA$, the derivation module $\gd\cA=\gd\cA(\cA,\cA)$ of
$\cA$ also is a Lie algebra over a ring $\cK$ with respect to a
Lie bracket
\mar{nn60}\beq
[u,u']=u\circ u'-u'\circ u, \qquad u,u'\in \gd\cA. \label{nn60}
\eeq
Accordingly, the decomposition (\ref{spr156'}) takes a form
\mar{spr156}\beq
\dif_1(\cA) = \cA \oplus\gd\cA. \label{spr156}
\eeq

Since, an $s$-order differential operator also is of
$(s+1)$-order, we have a direct sequence
\mar{nn10}\beq
\dif_0(P,Q)\ar^\mathrm{in} \dif_1(P,Q)\cdots
\ar^\mathrm{in}\dif_r(P,Q)\ar\cdots \label{nn10}
\eeq
of $Q$-valued differential operators on an $\cA$-module $P$. Its
direct limit is an $\cA-\cA^\bullet$-module $\dif_\infty(P,Q)$ of
all $Q$-valued differential operators on $P$.

\begin{example} \label{nn221} \mar{nn221}
Let $Q$ be a free $\cK$-module of finite rank and $\cP[Q]$ a
polynomial ring of $Q$ (Example \ref{ws40}). Any differential
operator on $\cP[Q]$ is a composition of derivations. Every
derivation of $\cP[Q]$ is defined by its action in $Q$. Let
$\{q^i\}$ be a basis for $Q$. Let us consider derivations
\mar{nn222}\beq
\dr_i(q^j)=\dl^j_i, \qquad \dr_i\circ\dr_j=\dr_j\circ\dr_i.
\label{nn222}
\eeq
Then any derivation of $\cP[Q]$ takes a form
\mar{nn223}\beq
u=u^i\dr_i, \qquad u_i\in \cP[Q]. \label{nn223}
\eeq
Derivations (\ref{nn223}) constitute a free $\cP[Q]$-module
$\gd\cP[Q]$ of finite rank. It is also a Lie algebra over $\cK$
with respect to the Lie bracket (\ref{nn60}).
\end{example}

\subsection{Jets of modules and rings}

An $s$-order differential operator on an $\cA$-module $P$ is
represented by a zero-order differential operator on a module of
$s$-order jets of $P$ (Theorems \ref{t6} and \ref{nn31}). We also
use modules of jets in order to define differential forms over a
ring (Remark \ref{nn244}) and connections on modules and rings
(Section 2.5), and ringed spaces (Section 3.1). Afterwards, we
however can leave jets of modules. Firstly, the cochain complex
(\ref{55.63}) of differential forms over a $\cK$-ring $\cA$
coincides with the minimal Chevalley--Eilenberg differential
calculus (\ref{t10}) over $\cA$ (Theorem \ref{nn241}). Secondly,
Definition \ref{+176} of connections on an $\cA$-module $P$ leads
to an equivalent to Definition \ref{+181} which does not involve
jets. Thirdly, jets of projective modules of finite rank over a
ring $C^\infty(X)$ of smooth real functions on a manifold $X$ are
jets of sections of vector bundles over $X$ (corollary (iv) of
Theorem \ref{sp60}).

Given an $\cA$-module $P$, let us consider a tensor product
$\cA\otimes_\cK P$ of $\cK$-modules $\cA$ and $P$. We put
\mar{spr173}\beq
\dl^b(a\otimes p)= (ba)\otimes p - a\otimes (b p), \qquad p\in P,
\qquad a,b\in\cA.  \label{spr173}
\eeq
Let us denote by $\m^{k+1}$ the submodule of $\cA\ot_\cK P$
generated by elements of the type
\be
\dl^{b_0}\circ \cdots \circ\dl^{b_k}(a\otimes p)=a\dl^{b_0}\circ
\cdots \circ\dl^{b_k}(\bb\otimes p).
\ee

\begin{definition} \label{nn242} \mar{nn242}
A  $k$-order jet module $\cJ^k(P)$  of a module $P$ is defined as
the quotient of a $\cK$-module $\cA\otimes_\cK P$ by $\m^{k+1}$.
We denote its elements $c\ot_kp$.
\end{definition}

In particular, a first-order jet module $\cJ^1(P)$ consists of
elements $c\ot_1 p$ modulo the relations
\mar{mos041}\beq
\dl^a\circ \dl^b(\bb\ot_1 p)= ab\otimes_1 p -b\otimes_1 (ap)
-a\otimes_1 (bp) +\bb\ot_1(abp) =0. \label{mos041}
\eeq

A $\cK$-module $\cJ^k(P)$ is endowed with the $\cA$- and
$\cA^\bll$-module structures
\mar{+a21}\beq
b(a\ot_k p)= ba\ot_k p, \qquad b\bll(a\otimes_k p)= a\otimes_k
(bp). \label{+a21}
\eeq
There exists a module morphism
\mar{5.44}\beq
J^k: P\ni p\to \bb\otimes_k p\in \cJ^k(P) \label{5.44}
\eeq
of an $\cA$-module $P$ to an $\cA^\bll$-module $\cJ^k(P)$ such
that $\cJ^k(P)$, seen as an $\cA$-module, is generated by elements
$J^kp$, $p\in P$.

The above mentioned relation between differential operators on
modules and jets of modules is stated by the following theorem
\cite{book,kras}.

\begin{theorem} \label{t6} \mar{t6}
Any $k$-order $Q$-valued differential operator $\Delta$ on an
$\cA$-module $P$ uniquely factorizes as
\mar{nn35}\beq
\Delta: P\ar^{J^k} \cJ^k(P)\ar^{\mathfrak{f}^\Delta} Q
\label{nn35}
\eeq
through the morphism $J^k$ (\ref{5.44}) and some $\cA$-module
homomorphism $\mathfrak{f}^\Delta: \cJ^k(P)\to Q$. The
correspondence $\Delta\to \mathfrak{f}^\Delta$ defines an
$\cA$-module isomorphism
\mar{5.50}\beq
\dif_k(P,Q)=\hm_{\cA}(\cJ^k(P),Q). \label{5.50}
\eeq
\end{theorem}

Due to natural monomorphisms $\m^r\to \m^s$ for all $r>s$, there
are $\cA$-module epimorphisms of jet modules
\mar{t4}\beq
\pi^{i+1}_i: \cJ^{i+1}(P)\to \cJ^i(P). \label{t4}
\eeq
In particular,
\mar{+a13}\beq
\pi^1_0:\cJ^1(P) \ni a\ot_1 p\to ap \in P.\label{+a13}
\eeq
Thus, there is an inverse sequence
\mar{nn8}\beq
P \lla^{\pi^1_0} \cJ^1(P)
\cdots\lla^{\pi^r_{r-1}}\cJ^r(P)\lla\cdots \label{nn8}
\eeq
of jet modules. Its inverse limit $\cJ^\infty(P)$ is an
$\cA$-module together with $\cA$-module morphisms
\be
\pi^\infty_r: \cJ^\infty(P)\to \cJ^r(P), \qquad
\pi^\infty_{r<s}=\pi^s_r\circ \pi^\infty_s.
\ee

In particular, let us consider a module $P$ together with the
morphisms $J^r$ (\ref{5.44}) which obey compatibility conditions
$J^r(p)=\pi^{r+k}_r\circ J^{r+k}(p)$, $p\in P$. Then it follows
from Example \ref{nn21} that there exists an $\cA$-module morphism
\mar{nn30}\beq
J^\infty: P\ni p\to (p,J^1p,\ldots,J^rp,\ldots)\in \cJ^\infty(P)
\label{nn30}
\eeq
so that $J^r(p)=\pi^\infty_r\circ J^\infty(p)$.

The inverse sequence (\ref{nn8}) yields a direct sequence
\mar{nn11}\beq
\hm_\cA(P,Q)\ar^{\pi^1_0{}^*}
\hm_\cA(\cJ^1(P),Q)\cdots\ar^{\pi^r_{r-1}{}^*}
\hm_\cA(\cJ^r(P),Q)\cdots,  \label{nn11}
\eeq
where
\be
\pi^r_{r-1}{}^*: \hm_\cA(\cJ^{r-1}(P),Q)\ni\Delta\to
\Delta\circ\pi^r_{r-1} \in \hm_\cA(\cJ^r(P),Q)
\ee
is the pull-back $\cA$-module morphism (\ref{nn25}). Its direct
limit is an $\cA$-module $\hm_\cA(\cJ^\infty(P),Q)$ (Example
\ref{nn23}).

\begin{theorem} \label{nn31} \mar{nn31}
We have the isomorphisms (\ref{5.50}) of the direct systems
(\ref{nn10}) and (\ref{nn11}) which leads to an $\cA$-module
isomorphism
\mar{nn32}\beq
\dif_\infty(P,Q)=\hm_{\cA}(\cJ^\infty(P),Q) \label{nn32}
\eeq
of their direct limits in accordance with Proposition \ref{nn17}.
\end{theorem}

\begin{proof}
Any element $\Delta_\infty=\Delta \in \dif_\infty(P,Q)$ factorizes
as
\mar{nn36}\beq
\Delta_\infty: P\ar^{J^\infty}
\cJ^\infty(P)\ar^{\mathfrak{f}^\Delta_\infty} Q \label{nn36}
\eeq
through the morphism $J^\infty$ (\ref{nn30}) and an $\cA$-module
homomorphism
$\mathfrak{f}^\Delta_\infty=\mathfrak{f}^\Delta\circ\pi^\infty_k$
(Example \ref{nn22}) so that the diagram
\be
\begin{array}{llcrr}
 &  & \cJ^\infty(P)& &  \\
 & ^{J^\infty}\put(-10,-10){\vector(1,1){15}} & \put(0,10){\vector(0,-1){15}}
 & \put(-15,5){\vector(1,-1){15}}^{\mathfrak{f}^\Delta_\infty} & \\
  P & \ar_{J^k} & \cJ^k(P) & \ar_{\mathfrak{f}^\Delta} & Q
\end{array}
\ee
is commutative.
\end{proof}

Let us consider jet modules $\cJ^s=\cJ^s(\cA)$ of a ring $\cA$
itself. In particular, the first-order jet module $\cJ^1$ consists
of the elements $a\otimes_1 b$, $a,b\in\cA$, subject to the
relations
\mar{5.53}\beq
ab\otimes_1 \bb -b\otimes_1 a -a\otimes_1 b +\bb\ot_1(ab) =0.
\label{5.53}
\eeq
The $\cA$- and $\cA^\bll$-module structures (\ref{+a21}) on
$\cJ^1$ read
\be
c(a\ot_1 b)=(ca)\ot_1 b,\qquad c\bll(a\ot_1 b)=
a\ot_1(cb)=(a\ot_1b)c.
\ee

Besides the monomorphism (\ref{5.44}):
\be
J^1: \cA\ni a\to \bb\otimes_1 a\in \cJ^1,
\ee
there is an $\cA$-module monomorphism
\be
i_1: \cA \ni a  \to a\otimes_1 \bb\in \cJ^1.
\ee
With these monomorphisms, we have a canonical $\cA$-module
splitting
\mar{mos058}\ben
&& \cJ^1=i_1(\cA)\oplus \cO^1, \label{mos058} \\
&& aJ^1(b)= a\ot_1 b=ab\ot_1\bb + a(\bb\ot_1 b- b\ot_1\bb),
\nonumber
\een
where an $\cA$-module $\cO^1$ is generated by elements $\bb\ot_1
b-b\ot_1 \bb$ for all $b\in\cA$. Let us consider the corresponding
$\cA$-module epimorphism
\mar{+216}\beq
h^1:\cJ^1\ni \bb\ot_1 b\to \bb\ot_1 b-b\ot_1 \bb\in \cO^1
\label{+216}
\eeq
and the composition
\mar{mos045}\beq
d^1=h^1\circ J^1: \cA \ni b \to \bb\ot_1 b- b\ot_1\bb \in \cO^1,
\label{mos045}
\eeq
which is a $\cK$-module morphism. This is an $\cO^1$-valued
derivation of a $\cK$-ring $\cA$ which obeys the Leibniz rule
\be
d^1(ab)= \bb\ot_1 ab-ab\ot_1\bb +a\ot_1 b  -a\ot_1 b  =ad^1b +
(d^1a)b.
\ee
It follows from the relation (\ref{5.53}) that
\mar{nn89}\beq
ad^1b=(d^1b)a \label{nn89}
\eeq
for all $a,b\in \cA$. Thus, seen as an $\cA$-module, $\cO^1$ is
generated by elements $d^1a$ for all $a\in\cA$.

Let $\cO^{1*}=\hm_{\cA}(\cO^1,\cA)$ be the dual  of an
$\cA$-module $\cO^1$. In view of the splittings (\ref{spr156}) and
(\ref{mos058}), the isomorphism (\ref{5.50}) reduces to the
duality relation
\mar{5.81,a}\ben
&& \gd\cA=\cO^{1*}, \label{5.81}\\
&& \gd\cA\ni u\leftrightarrow
   f_u\in \cO^{1*}, \qquad f_u(d^1a)=u(a), \qquad  a\in \cA.
   \label{5.81a}
\een
However, a monomorphism $\cO^1 \to \cO^{1**}=\gd\cA^*$ need not be
an isomorphism.

\begin{remark} \label{nn244} \mar{nn244} In view of the relation (\ref{5.81}), one thinks of elements of
an $\cA$-module $\cO^1$ as being differential forms over a ring
$\cA$.
\end{remark}

Let us consider the exterior algebra $\cO^*=\w\cO^1$ of an
$\cA$-module $\cO^1$ (Example \ref{ws40}). There exist the higher
degree generalizations
\mar{5.4}\ben
&& h^k: \cJ^1(\cO^{k-1})\to \cO^k, \nonumber\\
&& d^k=h^k\circ J^1: \cO^{k-1}\to \cO^k \label{5.4}
\een
of the morphisms (\ref{+216}) and (\ref{mos045}). The operators
(\ref{5.4}) are nilpotent, i.e., $d^k\circ d^{k-1}=0$. They  form
a cochain complex
\mar{55.63}\beq
0\to \cK\ar\cA\ar^{d^1}\cO^1\ar^{d^2} \cdots  \cO^k\ar^{d^{k+1}}
\cdots, \label{55.63}
\eeq
which is the minimal Chevalley--Eilenberg differential calculus
(\ref{t10}) over a $\cK$-ring $\cA$ (Theorem \ref{nn241}).

\subsection{Chevalley--Eilenberg differential calculus}

We start with a general notion of the differential graded ring
which is not necessarily commutative.

\begin{definition} \label{nn51} \mar{nn51}
An $\mathbb N$-graded ring $\Om^*$ (Definition \ref{nn40}) is
called the differential graded ring (henceforth, DGR) if it is a
cochain complex of $\cK$-modules
\mar{spr260}\beq
0\to \cK\ar\Om^0\ar^\dl\Om^1\ar^\dl\cdots\Om^k\ar^\dl\cdots
\label{spr260}
\eeq
with respect to a coboundary operators $\dl$ which obeys the
graded Leibniz rule
\mar{1006}\beq
\dl(\al\cdot\bt)=\dl\al\cdot\bt +(-1)^{|\al|}\al\cdot \dl\bt.
\label{1006}
\eeq
\end{definition}

In particular, $\dl:\Om^0\to \Om^1$ is a $\Om^1$-valued derivation
of a $\cK$-ring $\Om^0$ (Definition \ref{nn95}).

The cochain complex (\ref{spr260}) is called the de Rham complex
of a differential graded ring $(\Om^*,\dl)$. It also is said to be
the differential graded calculus over a $\cK$-ring $\Om^0$.
Cohomology $H^*(\Om^*)$ of the complex (\ref{spr260}) is called
the   de Rham cohomology of a differential graded ring
$(\Om^*,\dl)$.

Given a differential graded ring $(\Om^*,\dl)$, one considers its
minimal differential graded subring $(\ol\Om^*,\dl)$ which
contains $\Om^0$. Seen as a $(\Om^0-\Om^0)$-ring, it is generated
by elements $\dl a$, $a\in \cA$, and consists of monomials
$\al=a_0\dl a_1\cdots \dl a_k$, $a_i\in \Om^0$, whose product
obeys the   juxtaposition rule
\be
(a_0\dl a_1)\cdot (b_0\dl b_1)=a_0\dl (a_1b_0)\cdot \dl b_1-
a_0a_1\dl b_0\cdot \dl b_1
\ee
in accordance with the equality (\ref{1006}).

\begin{definition} \label{nn240} \mar{nn240}
A complex $(\ol\Om^*,\dl)$ is called the minimal differential
graded calculus over $\Om^0$.
\end{definition}

Let us show that any commutative $\cK$-ring $\cA$ defines the
differential graded calculus (\ref{nn55}), called the
Chevalley--Eilenberg differential calculus over $\cA$.

Since the derivation module $\gd\cA$ of $\cA$ is a Lie
$\cK$-algebra, let us consider the extended Chevalley--Eilenberg
complex $C^*[\gd\cA;\cA]$ (\ref{spr997}):
\mar{ws102}\beq
0\to \cK\ar^\mathrm{in}C^*[\gd\cA;\cA], \label{ws102}
\eeq
of the Lie algebra $\gd\cA$ with coefficients in a ring $\cA$,
regarded as a $\gd\cA$-module \cite{book05,book12}. This complex
contains a subcomplex $\cO^*[\gd\cA]$:
\mar{nn55}\beq
0\to \cK\ar^\mathrm{in}\cA\ar^d\cO^1[\gd\cA]\ar^d\cdots,
\label{nn55}
\eeq
of $\cA$-multilinear skew-symmetric maps
\mar{+840'}\beq
\cO^k[\gd\cA]=\hm_\cA(\op\times^k\gd\cA,\cA)\ni\f^k:\op\times^k
\gd\cA\to \cA \label{+840'}
\eeq
with respect to the Chevalley--Eilenberg coboundary operator
(\ref{spr132}):
\mar{+840}\ben
&& d\f(u_0,\ldots,u_k)=\op\sum^k_{i=0}(-1)^iu_i
(\f(u_0,\ldots,\wh{u_i},\ldots,u_k)) +\label{+840}\\
&& \qquad \op\sum_{i<j} (-1)^{i+j}
\f([u_i,u_j],u_0,\ldots, \wh u_i, \ldots, \wh u_j,\ldots,u_k).
\nonumber
\een
Indeed, a direct verification shows that if $\f$ (\ref{+840'}) is
an $\cA$-multilinear map, $d\f$ (\ref{+840}) also is well.

In particular,
\mar{spr708}\ben
&& (d a)(u)=u(a), \qquad a\in\cO^0[\gd\cA]=\cA, \label{spr708}\\
&&(d\f)(u_0,u_1)= u_0(\f(u_1)) -u_1(\f(u_0)) -\f([u_0,u_1]),
\label{+921}\\
&&  \f\in \cO^1[\gd\cA]=\hm_\cA(\gd\cA,\cA). \nonumber
\een \mar{+921}
It follows that $d(\bb)=0$, i.e., $d$ is an $\cO^1[\gd\cA]$-valued
derivation of $\cA$ (Definition \ref{nn95}).

Let us define an $\mathbb N$-graded module
\mar{nn99}\beq
\cO^*[\gd\cA] =\op\oplus_{i\in\mathbb N} \cO^i[\gd\cA].
\label{nn99}
\eeq
It is provided with the structure of an $\mathbb N$-graded
$\cA$-ring with respect to a product
\mar{ws103}\ben
&& \f\w\f'(u_1,...,u_{r+s})= \label{ws103}\\
&& \qquad \op\sum_{i_1<\cdots<i_r;j_1<\cdots<j_s} \mathrm{sgn}^{i_1\cdots i_rj_1\cdots j_s}_{1\cdots r+s} \f(u_{i_1},\ldots,
u_{i_r}) \f'(u_{j_1},\ldots,u_{j_s}), \nonumber \\
&& \f\in \cO^r[\gd\cA], \qquad \f'\in \cO^s[\gd\cA], \qquad u_k\in \gd\cA,
\nonumber
\een
where sgn$^{...}_{...}$ denotes the sign of a permutation. This
product obeys relations
\mar{ws98,9}\ben
&& d(\f\w\f')=d(\f)\w\f' +(-1)^{|\f|}\f\w d(\f'),
\quad \f,\f'\in \cO^*[\gd\cA], \label{ws98}\\
&& \f\w \f' =(-1)^{|\f||\f'|}\f'\w \f. \label{ws99}
\een
By virtue of the first one, $(\cO^*[\gd\cA],d)$ is a differential
graded ring (Definition \ref{nn51}), called the
Chevalley--Eilenberg differential calculus over a $\cK$-ring $\cA$
\cite{book05,book12}. The relation (\ref{ws99}) shows that
$\cO^*[\gd\cA]$ is a graded commutative ring (Definition
\ref{nn44}).

Since $\cO^1[\gd\cA]=\hm_\cA(\gd\cA,\cA)=\gd\cA^*$ and,
consequently, $\gd\cA\subset\gd\cA^{**}=\cO^1[\gd\cA]^*$, we have
the interior product
\mar{nn80}\beq
u\rfloor\f=\f(u), \qquad u\in\gd\cA, \qquad \f\in \cO^1[\gd\cA].
\label{nn80}
\eeq
It is extended as
\mar{nn81}\beq
(u\rfloor\f)(u_1,\ldots,u_{k-1})= k\f(u,u_1,\ldots,u_{k-1}), \quad
u\in\gd\cA, \quad \f\in \cO^*[\gd\cA],\label{nn81}
\eeq
to a differential graded ring $(\cO^*[\gd\cA],d)$, and obeys a
relation
\mar{nn100}\beq
u\rfloor(\f\w \si)=u\rfloor \f\w \si +(-1)^{|\f|}\f\w u\rfloor\si.
\label{nn100}
\eeq
With the interior product (\ref{nn81}), one defines a derivation
\mar{nn82,'}\ben
&& \bL_u(\f)=d(u\rfloor\f) +u\rfloor d\f, \quad \f\in
\cO^*[\gd\cA], \label{nn82}\\
&& \bL_u(\f\w\si)=\bL_u(\f)\w \si +\f\w\bL_u\si, \label{nn82'}
\een
of an $\mathbb N$-graded ring $(\cO^*[\gd\cA]$ for any
$u\in\gd\cA$. Then one can think of elements of $\cO^*[\gd\cA]$ as
being differential forms over $\cA$.

The minimal Chevalley--Eilenberg differential calculus $\cO^*\cA$
over a ring  $\cA$ consists of the monomials $a_0da_1\w\cdots\w
da_k$, $a_i\in\cA$ (Definition \ref{nn240}).

\begin{theorem} \label{nn241} \mar{nn241}
The de Rham complex
\mar{t10}\beq
0\to\cK\ar \cA\ar^d\cO^1\cA\ar^d \cdots  \cO^k\cA\ar^d \cdots
\label{t10}
\eeq
of $\cO^*\cA$ is exactly the cochain complex (\ref{55.63}).
\end{theorem}

\begin{proof}
Comparing the equalities (\ref{5.81a}) and (\ref{spr708}) shows
that $d^1=d$ on an $\cA$-module
\be
\cO^1\cA=\cO^1\subseteq \cO^1[\gd\cA]=\gd\cA^*,
\ee
generated by elements $d^1a$, $a\in\cA$. The complex (\ref{t10})
is called the de Rham complex of a $\cK$-ring $\cA$, and its
cohomology $H^*(\cA)$ is said to be the de Rham cohomology of
$\cA$.
\end{proof}

\begin{example} \label{nn224} \mar{nn224}
Let $\cP[Q]$ be a polynomial ring of a free $\cK$-module $Q$ of
finite rank in Example \ref{nn221}. Since $\gd \cP[Q]$ is a free
$\cP[Q]$-module of finite rank, its $\cP[Q]$-dual $\cO^1[\gd
\cP[Q]]$ is a free $\cP[Q]$-module of finite rank possessing the
dual basis $\{dq^i\}$ such that the relation (\ref{nn80}) takes a
form $\dr_i\rfloor da^j=\dl^j_i$. It follows that the
Chevalley--Eilenberg differential calculus (\ref{nn99}) of
$\cP[Q]$ consists of monomials
\mar{nn227}\beq
\phi=\phi_{i_1\ldots i_r}dq^{i_1}\w\cdots\w dq^{i_r}, \qquad
\phi_{i_1\ldots i_r}\in \cP[Q]. \label{nn227}
\eeq
Thus, it is the minimal Chevalley--Eilenberg differential calculus
(\ref{t10}):
\mar{nn225}\beq
0\to\cK\ar \cP[Q] \ar^d\cO^1\ar^d \cdots  \cO^r\ar^d \cdots,
\label{nn225}
\eeq
where the Chevalley--Eilenberg coboundary operator reads
\mar{nn226}\beq
d\phi =\dr_k(\phi_{i_1\ldots i_r})dq^k\w dq^{i_1}\w\cdots\w
dq^{i_r}. \label{nn226}
\eeq
\end{example}

\subsection{Connections on modules and rings}

We employ  the jets of modules in Section 2.3 in order to
introduce connections on modules over commutative rings
\cite{book00,book12}.

It is readily observed that a first-order jet module $\cJ^1(P)$ of
an $\cA$-module $P$ is isomorphic to a tensor product
\be
\cJ^1(P)=\cJ^1\ot P, \qquad (a\ot_1 bp) \leftrightarrow (a\ot_1
b)\ot p.
\ee
of an $\cA^\bll$-module $\cJ^1$ and an $\cA$-module $P$. Then the
isomorphism (\ref{mos058}) leads to the splitting
\mar{mos071}\ben
&& \cJ^1(P)= (\cA\oplus \cO^1)\ot P=
(\cA \ot P)\oplus (\cO^1\ot P), \label{mos071}\\
&& a\ot_1 bp\leftrightarrow  (ab +ad^1(b))\ot p. \nonumber
\een
Applying the epimorphism $\pi^1_0$ (\ref{+a13}) to this splitting,
one obtains a short exact sequence of $\cA$- and
$\cA^\bll$-modules
\mar{+175}\ben
&& 0\ar \cO^1\ot P\to \cJ^1(P)\ar^{\pi^1_0} P\ar 0, \label{+175}\\
&&  (a\ot_1 b -ab\ot_1 \bb)\ot p\to
(c\ot_1 \bb+ a\ot_1 b -ab\ot_1 \bb)\ot p \to cp. \nonumber
\een
This exact sequence canonically is split by an $\cA^\bll$-module
morphism
\be
P\ni ap \to \bb\ot ap= a\ot p + d^1(a)\ot p\in\cJ^1(P).
\ee
However, it need not be split by an $\cA$-module morphism, unless
$P$ is a projective $\cA$-module.

\begin{definition} \label{+176} \mar{+176}
A connection on an $\cA$-module $P$ is defined as an $\cA$-module
morphism
\mar{+179'}\beq
\G:P\to \cJ^1(P), \qquad \G (ap)=a\G(p), \label{+179'}
\eeq
which splits the exact sequence (\ref{+175}) or, equivalently, the
exact sequence
\mar{+183}\beq
0\to \cO^1\ot P\to (\cA\oplus \cO^1)\ot P\to P\to 0. \label{+183}
\eeq
\end{definition}

If the splitting $\G$ (\ref{+179'}) exists, it reads
\be
J^1p=\G(p) + \nabla(p),
\ee
where $\nabla$ is the complementary morphism
\mar{+179}\beq
\nabla: P\to \cO^1\ot P, \qquad \nabla(p)= \bb\ot_1 p-
\G(p).\label{+179}
\eeq
Though this complementary morphism in fact is a   covariant
differential on a module $P$, it is traditionally called the
connection on a module. It satisfies the   Leibniz rule
\mar{+180}\beq
\nabla(ap)= d^1a \ot p +a\nabla(p), \label{+180}
\eeq
i.e., $\nabla$ is  an $(\cO^1\ot P)$-valued first-order
differential operator on $P$. Thus, we come to the following
equivalent definition of a connection \cite{kosz60}.

\begin{definition} \label{+181} \mar{+181}
A   connection on an $\cA$-module $P$ is the $\cK$-module morphism
$\nabla$ (\ref{+179}) which obeys the Leibniz rule (\ref{+180}).
Sometimes, it is called the   Koszul connection.
\end{definition}

The morphism $\nabla$ (\ref{+179}) naturally can be extended to a
morphism
\be
\nabla: \cO^1\ot P\to \cO^2\ot P.
\ee
Then we have a morphism
\be
R=\nabla^2: P\to \cO^2 \ot P,
\ee
called the   curvature of a connection $\nabla$ on a module $P$.

In view of the isomorphism (\ref{5.81}), any connection in
Definition \ref{+181} determines a connection in the following
sense.

\begin{definition} \label{1016} \mar{1016}
A   connection on an $\cA$-module $P$ is an $\cA$-module morphism
\mar{1017}\beq
\gd\cA\ni u\to \nabla_u\in \dif_1(P,P) \label{1017}
\eeq
such that first-order differential operators $\nabla_u$ obey the
  Leibniz rule
\mar{1018}\beq
\nabla_u (ap)= u(a)p+ a\nabla_u(p), \quad a\in \cA, \quad p\in P.
\label{1018}
\eeq
\end{definition}

Definitions \ref{+181} and \ref{1016} are equivalent if
$\cO^1=\gd\cA^*=\cO^{1**}$. For instance, this is the case of a
projective module $\cO^1$ of finite rank (Theorem \ref{nn228}).

A   curvature of the connection (\ref{1017}) is defined as a
zero-order differential operator
\mar{+100}\beq
R(u,u')=[\nabla_u,\nabla_{u'}] -\nabla_{[u,u']} \label{+100}
\eeq
on a module $P$ for all $u,u'\in \gd\cA$.

Let $P$ be a commutative $\cA$-ring and $\gd P$ the derivation
module of $P$ as a $\cK$-ring. Definition \ref{1016} is modified
as follows.

\begin{definition} \label{mos088} \mar{mos088}
A   connection on an $\cA$-ring $P$ is an $\cA$-module morphism
\mar{mos090}\beq
\gd\cA\ni u\to \nabla_u\in \gd P, \label{mos090}
\eeq
which is a connection on $P$ as an $\cA$-module, i.e., obeys the
Leinbniz rule (\ref{1018}).
\end{definition}

Two such connections $\nabla_u$ and $\nabla'_u$ differ from each
other in a derivation of an $\cA$-ring $P$, i.e., which vanishes
on $\cA\subset P$. A curvature of the connection (\ref{mos090}) is
given by the formula (\ref{+100}).

\section{Local-ringed spaces}

Local-ringed spaces are sheaves of local rings. For instance,
smooth manifolds, represented by sheaves of real smooth functions,
constitute a subcategory of the category of local-ringed spaces
(Section 4).

A sheaf $\gR$ on a topological space $X$ is said to be the ringed
space if its stalk $\gR_x$ at each point $x\in X$ is a commutative
ring \cite{book05,book12,tenn}. A ringed space often is denoted by
a pair $(X,\gR)$ of a topological space $X$ and a sheaf $\gR$ of
rings on $X$ which are called the   body and the   structure sheaf
of a ringed space, respectively.

In comparison with morphisms of sheaves on the same topological
space in Section 8.3, morphisms of ringed spaces are defined to be
particular morphisms of sheaves on different topological spaces as
follows.

\begin{example} \label{nn310} \mar{nn310}
Let $\vf:X\to X'$ be a continuous map. Given a sheaf $S$ on $X$,
its   direct image $\vf_*S$ on $X'$ is generated by the presheaf
of assignments
\be
X'\supset U'\to S(\vf^{-1}(U'))
\ee
for any open subset $U'\subset X'$. Conversely, given a sheaf $S'$
on $X'$, its   inverse image $\vf^*S'$ on $X$ is defined as the
pull-back onto $X$ of a topological fibre bundle $S'$ over $X'$,
i.e., $\vf^*S'_x=S_{\vf(x)}$. This sheaf is generated by the
presheaf which associates to any open $V\subset X$ the direct
limit of modules $S'(U)$ over all open subsets $U\subset X'$ such
that $V\subset f^{-1}(U)$.
\end{example}

\begin{example} \label{spr201} \mar{spr201}
Let $i:X\to X'$ be a closed subspace of $X'$. Then $i_*S$ is a
unique sheaf on $X'$ such that
\be
i_*S|_X=S, \qquad i_*S|_{X'\setminus X}=0.
\ee
Indeed, if $x'\in X\subset X'$, then $i_*S(U')= S(U'\cap X)$ for
any open neighborhood $U$ of this point. If $x'\not\in X$, there
exists its neighborhood $U'$ such that $U'\cap X$ is empty, i.e.,
$i_*S(U')=0$. A sheaf $i_*S$ is called the   trivial extension
 of a sheaf $S$.
\end{example}

By a   morphism of ringed spaces $(X,\gR)\to (X',\gR')$ is meant a
pair $(\vf,\Phi)$ of a continuous map $\vf:X\to X'$ and a sheaf
morphism $\Phi:\gR'\to \vf_*\gR$ or, equivalently, a sheaf
morphisms $\vf^*\gR'\to \gR$ \cite{tenn}. Restricted to each
stalk, a sheaf morphism $\Phi$ is assumed to be a ring
homomorphism. A morphism of ringed spaces is said to be:

$\bullet$ a monomorphism if $\vf$ is an injection and $\Phi$ is an
epimorphism,

$\bullet$ an epimorphism if $\vf$ is a surjection, while $\Phi$ is
a monomorphism.

\begin{definition} \label{nn250} \mar{nn250}
A ringed space is said to be the   local-ringed space (the
geometric space in the terminology of \cite{tenn}) if it is a
sheaf of local rings.
\end{definition}

A key point of the study of local-ringed space is that any
projective module over a local ring is free (Theorem \ref{nn1}).

\begin{example} \label{nn230} \mar{nn230}
A sheaf $C^0_X$ of germs of continuous real functions on a
topological space $X$ (Example \ref{spr7}) is a local-ringed
space. Its stalk $C^0_x$, $x\in X$, contains a unique maximal
ideal of germs of functions vanishing at $x$.
\end{example}

\subsection{Differential calculus over local-ringed spaces}

Let $(X,\gR)$ be a local-ringed space. By a   sheaf $\gd \gR$ of
derivations of the sheaf $\gR$ is meant a subsheaf of
endomorphisms of $\gR$ such that any section $u$ of $\gd \gR$ over
an open subset $U\subset X$ is a derivation of a ring $\gR(U)$. It
should be emphasized that, since the monomorphism (\ref{+212}) is
not necessarily an isomorphism, a derivation of a ring $\gR(U)$
need not be a section of a sheaf $\gd \gR|_U$. Namely, it may
happen that, given open sets $U'\subset U$, there is no
restriction morphism
\be
\gd (\gR(U)) \to\gd (\gR(U')).
\ee

Given a local-ringed space $(X,\gR)$, a sheaf $P$ on $X$ is called
the   sheaf of $\gR$-modules if every stalk $P_x$, $x\in X$, is an
$\gR_x$-module or, equivalently, if $P(U)$ is an $\gR(U)$-module
for any open subset $U\subset X$. A sheaf of $\gR$-modules $P$ is
said to be   locally free if there exists an open neighborhood $U$
of every point $x\in X$ such that $P(U)$ is a free
$\gR(U)$-module. If all these free modules are of finite rank
(resp. of the same finite rank), one says that $P$ is of finite
type (resp. of constant rank). The structure module of a locally
free sheaf is called the  locally free module.

Let $(X,\gR)$ be a local-ringed space and $\gP$ a sheaf of
$\gR$-modules on $X$.  For any open subset $U\subset X$, let us
consider a jet module $\cJ^1(\gP(U))$ of a module $\gP(U)$. It
consists of elements of $\gR(U)\ot \gP(U)$ modulo the pointwise
relations (\ref{mos041}). Hence, there is the restriction morphism
\be
\cJ^1(\gP(U))\to \cJ^1(\gP(V))
\ee
for any open subsets $V\subset U$, and the jet modules
$\cJ^1(\gP(U))$ constitute a presheaf. This presheaf defines the
  sheaf $\gj^1\gP$ of jets of $\gP$ (or simply the   jet
sheaf).  The jet sheaf $\gj^1\gR$ of a sheaf $\gR$ of local rings
is introduced in a similar way. Since the relations (\ref{mos041})
and (\ref{5.53}) on a ring $\gR(U)$ and modules $\gP(U)$,
$\cJ^1(\gP(U))$, $\cJ^1(\gR(U))$  are pointwise relations for any
open subset $U\subset X$, they commute with the restriction
morphisms. Therefore, the direct limits of the quotients modulo
these relations exist \cite{massey}. Then we have the sheaf
$\cO^1\gR$ of one-forms over a sheaf $\gR$, the sheaf isomorphism
\be
\gj^1(\gP)=(\gR\oplus \cO^1\gR)\ot \gP,
\ee
and the exact sequences of sheaves
\mar{+213}\ben
&&0\to \cO^1\gR\ot \gP\to \gj^1(\gP)\to \gP\to 0, \label{+213}\\
&& 0\to \cO^1\gR\ot \gP\to (\gR\oplus \cO^1\gR)\ot \gP\to \gP\to 0. \label{+214}
\een
They reflect the quotient (\ref{+216}), the isomorphism
(\ref{mos071}) and the exact sequences of modules (\ref{+175}),
(\ref{+183}), respectively.

\begin{remark}
It should be emphasized that, because of the inequality
(\ref{+212}), the duality relation (\ref{5.81}) is not extended to
the sheaves $\gd \gR$ and $\cO^1\gR$ in general, unless $\gd \gR$
and $\cO^1\gR$ are locally free sheaves of finite rank. If $\gP$
is a locally free sheaf of finite rank, so is $\gj^1\gP$.
\end{remark}

Following Definitions \ref{+176} and \ref{+181} of a connection on
modules, we come to the following notion of a connection on
sheaves.

\begin{definition} \label{+217} \mar{+217}
Given a local-ringed space $(X,\gR)$ and a sheaf $\gP$ of
$\gR$-modules on $X$, a   connection on a sheaf $\gP$ is defined
as a splitting of the exact sequence (\ref{+213}) or,
equivalently, the exact sequence (\ref{+214}).
\end{definition}

Theorem \ref{spr30} leads to the following compatibility of the
notion of a connection on sheaves with that of a connection on
modules.

\begin{proposition} \label{+219} \mar{+219}
If there exists a connection on a sheaf $\gP$ in Definition
\ref{+217}, then there exists a connection on a module $\gP(U)$
for any open subset $U\subset X$. Conversely, if for any open
subsets $V\subset U\subset X$ there are connections on modules
$\gP(U)$ and $\gP(V)$ related by the restriction morphism, then
the sheaf $\gP$ admits a connection.
\end{proposition}

As an immediate consequence of Proposition \ref{+219}, we find
that the exact sequence of sheaves (\ref{+214}) is split iff there
exists a sheaf morphism
\mar{+3}\beq
\nabla: \gP\to \cO^1\gR\ot \gP, \label{+3}
\eeq
satisfying the Leibniz rule
\be
\nabla (fs)=df\ot s + f\nabla(s), \qquad f\in \cA(U), \qquad s\in
\gP(U),
\ee
for any open subset $U\in X$. It leads to the following equivalent
definition of a connection on sheaves in the spirit of Definition
\ref{+181}.

\begin{definition} \label{+4} \mar{+4}
The sheaf morphism (\ref{+3}) is a   connection on a sheaf $\gP$.
\end{definition}

Similarly to the case of connections on modules, a   curvature of
the connection (\ref{+3}) on a sheaf $\gP$ is given by the
expression
\mar{+105}\beq
R=\nabla^2:\gP\to \cO^2_X\ot \gP. \label{+105}
\eeq

The exact sequence (\ref{+214}) need not be split. One can obtain
the following criteria of the existence of a connection on a
sheaf.

Let $\gP$ be a locally free sheaf of $\gR$-modules. Then we have
the exact sequence of sheaves
\be
0\to \hm(\gP,\cO^1\gR\ot \gP)\to \hm(\gP,(\gR\oplus\cO^1\gR)\ot
\gP) \to \hm(\gP,\gP)\to 0
\ee
and the corresponding exact sequence (\ref{spr227}) of the
cohomology groups
\be
&& 0\to H^0(X;\hm(\gP,\cO^1\gR\ot \gP)) \to H^0(X;
\hm(\gP,(\gR\oplus\cO^1\gR)\ot \gP))\to \\
&& \qquad H^0(X;\hm(\gP,\gP))\to H^1(X;\hm(\gP,\cO^1\gR\ot \gP))\to \cdots.
\ee
The identity morphism $\id :\gP\to \gP$ belongs to
$H^0(X;\hm(\gP,\gP))$. Its image in $H^1(X;\hm(\gP,\cO^1\gR\ot
\gP))$ is called the   Atiyah class.  If this class vanishes,
there exists an element of $\hm(\gP,(\gR\oplus\cO^1\gR)\ot \gP))$
whose image is $\id \gP$, i.e., a splitting of the exact sequence
(\ref{+214}).

\subsection{Affine schemes}

One can associate to any commutative ring $\cA$ a local-ringed
space as follows.

Let $\cI$ be a prime ideal of $\cA$. Then $\cA\setminus\cI$ is a
multiplicative subset of $\cA$ (Remark \ref{ws90}) and
$(\cA\setminus\cI)^{-1}\cA$ is a local ring.

Let Spec$\,\cA$ be a set of prime ideals of a ring $\cA$. It is
called the spectrum of $\cA$. Let us assign to each ideal $\cI$ of
$\cA$ a set
\mar{t20}\beq
V(\cI)=\{\, x\in \mathrm{Spec}\,\cA\,:\, \cI\subseteq x\,\}.
\label{t20}
\eeq
These sets possess the properties
\be
&& V(\{0\})=\mathrm{Spec}\,\cA, \qquad V(\cA)=\emptyset, \\
&& \op\cap_i V(\cI_i)=V(\op\sum_i \cI_i), \qquad V(\cI)\cup
V(\cI')=V(\cI\cI').
\ee
In view of these properties, one can regard the sets (\ref{t20})
as closed sets of some topology on a spectrum Spec$\,\cA$. It is
called the Zariski topology.  A base for this topology consists of
sets
\mar{t21}\beq
U(a)=\{\, x\in \mathrm{Spec}\,\cA\, :\, a\not\in x\,\} =
\mathrm{Spec}\,\cA \setminus V(\cA a) \label{t21}
\eeq
as $a$ runs through $\cA$. In particular, a set of closed points
is nothing but the set Specm$\,\cA$ of all maximal ideals of
$\cA$. Endowed with the relative Zariski topology, it is called
the maximal spectrum of $\cA$. A ring morphism $\zeta:\cA\to\cA'$
yields a continuous map
\mar{t30}\beq
\zeta^\natural:\mathrm{Spec}\,\cA'\ni x'\mapsto \zeta^{-1}(x')\in
\mathrm{Spec}\,\cA. \label{t30}
\eeq
In particular, let $\pi:\cA\to\cA/\cI$ be the quotient morphism
with respect to an ideal $\cI$ of $\cA$. Then $\zeta^\natural$
(\ref{t30}) is a homeomorphism of Spec$(\cA/\cI)$ onto a closed
subspace $V(\cI)\subset \mathrm{Spec}\,\cA$.

Given a commutative ring $\cA$ and its spectrum Spec$\,\cA$, one
can define a sheaf $\gA$ on Spec$\,\cA$ whose stalk at a point
$x\in \mathrm{Spec}\,\cA$ is a local ring $\gA_x= (\cA\setminus
x)^{-1}\cA$. A structure ring $\gA(\mathrm{Spec}\,\cA)$ of global
sections of a sheaf $\gA$ is exactly a  ring $\cA$ itself. The
local-ringed space $(\mathrm{Spec}\,\cA,\gA)$ is called an affine
scheme \cite{lomb,shaf,ueno}. A local-ringed space $(X,\gR)$ which
is locally isomorphic to an affine scheme is called a scheme.  Let
us recall the following standard notions.

(i) If a scheme $(X,\gR)$ has an affine open cover
$\{U_i=\mathrm{Spec}\,\cA_i\}$ such that every $\cA_i$ is a
Noetherian ring (i.e., any ideal of $\cA_i$ is finitely
generated), $(X,\gR)$ is said to be locally Noetherian. A locally
Noetherian scheme is called Noetherian if its body $X$ is
quasi-compact.

(ii) A scheme $(X,\gR)$ is called reduced if the stalk of $\gR$ at
each point of $X$ has no nilpotent elements.

(iii) A scheme $(X,\gR)$ is said to be irreducible
 if its body $X$ is not a union of two
proper closed subsets.

(iv) A scheme is called integral if it is reduced and irreducible.

A morphism of schemes, by definition, is a morphism between them
as local-ringed spaces. Given a morphism of schemes
\mar{ws92}\beq
\vf:(X,\gR)\to (X',\gR'), \label{ws92}
\eeq
  $(X,\gR)$ is said to be a scheme over $(X',\gR')$.
 The morphism $\vf$ (\ref{ws92}) is called
separated if the range of the diagonal morphism
\be
X\to X\op\times_{X'} X
\ee
is closed. In this case, one says that $(X,\gR)$ is separated over
$X'$. A scheme $(X,\gR)$ is called separated if it is a separated
over Spec$\,\mathbb Z$. All affine schemes are separated.

A morphism of schemes
\be
(X,\gR)\to (X'=\mathrm{Spec}\,\cA,\gA)
\ee
is said to be locally of finite type (resp. of finite type) if
$(X,\gR)$ has an open affine cover (resp. a finite open affine
cover) $\{U_i=\mathrm{Spec}\,\cA_i\}$ such that each $\cA_i$ is a
finitely generated $\cA$-algebra. A general morphism of schemes
$\vf$ (\ref{ws92}) is said to be locally of finite type (resp. of
finite type) if there is an open affine cover $\{V_i\}$ of $X'$
such that every restriction of $\vf$ to $\vf^{-1}(V_i)$ is locally
of finite type (resp. of finite type). In this case, one says that
$(X,\gR)$ is locally of finite type (resp. of finite type) over
$X'$.

\begin{example} \label{ws96} \mar{ws96}
Let $\cK$ be a field. Its Spec$\,\cK$ consists only of one point
with $\cK$ as the stalk of the structure sheaf. A scheme of finite
type over Spec$\,\cK$ is called an algebraic scheme
 over $\cK$.
\end{example}

Let $(X,\gR)$ be a ringed space. A sheaf of $\gR$-modules $S$ is
said to be quasi-coherent  if for each point $x$ of $X$ there
exists a neighborhood $U$ of $x$ and an exact sequence
\be
M\ar N\ar S|_U\to 0,
\ee
where $M$ and $N$ are free $\gR|_U$-modules. Let $S$ be a locally
free sheaf of $\gR$-modules of finite type. It is said to be of
finite presentation if, locally, there exists an exact sequence
\be
\gR^m\ar \gR^n\ar S\to 0,
\ee
where $m$ and $n$ are positive integers (which need not be
globally constant). A locally free sheaf of finite type is called
 coherent if the kernel of any homomorphism $\gR|_U^n\to S|_U$,
where $n$ is an arbitrary positive integer and $U$ is an open set,
is of finite type. Obviously, if $S$ is coherent, then it is of
finite presentation, and is quasi-coherent. If $\gR$ itself is
coherent as a sheaf of $\gR$-modules, then it is called the
coherent sheaf of rings.  In this case, every sheaf of
$\gR$-modules of finite presentation is coherent. For instance,
the structure sheaf of a locally Noetherian scheme is a coherent
sheaf of rings.

Let $(X=\mathrm{Spec}\,\cA,\gR=\gA)$ be an affine scheme. Then
every quasi-coherent sheaf $S$ of $\gA$-modules on $X$ is
generated by its global sections. The correspondence $S\to S(X)$
defines an equivalence between the category of quasi-coherent
sheaves on $X$ and the category of $\cA$-modules. If $\cA$ is
Noetherian, then the coherent sheaves and the finite $\cA$-modules
correspond to each other under this equivalence.

Let $(X,\gR)$ be a separated scheme and $\gU=\{U_i\}$ an affine
open cover of $X$. For each quasi-coherent sheaf $S$ of
$\gR$-modules, the cohomology $H^*(X;S)$ of $X$ with coefficients
in $S$ is canonically isomorphic to the cohomology
$H^*(\gU;S(U))$. One defines the  cohomological dimension cd$(X)$
of a scheme $(X,\gR)$ as the largest integer $r$ such that
$H^m(X;S)\neq 0$ for a quasi-coherent sheaf of $\gR$-modules on
$X$. For instance, if $X$ is an affine scheme, then cd$(X)=0$. The
converse is true under the assumption that $X$ is Noetherian.

\subsection{Affine varieties}

Let $\cK$ throughout this Section be an algebraically closed
field, i.e., any polynomial of non-zero degree with coefficients
in $\cK$ has a root in $\cK$. The reason is that, dealing with
non-linear algebraic equations, one can not expect a simple
clear-cut theory, without assuming that a field is algebraically
closed. If $\cK$ fails to be algebraically closed, one can extend
it in an appropriate way.

Let a commutative $\cK$-ring $\cA$ be finitely generated, and let
it possess no nilpotent elements. Then it is isomorphic to the
quotient of some polynomial $\cK$-ring which is the coordinate
ring (\ref{nn260}) of a certain affine variety (Theorem
\ref{ws93}).

A subset of an $n$-dimensional affine space $\cK^n$ is called an
affine variety  if it is a set of zeros (common roots) of some set
of polynomials of $n$ variables with coefficients in $\cK$
\cite{lomb,shaf}. Unless otherwise stated, the dimension $n$ holds
fixed. Let $\cK[x]$ be a ring of polynomials of $n$ variables with
coefficients in a field $\cK$ (Example \ref{ws40}). Given an
affine variety $\cV$, a set $I(\cV)$ of polynomials in $\cK[x]$
which vanish at every point of $\cV$ is an ideal of $\cK[x]$
called the characteristic ideal of $\cV$. Herewith, $\cV=\cV'$ if
and only if $I(\cV)=I(\cV')$. Therefore, an affine variety $\cV$
can be given by the generating set of its characteristic ideal
$I(\cV)$, i.e., by a finite system $f_i=0$ of polynomials
$f_i\in\cK[x]$.

An affine variety which is a subset of another affine variety is
called a subvariety.  An affine variety is said to be irreducible
if it is not the union of two proper subvarities. A maximal
irreducible subvariety of an affine variety is called its
irreducible component. Note that any affine variety can be written
uniquely as the union of a finite number of irreducible
components. An affine variety $\cV$ is irreducible iff $I(\cV)$ is
a prime ideal. Moreover, let $\cI$ be a prime ideal of $\cK[x]$
and $\cV$ be an affine variety in $\cK^n$ of zeros of elements of
$\cI$. Then $I(\cV)=\cI$. This fact states one-to-one
correspondence between the prime ideals of $\cK[x]$ and the
irreducible affine varieties in $\cK^n$. In particular, maximal
ideals correspond to points of $\cK^n$.

The intersection and the union of subvarieties of an affine
variety $\cV$ are also subvarieties. Thus, subvarieties can be
taken as a system of closed sets of a topology on $\cV$ which is
called the Zariski topology on the affine variety $\cV$. Unless
otherwise stated, an affine variety is provided with this
topology.

Given an affine variety $\cV$, the factor ring
\mar{nn260}\beq
\cK_\cV=\cK[x]/I(\cV) \label{nn260}
\eeq
is called the coordinate ring of $\cV$. If an affine variety $\cV$
is irreducible, the ring $\cK_\cV$ (\ref{nn260}) has no divisor of
zero.

\begin{theorem} \label{ws93} \mar{ws93}
Let a $\cK$-ring $\cA$ be finitely generated, and let it possess
no nilpotent elements. Given a set $(a_1,\ldots,a_n)$ of
generating elements of $\cA$, let us consider the epimorphism
$\phi: \cK[x]\to\cA$ defined by the equalities $\phi(x_i)=a_i$.
Zeros of polynomials in Ker$\,\phi$ make up an affine variety whose
coordinate ring $\cK_\cV$ (\ref{nn260}) is exactly $\cA$.
\end{theorem}

Let $R_\cV$ be the fraction field of a coordinate ring $\cK_\cV$
(Remark \ref{ws90}). It is called the function field of $\cV$.
There is the monomorphism $\cK_\cV\to R_\cV$ (\ref{ws91}). The
function field $R_\cV$ is finitely generated over $\cK$, and its
transcendence degree is called the dimension of the irreducible
affine variety $\cV$.

\begin{example} \label{nn262} \mar{nn262}
Let $\cA$ in Theorem \ref{ws93} be a polynomial $\cK$-ring
$\cK[x]$. Then Ker$\,\phi=0$, and it yields an affine variety
$\cV=\{0\}$ whose coordinate ring (\ref{nn260}) is
$\cK_\cV=\cK[x]$.
\end{example}

Let $\cW$ be an irreducible subvariety of $\cV$ and $I(\cW)$ a
subset of $\cK_\cV$ consisting of elements which vanish on $\cW$.
Then $I(\cW)$ is a prime ideal of $\cK_\cV$. Let us consider a
multiplicative subset $\cK_\cV\setminus I(\cW)$ and a subring
\be
R_\cW=(\cK_\cV\setminus I(\cW))^{-1} \cK_\cV
\ee
of $R_\cV$ (Remark \ref{ws90}). It is said to be a local ring of a
subvariety $\cW$. Functions $f\in R_\cW\subset R_\cV$ are called
regular at $\cW$. For a given function $f\in R_\cV$, a set of
points of $\cV$ where $f$ is regular is Zariski open. Given an
open subset $U\subset\cV$, let us denote $R_U$ a ring of regular
functions on $U$. Assigning $R_U$ to each open set $U$, one can
define a sheaf of rings $\gR_\cV$ of germs of regular functions on
$\cV$. Its stalk at a point $x\in\cV$ coincides with a local ring
$R_x$. A sheaf $\gR_\cV$ is called the structure sheaf of an
affine variety $\cV$. The pair $(\cV,\gR_\cV)$ is a local-ringed
space.

Let us consider a pair $(X,\gR)$ of a topological space $X$ and
some sheaf $\gR$ of germs of $\cK$-valued functions on $X$. This
pair is called a prealgebraic variety if $X$ admits a finite open
cover $\{U_i\}$ such that each $U_i$ is homeomorphic to some
affine variety $\cV_i$ and $\gR|_{U_i}$ is isomorphic to the
structure sheaf $\gR_{\cV_i}$ of $\cV_i$. Let us note that the
Cartesian product $\cV\times\cV'$ of affine varieties
$\cV\in\cK^n$ and $\cV'\in\cK^m$ is an affine variety in
$\cK^{n+m}$ though the Zariski topology on $\cV\times\cV'$ is
finer than the product topology. Since the Cartesian product
$X\times X'$ of prealgebraic varieties is locally a product of
affine varieties, this product is a prealgebraic variety. A
prealgebraic variety $(X,\gR)$ is said to be an algebraic variety
if the diagonal map $X\to X\times X$ is closed in the Zariski
topology of the product variety. This condition corresponds to
Hausdorff's separation axiom. If $\cW$ is a locally closed subset
(i.e., the intersection of open and closed sets) of an algebraic
variety, it becomes an algebraic variety in a natural manner since
the germs of regular functions at $x\in\cW$ are taken to be the
germs of functions on $\cW$ induced by functions in the stalk
$\gR_x$. The definitions of irreducibility and local rings of
subvarieties for algebraic varieties are given in the same manner
as before. From now on, by a variety  is meant an algebraic
variety. Any variety $(X,\gR)$, by definition, is a local-ringed
space.

\begin{remark} \label{nn261} \mar{nn261}
There is the following correspondence between the affine varieties
and the affine schemes. A $\cK$-ring $\cA$ in Theorem \ref{ws93}
is a case an algebraic scheme of finite type over a field $\cK$ in
Example \ref{ws96}, and it is associated to an affine variety in
accordance with this theorem. Conversely, every algebraic affine
variety $\cV$ yields the affine scheme Spec$\cK_\cV$ such that
there is one-to-one correspondence between the points of
Spec$\cK_\cV$ and the irreducible subvarieties of $\cV$.
\end{remark}

Given a variety, one says that its point $x$ is simple and that
$\cV$ is non-singular  or smooth  at $x$ if the local ring $R_x$
of $x$ is regular. Since a problem is local, one can assume that
$\cV$ is an affine variety in $\cK^n$. Then the simplicity of $x$
implies that $x$ is contained in only one irreducible component of
$\cV$ and, if this component is $r$-dimensional, there exists
$n-r$ polynomials $f_i(x)$ in the characteristic ideal $I(\cV)$ of
$\cV$ such that the rank of a matrix $(\dr f_i/\dr x_j)$ at $x$
equals $n-r$. A point of $\cV$ which is not simple is called a
singular point.  The set of singular points, called the  singular
locus of $\cV$,  is a proper closed subset of $\cV$. A variety
possessing no singular point is said to be smooth.
 A point $x$ of a variety $\cV$ is called
 normal if the local ring $R(x)$ is normal. A simple point is
normal. Normal points make up a non-empty open subset of $\cV$. An
irreducible variety whose points are all normal is called a normal
variety.

\begin{example} \label{ws97} \mar{ws97}
When $\cK=\mathbb C$, an algebraic variety $\cV$, called a complex
algebraic variety,  has the structure of a complex analytic
manifold so that the stalk $\gR_{\cV,x}=R_x$ at $x\in\cV$ contains
in the stalk $\mathbb C^h_{\cV,x}$ of germs of holomorphic
functions at $x$ and their completion coincide. If $x$ is a simple
point, $\mathbb C^h_{\cV,x}$ is the ring of converged power series
and its completion is the ring of formal power series.
\end{example}

A regular morphism $(\cV,\gR)\to (\cV',\gR')$ of varieties over
the same field $\cK$ is defined as a morphism of local-ringed
spaces
\be
\vf:\cV\to\cV', \qquad \Phi:\vf^*\gR'\to\gR,
\ee
where $\vf^*$ is the pull-back onto $\cV$ of $\cK$-valued
functions on $\cV'$. An isomorphism of varieties also is called a
 biregular morphism.

Let $\cV$ and $\cW$ be irreducible varieties. Let a closed subset
$T\subset \cV\times\cW$ be an irreducible variety such that a
closure of the range of the projection $T\to\cV$ coincides with
$\cV$. Then the function field $R_\cV$ of $\cV$ can be identified
with a subfield of $R_T$. If $R_\cV=R_T$, then $T$ is called a
rational morphism  of $\cV$ to $\cW$. One can show that, if
$T:\cV\to\cW$ is a rational morphism and $x\in\cV$ is a normal
point of $\cV$ such that $T(x)$ contains an isolated point, then
$T$ is regular at $x$.

Let $\cV$ be an $m$-dimensional irreducible affine variety. One
can associate to $\cV$ the two algebras $\dif_*(R_\cV)$ and
$\dif_*(\cK_\cV)$ of (linear) differential operators on the
function field $R_\cV$ and the coordinate ring $\cK_\cV$ of $\cV$,
respectively \cite{mcc}.

In this case of the function field $R_\cV$, one can choose  a
separating transcendence $\{x^1,\ldots,x^m\}$ basis for $R_\cV$
over $\cK$. Let us consider the derivation module $\gd R_\cV$ of
the $\cK$-field $R_\cV$. It is finitely generated by the
derivations $\dr_i$ of $R_\cV$ such that $\dr_i(x^j)=\dl_i^j$.
Moreover, any differential operator $\Delta\in\dif_r(R_\cV)$ is
uniquely expressed as a polynomial of $\dr_i$ with coefficients in
$R_\cV$. Let $\gd R_\cV^*$ be the $R_\cV$-dual of the derivation
module $\gd R_\cV$. Given the above mentioned transcendence basis
for $R_\cV$, it is finitely generated by the elements $dx_i$ which
are the duals of $\dr_i$. As a consequence, the
Chevalley--Eilenberg calculus over $R_\cV$ coincides with the
universal differential calculus $\cO^*(R_\cV)$ over $R_\cV$
(Remark \ref{nn102}).

The case of the coordinate ring $\cK_\cV$ is more subtle. In this
case, $\dif_*(\cK_\cV)$ is called the ring of differential
operators on an affine variety $\cV$. In general, there are no
global coordinates on $\cV$, but if $\cK$ is of characteristic
zero and $\cV$ is smooth, the structure of $\dif_*(\cK_\cV)$ is
still well understood. Namely, $\dif_*(\cK_\cV)$ is a simple (left
and right) Noetherian ring without divisors of zero, and it is
generated by finitely many elements of $\dif_1(\cK_\cV)$.

If $\cV$ is singular, the construction of $\dif_*(\cK_\cV)$ is
less clear. One can show that any differential operator on
$\cK_\cV\subset R_\cV$ admits a unique extension to a differential
operator of the same order on $R_\cV$. Thus, one can regard
$\dif_*(\cK_\cV)$ as a subalgebra of $\dif_*(R_\cV)$. Furthermore,
a differential operator $\Delta$ on $R_\cV$ which preserves
$\cK_\cV$ is a differential operator on $\cK_\cV$. In particular,
it follows that $\dif_*(\cK_\cV)$ has no zero divisors. In
contrast with the smooth case, $\dif_*(\cK_\cV)$ fails to be
generated by elements of $\dif_1(\cK_\cV)$ in general. One has
conjectured that this is true if and only if $\cV$ is smooth
\cite{nakai}. This conjecture has been proved for algebraic curves
\cite{mount} and, more generally, for varieties with smooth
normalization \cite{traves}.

Let us note that, if a field $\cK$ is of positive characteristic,
the ring $\dif_*(\cK_\cV)$ is not Noetherian, or finitely
generated, or without zero divisors \cite{smith}.

\section{Differential geometry of $C^\infty(X)$-modules}

Let $X$ be a smooth manifold (Remark \ref{ws1}). Similarly to a
sheaf $C^0_X$ of continuous functions in Example \ref{nn230}, a
sheaf $C^\infty_X$ of smooth real functions on $X$ (Example
\ref{spr7}) provides an important example of local-ringed spaces
that this Section is devoted to \cite{book05,book12}.

\begin{remark} \label{ws1} \mar{ws1}
Throughout the work,   smooth manifolds are finite-dimensional
real manifolds. A smooth real manifold is customarily assumed to
be Hausdorff and second-countable (i.e., it has a countable base
for topology). Consequently, it is a locally compact space which
is a union of a countable number of compact subsets, a separable
space (i.e., it has a countable dense subset), a paracompact and
completely regular space. Being paracompact, a smooth manifold
admits the partition of unity by smooth real functions. One also
can show that, given two disjoint closed subsets $N$ and $N'$ of a
smooth manifold $X$, there exists a smooth function $f$ on $X$
such that $f|_N=0$ and $f|_{N'}=1$. Unless otherwise stated,
manifolds are assumed to be connected and, consequently, arcwise
connected. We follow the notion of a manifold without boundary.
\end{remark}

Similarly to a sheaf $C^0_X$ of continuous functions, a stalk
$C^\infty_x$ of a sheaf $C^\infty_X$ at a point $x\in X$ has a
unique maximal ideal of germs of smooth functions vanishing at $x$
(Example \ref{nn230}). Therefore, $C^\infty_X$ is a local-ringed
space (Definition \ref{nn250}).

Though a sheaf $C^\infty_X$ is defined on a topological space $X$,
it fixes a unique smooth manifold structure on $X$ as follows.

\begin{theorem} \label{+26} \mar{+26}
Let $X$ be a paracompact topological space and $(X,\gR)$ a
local-ringed space. Let $X$ admit an open cover $\{U_i\}$ such
that a sheaf $\gR$ restricted to each $U_i$ is isomorphic to a
local-ringed space  $(\mathbb{R}^n, C^\infty_{R^n})$. Then $X$ is
an $n$-dimensional smooth manifold together with a natural
isomorphism of local-ringed spaces $(X,\gR)$ and $(X,C^\infty_X)$.
\end{theorem}

One can think of this result as being an alternative definition of
smooth real manifolds in terms of local-ringed spaces. A smooth
manifold $X$ also is algebraically reproduced as a certain
subspace of the spectrum of a real ring $C^\infty(X)$ of smooth
real functions on $X$ as follows \cite{atiy,book05}.

Let $\cA$ be a commutative real ring and Specm$\, \cA$ its maximal
spectrum. The real spectrum of $\cA$ is a subspace Spec$_\mathbb R
\cA\subset \mathrm{Specm}\,\cA$ of the maximal ideals $\cI$ such
that the quotients $\cA/\cI$ are isomorphic to $\mathbb R$. It is
endowed with the relative Zariski topology. There is the bijection
between the set of real algebra morphisms of $\cA$ to a field
$\mathbb R$ and the real spectrum of $\cA$, namely,
\be
&& \hm_\mathbb R(\cA,\mathbb R)\ni \f\mapsto \Ker\f\in \mathrm{Spec}_\mathbb R\cA,\\
&& \mathrm{Spec}_\mathbb R\cA\ni x\mapsto \pi_x\in \hm_\mathbb R(\cA,\mathbb R), \qquad
\pi_x: \cA\to \cA/x\cong\mathbb R.
\ee
Any element $a\in\cA$ induces a real function
\be
f_a: \mathrm{Spec}_\mathbb R\cA \ni x\mapsto \pi_x(a)
\ee
on the real spectrum Spec$_\mathbb R \cA$. This function need not
be continuous with respect to the Zariski topology, but one can
provide Spec$_\mathbb R \cA$ with another topology, called the
Gel'fand one,  which is the coarsest topology which makes all such
functions continuous. If $\cA=C^\infty(X)$, the Zariski and
Gel'fand topologies coincide.

\begin{theorem} \label{t31} \mar{t31}
Given a ring $C^\infty(X)$ of smooth real functions on a manifold
$X$, let $\m_x$ denote the maximal ideal of functions vanishing at
a point $x\in X$. Then there is a homeomorphism
\mar{t32}\beq
\chi_X:X\ni x\mapsto \m_x\in \mathrm{Spec}_\mathbb RC^\infty(X).
\label{t32}
\eeq
\end{theorem}

Let $X$ and $X'$ be two smooth manifolds. Any smooth map $\g: X\to
X'$ induces a $\mathbb R$-ring morphism
\be
\g^*:C^\infty(X')\to C^\infty(X)
\ee
which associates the pull-back function $\g^*f= f\circ\g$ on $X$
to a function $f$ on $X'$. Conversely, each $\mathbb R$-ring
morphism
\be
\zeta:C^\infty(X')\to C^\infty(X)
\ee
yields the continuous map $\zeta^\natural$ (\ref{t30}) which sends
Spec$_\mathbb RC^\infty(X)\subset \mathrm{Spec}\,C^\infty(X)$ to
Spec$_\mathbb RC^\infty(X')\subset \mathrm{Spec}\,C^\infty(X')$ so
that the induced map
\be
\chi_{X'}^{-1}\zeta^\natural \circ \chi_X: X\to X'
\ee
is smooth. Thus, there is one-to-one correspondence between smooth
manifold morphisms $X\to X'$ and the $\mathbb R$-ring morphisms
$C^\infty(X')\to C^\infty(X)$.

\begin{remark} \label{ws2} \mar{ws2}
Let $X\times X'$ be a manifold product. A ring $C^\infty(X\times
X')$ is constructed from rings $C^\infty(X)$ and $C^\infty(X')$ as
follows. Whenever referring to a topology on a ring $C^\infty(X)$,
we will mean the topology of compact convergence for all
derivatives \cite{rob}. The $C^\infty(X)$ is a
 Fr\'echet ring  with respect to this topology, i.e., a complete
metrizable locally convex topological vector space. There is an
isomorphism of Fr\'echet rings
\mar{+55}\beq
C^\infty(X)\wh\ot C^\infty(X') \cong C^\infty(X\times X'),
\label{+55}
\eeq
where the left-hand side, called the topological tensor product,
is the completion of $C^\infty(X)\ot C^\infty(X')$ with respect to
Grothendieck's topology, defined as follows. If $E_1$ and $E_2$
are locally convex topological vector spaces,  Grothendieck's
topology  is the finest locally convex topology on $E_1\ot E_2$
such that the canonical mapping of $E_1\times E_2$ to $E_1\ot E_2$
is continuous \cite{rob}. It also is called the $\pi$-topology in
contrast with the coarser $\ve$-topology on $E_1\ot E_2$
\cite{piet,trev}. Furthermore, for any two open subsets $U\subset
X$ and $U'\subset X'$, let us consider the topological tensor
product of rings $C^\infty(U)\wh\ot C^\infty(U')$. These tensor
products define a local-ringed space $(X\times X',C^\infty_X\wh\ot
C^\infty_{X'})$. Due to the isomorphism (\ref{+55}) written for
all $U\subset X$ and $U'\subset X'$, we obtain a sheaf isomorphism
\be
C^\infty_X\wh\ot C^\infty_{X'}=C^\infty_{X\times X'}.
\ee
\end{remark}

Since a smooth manifold admits the partition of unity by smooth
functions, it follows from Proposition \ref{spr256} that any sheaf
of $C^\infty_X$-modules on $X$ is fine and, consequently, acyclic.

For instance, let $Y\to X$ be a smooth vector bundle. The germs of
its sections form a sheaf $Y_X$ of $C^\infty_X$-modules which,
thus, is fine.

In particular, all sheaves $\cO^k_X$, $k\in\mathbb N_+$, of germs
of exterior forms on $X$ are fine. These sheaves constitute the de
Rham complex
\mar{t67}\beq
0\to \mathbb{R}\ar C^\infty_X\ar^d \cO^1_X\ar^d\cdots \cO^k_X\ar^d
\cdots. \label{t67}
\eeq
The corresponding complex of structure modules of these sheaves is
the de Rham complex
\mar{t37}\beq
0\to \mathbb{R}\ar C^\infty(X)\ar^d \cO^1(X)\ar^d\cdots
\cO^k(X)\ar^d \cdots \label{t37}
\eeq
of exterior forms on a manifold $X$. Its cohomology is called the
  de Rham cohomology $H^*(X)$ of $X$. Due to the Poincar\'e
lemma, the complex (\ref{t67}) is exact and, thereby, is a fine
resolution of the constant sheaf $\mathbb R$ on a manifold. Then a
corollary of Theorem \ref{spr230} is the classical   de Rham
theorem.

\begin{theorem} \label{t60} \mar{t60} There is an isomorphism
\mar{t61}\beq
H^k(X)=H^k(X;\mathbb R) \label{t61}
\eeq
of the de Rham cohomology $H^*(X)$ of a manifold $X$ to the
cohomology of $X$ with coefficients in the constant sheaf $\mathbb
R$.
\end{theorem}

\begin{remark}
Let us consider a short exact sequence of constant sheaves
\mar{1320}\beq
0\to \mathbb{Z}\ar \mathbb{R}\ar U(1)\to 0, \label{1320}
\eeq
where $U(1)=\mathbb{R}/\mathbb{Z}$ is a circle group of complex
numbers of unit module. This exact sequence yields a long exact
sequence of sheaf cohomology groups
\be
&& 0\to \mathbb{Z}\ar\mathbb{R} \ar U(1) \ar
H^1(X;\mathbb Z) \ar H^1(X;\mathbb R)\ar\cdots \\
&& \qquad  H^p(X;\mathbb Z)\ar H^p(X;\mathbb R)\ar H^p(X;U(1))\ar
H^{p+1}(X;\mathbb Z) \ar\cdots,
\ee
where
\be
H^0(X;\mathbb Z)=\mathbb Z, \qquad H^0(X;\mathbb R)=\mathbb R
\ee
and $H^0(X;U(1))=U(1)$. This exact sequence defines a homomorphism
\mar{spr752}\beq
H^*(X;\mathbb Z)\to H^*(X;\mathbb R) \label{spr752}
\eeq
of cohomology with coefficients in the constant sheaf $\mathbb Z$
to that with coefficients in $\mathbb R$. Combining the
isomorphism (\ref{t61}) and the homomorphism (\ref{spr752}) leads
to a cohomology homomorphism
\mar{t62}\beq
H^*(X;\mathbb Z)\to H^*(X). \label{t62}
\eeq
Its kernel contains all cyclic elements of cohomology groups
$H^k(X;\mathbb Z)$.
\end{remark}

Given a vector bundle $Y\to X$, the structure module of a sheaf
$Y_X$ of germs of its sections coincides with the   structure
module $Y(X)$ of global sections of $Y\to X$ (Example \ref{t1}).
The forthcoming Serre--Swan theorem (Theorem \ref{sp60}), shows
that these modules exhaust all projective $C^\infty(X)$-modules of
finite rank. This theorem originally has been proved in the case
of a compact manifold $X$, but it is generalized to an arbitrary
smooth manifold \cite{book05,ren}.

\begin{theorem} \label{sp60} \mar{sp60}
Let $X$ be a smooth manifold. A $C^\infty(X)$-module $P$ is
isomorphic to the structure module of a smooth vector bundle over
$X$ iff it is a projective module of finite rank.
\end{theorem}

This theorem states the categorial equivalence between the vector
bundles over a smooth manifold $X$ and projective modules of
finite rank over the ring $C^\infty(X)$ of smooth real functions
on $X$. The following are \emph{Corollaries}  of this equivalence.

(i) The structure module $Y^*(X)$ of the dual $Y^*\to X$ of a
vector bundle $Y\to X$ is the $C^\infty(X)$-dual $Y(X)^*$ of the
structure module $Y(X)$ of $Y\to X$.

(ii) Any exact sequence of vector bundles
\mar{t51}\beq
0\to Y \ar Y'\ar Y''\to 0 \label{t51}
\eeq
over the same base $X$ yields an exact sequence
\mar{t52}\beq
0\to Y(X) \ar Y'(X)\ar Y''(X)\to 0 \label{t52}
\eeq
of their structure modules, and  \textit{vice versa}. In
accordance with the well-known theorem \cite{book00}, the exact
sequence (\ref{t51}) is always split. Every its splitting defines
that of the exact sequence (\ref{t52}), and  \textit{vice versa}
(Theorem \ref{nn233}).

(iii) In particular, the derivation module of a real ring
$C^\infty(X)$ coincides with a $C^\infty(X)$-module $\cT_1(X)$ of
vector fields on $X$, i.e., with the structure module of sections
of the tangent bundle $TX$ of $X$. Hence, it is a projective
$C^\infty(X)$-module of finite rank. It is the $C^\infty(X)$-dual
$\cT_1(X)=\cO^1(X)^*$ of the structure module $\cO^1(X)$ of the
cotangent bundle $T^*X$ of $X$ which is a module of one-forms on
$X$ and, conversely, $\cO^1(X)=\cT_1(X)^*$. It follows that the
Chevalley--Eilenberg differential calculus over a real ring
$C^\infty(X)$ is exactly the differential graded algebra
$(\cO^*(X),d)$ of exterior forms on $X$, where the
Chevalley--Eilenberg coboundary operator $d$ (\ref{+840})
coincides with the exterior differential. Accordingly, the de Rham
complex (\ref{t10}) of a real ring $C^\infty(X)$ is the de Rham
complex (\ref{t37}) of exterior forms on $X$. Moreover, one can
show that $(\cO^*(X),d)$ is the minimal differential calculus,
i.e., a $C^\infty(X)$-module $\cO^1(X)$ is generated by elements
$df$, $f\in C^\infty(X)$.

(iv) Let $Y\to X$ be a vector bundle and $Y(X)$ its structure
module. An $r$-order jet manifold $J^rY$ of $Y\to X$ consists of
equivalence classes $j^r_xs$, $x\in X$, of sections $s$ of $Y\to
X$ which are identified by the $r+1$ terms of their Taylor series
at points $x\in X$. Since $Y\to X$ is a vector bundle, so is a jet
bundle $J^rY\to X$. Its structure module $J^rY(X)$ is exactly the
$r$-order jet module $\cJ^r(Y(X))$ of a $C^\infty(X)$-module
$Y(X)$ in Section 2.3 \cite{book,kras}. As a consequence, the
notion of a connection on the structure module $Y(X)$ is
equivalent to the standard geometric notion of a connection on a
vector bundle $Y\to X$ \cite{book00}. Indeed, a connection on a
fibre bundle $Y\to X$ is defined as a global section $\G$ of an
affine jet bundle $J^1Y\to Y$. If $Y\to X$ is a vector bundle,
there exists an exact sequence
\mar{t50}\beq
0\to T^*X\op\ot_XY\ar J^1Y\ar Y\to 0 \label{t50}
\eeq
over $X$ which is split by $\G$. Conversely, any slitting of this
exact sequence yields a connection $Y\to X$. The exact sequence of
vector bundles (\ref{t50}) induces the exact sequence of their
structure modules
\mar{t53}\beq
0\to \cO^1(X)\op\ot Y(X)\ar J^1Y(X)\ar Y(X)\to 0. \label{t53}
\eeq
Then any connection $\G$ on a vector bundle $Y\to X$ defines a
splitting of the exact sequence (\ref{t53}) which, by Definition
\ref{+176}, is a connection on a $C^\infty(X)$-module $Y(X)$, and
\textit{vice versa}.

Let now $P$ be an arbitrary $C^\infty(X)$-module. One can
reformulate Definitions \ref{+181} and \ref{1016} of a connection
on $P$ as follows.

\begin{definition} \label{t55} \mar{t55}
A connection on a $C^\infty(X)$-module $P$ is a
$C^\infty(X)$-module morphism
\be
\nabla: P\to \cO^1(X)\ot P,
\ee
which satisfies the Leibniz rule
\be
\nabla(fp)=df\ot p +f\nabla(p), \qquad f\in C^\infty(X), \qquad
p\in P.
\ee
\end{definition}

\begin{definition} \label{t57} \mar{t57}
A connection on a $C^\infty(X)$-module $P$ associates to any
vector field $\tau\in\cT_1(X)$ on $X$ a first-order differential
operator $\nabla_\tau$ on $P$ which obeys the Leibniz rule
\be
\nabla_\tau(fp)=(\tau\rfloor df)p +f\nabla_\tau p.
\ee
\end{definition}

Since $\cO^1(X)=\cT_1(X)^*$, Definitions \ref{t55} and \ref{t57}
are equivalent.

Let us note that a connection on an arbitrary $C^\infty(X)$-module
need not exist, unless it is a projective or locally free module
(Proposition \ref{w715}).

A curvature of a connection $\nabla$ in Definitions \ref{t55} and
\ref{t57} is defined as a zero-order differential operator
\mar{t59}\beq
R(\tau,\tau')=[\nabla_\tau,\nabla_{\tau'}]-\nabla_{[\tau,\tau']}
\label{t59}
\eeq
on a module $P$ for all vector fields $\tau,\tau'\in\cT_1(X)$ on
$X$.

In accordance with Proposition \ref{+219}, we come to the
following relation between connections on $C^\infty(X)$-modules
and sheaves of $C^\infty_X$-modules (Proposition \ref{w715}).

Let $X$ be a manifold and $C^\infty_X$ a sheaf of smooth real
functions on $X$. A sheaf $\gd C^\infty_X$ of its derivations is
isomorphic to a sheaf of vector fields on a manifold $X$
(corollary (iii) of Theorem \ref{sp60}). It follows that:

$\bullet$ there is the restriction morphism $\gd(C^\infty(U))\to
\gd(C^\infty(V))$ for any open sets $V\subset U$,

$\bullet$ $\gd C^\infty_X$ is a locally free sheaf of
$C^\infty_X$-modules of finite rank (Theorem \ref{nn1}),

$\bullet$ the sheaves $\gd C^\infty_X$ and $\cO^1_X$ are mutually
dual (Theorem \ref{nn228}).

Let $\gP$ be a locally free sheaf of $C^\infty_X$-modules. In this
case, $\hm(\gP,\cO^1_X\ot \gP)$ is a locally free sheaf of
$C^\infty_X$-modules. It is fine and acyclic. Its  cohomology
group
\be
H^1(X;\hm(\gP,\cO^1_X\ot \gP))
\ee
vanishes, and the exact sequence
\mar{+2}\beq
0\to \cO^1_X\ot \gP\to (C^\infty_X\oplus \cO^1_X)\ot \gP\to \gP\to
0 \label{+2}
\eeq
admits a splitting. This proves the following.

\begin{proposition} \label{w715} \mar{w715}
Any locally free sheaf of $C^\infty_X$-modules on a manifold $X$
admits a connection and, in accordance with Proposition
\ref{+219}, any locally free $C^\infty(X)$-module does well.
\end{proposition}

\begin{example} \label{nn265} \mar{nn265}
Let $Y\to X$ be a vector bundle and $Y_X$ a sheaf of germs of
sections of $Y\to X$ (Example \ref{t1}). Every linear connection
$\G$ on $Y\to X$ defines a connection on the structure module
$Y(X)$ of sections of $Y\to X$ such that the restriction $\G|_U$
is a connection on a module $Y(U)$ for any open subset $U\subset
X$ (corollary (iv) of Theorem \ref{sp60}). Then we have a
connection on the structure sheaf $Y_X$ of germs of sections of
$Y\to X$. Conversely, a connection on the structure sheaf $Y_X$
defines a connection on a module $Y(X)$ and, consequently, a
connection on a vector bundle $Y\to X$.
\end{example}

In conclusion, let us consider a sheaf $S$ of commutative
$C^\infty_X$-rings on a manifold $X$. Basing on Definition
\ref{mos088}, we come to the following notion of a connection on a
sheaf $S$ of commutative $C^\infty_X$-rings.

\begin{definition} \label{+7} \mar{+7}
Any morphism
\be
\gd C^\infty_X \ni\tau \to \nabla_\tau\in \gd S,
\ee
which is a connection on $S$ as a sheaf of $C^\infty_X$-modules,
is called a connection on the sheaf $S$ of rings.
\end{definition}

Its curvature is given by the expression
\be
R(\tau,\tau')=[\nabla_\tau,\nabla_{\tau'}]-\nabla_{[\tau,\tau']},
\ee
similar to the expression (\ref{+100}) for a curvature of a
connection on modules.

\section{Differential calculus over $\mathbb Z_2$-graded \\ commutative rings}

This section addresses the differential calculus over
Grassmann-graded rings (Definition \ref{nn106}) as particular
$\mathbb N$-graded commutative rings. This also is a special case
of $\mathbb Z_2$-graded commutative rings which are Grassmann
algebras (Definition \ref{nn114}).

\subsection{$\mathbb Z_2$-Graded algebraic calculus}

Let us summarize the relevant notions of the $\mathbb Z_2$-graded
algebraic calculus \cite{bart,book05,sard09a,book12}.

Recall that the symbol $[.]$ stands for the $\mathbb Z_2$-degree.

\begin{definition} \label{nn110} \mar{nn110}
Let $\cK$ be a commutative ring without a divisor of zero. A
$\cK$-module $Q$ is called $\mathbb Z_2$-graded if it is endowed
with a grading automorphism $\g$, $\g^2=\id$. A $\mathbb
Z_2$-graded module falls into a direct sum of modules $Q=Q_0
\oplus Q_1$ such that
\be
\g(q)=(-1)^{[q]}q, \qquad q\in Q_{[q]}.
\ee
One calls $Q_0$ and $Q_1$ the even and odd parts of $Q$,
respectively.
\end{definition}

\begin{example} \label{nn270} \mar{nn270}
Any $\mathbb N$-graded $\cK$-module $P$ in Definition \ref{nn90})
is a $\mathbb Z_2$-graded $\cK$-module where
\be
P_0= \op\oplus_{i\in\mathbb N} P^{2i}, \qquad P_1=
\op\oplus_{i\in\mathbb N} P^{2i+1}.
\ee
\end{example}

A $\mathbb Z_2$-graded $\cK$-module is said to be free if it has a
basis composed by graded-homogeneous elements.

In particular, by a real $\mathbb Z_2$-graded vector space
$B=B_0\oplus B_1$ is meant a graded $\mathbb R$-module. A real
$\mathbb Z_2$-graded vector space is said to be
$(n,m)$-dimensional if $B_0=\mathbb R^n$ and $B_1=\mathbb R^m$.

\begin{definition} \label{nn111} \mar{nn111}
A $\cK$-ring $\cA$ is called $\mathbb Z_2$-graded if it is a
$\mathbb Z_2$-graded $\cK$-module such that
\be
[aa']=([a]+[a'])\mathrm{mod}\,2,
\ee
where $a$ and $a'$ are graded-homogeneous elements of $\cA$. In
particular, $[\bb]=0$.
\end{definition}

Its even part $\cA_0$ is a $\cK$-ring and the odd one $\cA_1$ is
an $\cA_0$-module.

\begin{example} \label{nn271} \mar{nn271}
Any $\mathbb N$-graded ring in Definition \ref{nn40}, regarded as
$\mathbb Z_2$-graded module (Example \ref{nn270}) is a $\mathbb
Z_2$-graded $\cK$-ring. The converse need not be true, unless a
$\mathbb Z_2$-graded $\cK$-ring is a Grassmann algebra (Definition
\ref{nn114}).
\end{example}

\begin{definition} \label{nn112} \mar{nn112}
A $\mathbb Z_2$-graded ring $\cA$ is called graded commutative if
\be
aa'=(-1)^{[a][a']}a'a, \qquad a,a'\in\cA.
\ee
\end{definition}

In particular, a commutative ring is the even $\mathbb Z_2$-graded
commutative one $\cA=\cA_0$.

Every $\mathbb N$-graded commutative $\cK$-ring $\Om^*$
(Definition \ref{nn44}) possesses the associated $\mathbb
Z_2$-graded commutative structure $\Om=\Om_0\oplus\Om_1$
(\ref{nn104}).

\begin{example} \label{nn272} \mar{nn272}
Clifford algebras exemplify $\mathbb Z_2$-graded rings which are
not $\mathbb Z_2$-graded commutative. A Clifford $\cK$-algebra is
defined to be a finitely generated $\mathbb Z_2$-graded $\cK$-ring
which admits a generating basis $\{e^A\}$ of odd elements so that
\mar{nn273}\beq
e^Ae^B+e^Be^A= \dl^{AB}\bb. \label{nn273}
\eeq
It also is an $\mathbb N$-graded $\cK$-module, but not an $\mathbb
N$-graded ring.
\end{example}

Given a $\mathbb Z_2$-graded ring $\cA$, a left $\mathbb
Z_2$-graded $\cA$-module $Q$ is defined as a left $\cA$-module
which is a $\mathbb Z_2$-graded $\cK$-module such that
\be
[aq]=([a]+[q])\mathrm{mod}\,2.
\ee
Similarly, right graded $\cA$-modules and graded
$(\cA-\cA)$-bimodules are defined. If $\cA$ is a $\mathbb
Z_2$-graded commutative ring, a $\mathbb Z_2$-graded left or right
$\cA$-module $Q$ can be provided with a $\mathbb Z_2$-graded
commutative $\cA$-bimodule structure by letting
\be
qa = (-1)^{[a][q]}aq, \qquad a\in\cA, \qquad q\in Q.
\ee
Therefore, unless otherwise stated (Section 5.2), any $\mathbb
Z_2$-graded $\cA$-module over a $\mathbb Z_2$-graded commutative
ring $\cA$ is a $\mathbb Z_2$-graded commutative $\cA$-bimodule
which is called the $\cA$-module if there is no danger of
confusion.

Given a $\mathbb Z_2$-graded commutative ring $\cA$, the following
are standard constructions of new $\mathbb Z_2$-graded modules
from the old ones.

$\bullet$ A direct sum of $\mathbb Z_2$-graded modules and a
$\mathbb Z_2$-graded factor module are defined just as those of
modules over a commutative ring.

$\bullet$ A tensor product $P\ot Q$ of $\mathbb Z_2$-graded
$\cA$-modules $P$ and $Q$ is their tensor product as $\cA$-modules
such that
\be
&& [p\ot q]=([p]+ [q])\mod 2, \qquad p\in P, \qquad q\in Q, \\
&&  ap\ot q=(-1)^{[p][a]}pa\ot q= (-1)^{[p][a]}p\ot aq,  \qquad a\in\cA.
\ee
In particular, the tensor algebra $\ot P$ of a $\mathbb
Z_2$-graded $\cA$-module $P$ is defined just as that
(\ref{spr620}) of a module over a commutative ring. Its quotient
$\w P$ with respect to the ideal generated by elements
\be
p\ot p' + (-1)^{[p][p']}p'\ot p, \qquad p,p'\in P,
\ee
is the bigraded exterior algebra of a $\mathbb Z_2$-graded module
$P$ with respect to the graded exterior product
\mar{nn202}\beq
p\w p' =- (-1)^{[p][p']}p'\w p. \label{nn202}
\eeq

$\bullet$ A morphism $\Phi:P\to Q$ of $\mathbb Z_2$-graded
$\cA$-modules seen as $\cK$-modules is said to be even morphism
(resp. odd morphism) if $\Phi$ preserves (resp. change) the
$\mathbb Z_2$-parity of all homogeneous elements of $P$ and the
relations
\mar{nn212}\beq
\Phi(ap)=(-1)^{[\Phi][a]}a\Phi(p), \qquad p\in P, \qquad a\in\cA,
\label{nn212}
\eeq
hold. A morphism $\Phi:P\to Q$ of $\mathbb Z_2$-graded
$\cA$-modules as the $\cK$-ones is called a graded $\cA$-module
morphism if it is represented by a sum of even and odd morphisms.
Therefore, a set $\hm_\cA(P,Q)$ of graded morphisms of a $\mathbb
Z_2$-graded $\cA$-module $P$ to a $\mathbb Z_2$-graded
$\cA$-module $Q$ is a $\mathbb Z_2$-graded $\cA$-module. A
$\mathbb Z_2$-graded $\cA$-module $P^*=\hm_\cA(P,\cA)$ is called
the dual of a $\mathbb Z_2$-graded $\cA$-module $P$.

\begin{remark} \label{nn201} \mar{nn201}
Let $\cA$ be a $\mathbb Z_2$-graded commutative ring. A $\mathbb
Z_2$-graded  $\cA$-module $\cG=\cG_0\oplus \cG_1$ is called the
Lie $\cA$-superalgebra if it is an $\cA$-algebra whose product
$[.,.]$, called the Lie superbracket, obeys the rules
\be
&& [\ve,\ve']=-(-1)^{[\ve][\ve']}[\ve',\ve],\\
&& (-1)^{[\ve][\ve'']}[\ve,[\ve',\ve'']]
+(-1)^{[\ve'][\ve]}[\ve',[\ve'',\ve]] +
(-1)^{[\ve''][\ve']}[\ve'',[\ve,\ve']] =0.
\ee
Even and odd parts of a Lie superalgebra $\cG$ satisfy
supercommutation relations
\mar{ss50}\beq
[{\cG_0},{\cG_0}]\subset \cG_0, \qquad [{\cG_0},{\cG_1}]\subset
\cG_1, \qquad [{\cG_1},{\cG_1}]\subset \cG_1. \label{ss50}
\eeq
In particular, an even part $\cG_0$ of a Lie $\cA$-superalgebra
$\cG$ is a Lie $\cA_0$-algebra. Given an $\cA$-superalgebra, a
$\mathbb Z_2$-graded $\cA$-module $P$ is called a $\cG$-module if
it is provided with an $\cA$-bilinear map
\be
&& \cG\times P\ni (\ve,p)\to \ve p\in P, \qquad [\ve
p]=([\ve]+[p])\mathrm{mod}\,2,\\
&& [\ve,\ve']p=(\ve\circ\ve'-(-1)^{[\ve][\ve']}\ve'\circ\ve)p.
\ee
\end{remark}

We mainly restrict our consideration to the following particular
class of $\mathbb Z_2$-graded commutative rings.

\begin{definition} \label{nn114} \mar{nn114}
A $\mathbb Z_2$-graded commutative $\cK$-ring $\cA$ is said to be
the Grassmann algebra if it is a free $\cK$-module of finite rank
so that
\mar{nn115}\beq
\cA_0=\cK\oplus  (\cA_1)^2. \label{nn115}
\eeq
\end{definition}

It follows from the expression (\ref{nn115}) that a $\cK$-module
$\cA$ admits a decomposition
\mar{+11}\beq
\cA=\cK\oplus R, \qquad R =\cA_1 \oplus (\cA_1)^2, \label{+11}
\eeq
where $R$ is the ideal of nilpotents of a ring $\cA$. The
corresponding surjections
\mar{nn204}\beq
\si:\cA\to \cK, \qquad s:\cA\to R \label{nn204}
\eeq
are called the body and soul maps, respectively. Automorphisms of
a Grassmann algebra preserve its ideal $R$ of nilpotents and the
splittings (\ref{nn115}) and (\ref{+11}), but need not the odd
sector $\cA_1$.

A Grassmann-graded ring $\Om^*$ in Definition \ref{nn106}, seen as
a $\mathbb Z_2$-graded commutative ring, exemplifies a Grassmann
algebra. Conversely, any Grassmann algebra $\cA$ admits an
associative $\mathbb N$-graded structure of a Grassmann-graded
ring $\cA^*$ by a choice of its minimal generating $\cK$-module
$\cA^1$.

\begin{theorem} \label{nn288} \mar{nn288}
If $\cK$ is a field, all $\mathbb N$-graded structures $\cA^*$ of
a Grassmann algebra $\cA^*$ are isomorphic by means of
automorphisms of a $\cK$-ring $\cA$ in accordance with Theorem
\ref{nn205}.
\end{theorem}

In particular, we come to the following.

\begin{theorem} \label{nn117} \mar{nn117}
Given a Grassmann algebra $\cA$ over a field $\cK$ and an
associated Grassmann-graded ring $\cA^*$, there exists a
finite-dimensional vector space $W$ over $\cK$ so that $\cA$ and
$\cA^*$ are isomorphic to the exterior algebra $\w W$ of $W$
(Example \ref{ws40}) seen as a Grassmann-graded ring $\cA^*$
generated by $\cA^1=W$.
\end{theorem}

Let us note that, by automorphisms of a Grassmann-graded ring
$\cA^*$ are meant automorphisms of a $\cK$-ring $\cA$ which
preserve its $\mathbb N$-gradation $\cA^*$ (Section 6).
Accordingly, automorphisms of a Grassmann algebra $\cA$ are
automorphisms of a $\cK$-ring $\cA$ which preserve its $\mathbb
Z_2$-gradation $\cA^*$. However, there exist automorphisms of a
$\cK$-ring $\cA$ which do not satisfy this condition as follows.

Given a generating basis $\{c^i\}$ for a $\cK$-module $\cA^1$,
elements of a Grassmann-graded ring $\cA^*$  take a form
\mar{z784}\beq
a=\op\sum_{k=0,1,\ldots} \op\sum_{(i_1\cdots i_k)}a_{i_1\cdots
i_k}c^{i_1}\cdots c^{i_k}, \label{z784}
\eeq
where the second sum runs through all the tuples $(i_1\cdots i_k)$
such that no two of them are permutations of each other. We agree
to call $\{c^i\}$ the generating basis for a Grassmann algebra
$\cA$ which brings it into a Grassmann-graded ring $\cA^*$.

Given a generating basis $\{c^i\}$ for a Grassmann algebra $\cA$,
any its $\cK$-ring automorphisms are compositions of automorphisms
\mar{nn127}\beq
c^i\to c'^i=\rho^i_jc^j + b^i, \label{nn127}
\eeq
where $\rho$ is an automorphism of $\cK$-module $\cA^1$ and $b^i$
are odd elements of $\cA^{>2}$, and automorphisms
\mar{nn280}\beq
c^i\to c'^i=c^i(\bb + \kappa), \qquad \kappa\in \cA_1.
\label{nn280}
\eeq
Automorphisms (\ref{nn127}) where $b^i= 0$ are automorphisms of a
Grassmann-graded ring $A^*$. If $b^i\neq 0$, the automorphism
(\ref{nn127}) preserve a $\mathbb Z_2$-graded structure  of $\cA$,
does not keep its $\mathbb N$-graded structure $\cA^*$ of, and
yields a different $\mathbb N$-graded structure $\cA'^*$ where
$\{c'^i\}$ (\ref{nn127}) is a basis for $\cA'^1$ and the
generating basis for $\cA'^*$. Automorphisms (\ref{nn280})
preserve an even sector $\cA_0$ of $A$, but not the odd one
$\cA_1$. However, it follows from Theorem \ref{nn288} that, since
$\mathbb N$-graded and $\mathbb Z_2$-graded structures of a
Grassmann algebra are associated, their $\mathbb Z_2$-graded
structures are isomorphic if $\cK$ is a field.

\subsection{$\mathbb Z_2$-Graded differential calculus}

The differential calculus on $\mathbb Z_2$-graded modules over
$\mathbb Z_2$-graded commutative rings is defined similarly to
that over commutative rings \cite{book05,sard09a,book12}, but
different from the differential calculus over non-commutative
rings in Section 9 (Remark \ref{nn121}).

Let $\cK$ be a commutative ring without a divisor of zero and
$\cA$ a $\mathbb Z_2$-graded commutative $\cK$-ring. Let $P$ and
$Q$ be $\mathbb Z_2$-graded $\cA$-modules. A $\mathbb Z_2$-graded
$\cK$-module $\hm_\cK (P,Q)$ of $\mathbb Z_2$-graded $\cK$-module
homomorphisms $\Phi:P\to Q$ can be endowed with the two $\mathbb
Z_2$-graded $\cA$-module structures
\mar{ws11}\beq
(a\Phi)(p)= a\Phi(p),  \qquad  (\Phi\bll a)(p) = \Phi (a p),\qquad
a\in \cA, \quad p\in P, \label{ws11}
\eeq
called $\cA$- and $\cA^\bll$-module structures, respectively. Let
us put
\mar{nn120}\beq
\dl_a\Phi= a\Phi -(-1)^{[a][\Phi]}\Phi\bll a, \qquad a\in\cA.
\label{nn120}
\eeq

\begin{definition} \label{nn122} \mar{nn122}
An element $\Delta\in\hm_\cK(P,Q)$ is said to be the $Q$-valued
$\mathbb Z_2$-graded differential operator of order $s$ on $P$ if
\mar{nn300}\beq
\dl_{a_0}\circ\cdots\circ\dl_{a_s}\Delta=0 \label{nn300}
\eeq
for any tuple of $s+1$ elements $a_0,\ldots,a_s$ of $\cA$. A set
$\dif_s(P,Q)$ of these operators inherits the $\mathbb Z_2$-graded
$\cA$-module structures (\ref{ws11}).
\end{definition}

\begin{remark} \label{nn121} \mar{nn121}
It should be emphasized that, though a $\mathbb Z_2$-graded
commutative ring is a non-commutative ring, the map (\ref{nn120})
differs from the map (\ref{ws133a}), and therefore the
differential calculus over a $\mathbb Z_2$-graded commutative ring
is not the particular non-commutative differential calculus in
Section 9. Let us note that, though $\mathbb Z_2$-graded Clifford
algebras (Example \ref{nn272}) fail to be graded commutative, the
definition of differential operators over them also is based on
the formula (\ref{nn120}), but not (\ref{ws133a}). However, graded
derivations of a Clifford algebra do not form a free module over
it.
\end{remark}

For instance, zero-order $\mathbb Z_2$-graded differential
operators obey a condition
\be
\dl_a \Delta(p)=a\Delta(p)-(-1)^{[a][\Delta]}\Delta(ap)=0, \qquad
a\in\cA, \qquad p\in P,
\ee
i.e., they coincide with graded $\cA$-module morphisms $P\to Q$
(cf. Remark \ref{nn103}).

A first-order $\mathbb Z_2$-graded differential operator $\Delta$
satisfies a condition
\be
&& \dl_a\circ\dl_b\,\Delta(p)=
ab\Delta(p)- (-1)^{([b]+[\Delta])[a]}b\Delta(ap)-
(-1)^{[b][\Delta]}a\Delta(bp)+\\
&& \qquad (-1)^{[b][\Delta]+([\Delta]+[b])[a]}
=0, \qquad a,b\in\cA, \quad p\in P.
\ee

For instance, let $P=\cA$. Any zero-order $Q$-valued $\mathbb
Z_2$-graded differential operator $\Delta$ on $\cA$ is defined by
its value $\Delta(\bb)$. Then there is a graded $\cA$-module
isomorphism $\dif_0(\cA,Q)=Q$ via the association
\be
Q\ni q\to \Delta_q\in \dif_0(\cA,Q),
\ee
where $\Delta_q$ is given by the equality $\Delta_q(\bb)=q$. A
first-order $Q$-valued $\mathbb Z_2$-graded differential operator
$\Delta$ on $\cA$ fulfils a condition
\be
\Delta(ab)= \Delta(a)b+ (-1)^{[a][\Delta]}a\Delta(b)
-(-1)^{([b]+[a])[\Delta]} ab \Delta(\bb), \qquad  a,b\in\cA.
\ee
It is called the $Q$-valued $\mathbb Z_2$-graded derivation of
$\cA$ if $\Delta(\bb)=0$, i.e., the graded Leibniz rule
\be
\Delta(ab) = \Delta(a)b + (-1)^{[a][\Delta]}a\Delta(b), \qquad
a,b\in \cA,
\ee
holds (cf. (\ref{+a20})). One obtains at once that any first-order
$\mathbb Z_2$-graded differential operator on $\cA$ falls into a
sum
\be
\Delta(a)= \Delta(\bb)a +[\Delta(a)-\Delta(\bb)a]
\ee
of a zero-order graded differential operator $\Delta(\bb)a$ and a
graded derivation $\Delta(a)-\Delta(\bb)a$. If $\dr$ is a $\mathbb
Z_2$-graded derivation of $\cA$, then $a\dr$ is so for any $a\in
\cA$. Hence, $\mathbb Z_2$-graded derivations of $\cA$ constitute
a $\mathbb Z_2$-graded $\cA$-module $\gd(\cA,Q)$, called the
graded derivation module.

If $Q=\cA$, $\mathbb Z_2$-graded derivation of $\cA$ obeys the
graded Leibniz rule
\mar{nn124}\beq
\dr(ab) = \dr(a)b + (-1)^{[a][\dr]}a\dr(b), \qquad a,b\in \cA,
\label{nn124}
\eeq
(cf. (\ref{ws100'})). The derivation module $\gd\cA$ also is a Lie
superalgebra (Definition \ref{nn201}) over a commutative ring
$\cK$ with respect to the superbracket
\mar{ws14}\beq
[u,u']=u\circ u' - (-1)^{[u][u']}u'\circ u, \qquad u,u'\in \cA.
\label{ws14}
\eeq
We have an $\cA$-module decomposition
\be
\dif_1(\cA) = \cA \oplus\gd\cA.
\ee

Graded differential operators on a $\mathbb Z_2$-graded
commutative $\cK$-ring $\cA$ form a direct system of $\mathbb
Z_2$-graded $(\cA-\cA^\bullet)$-modules.
\mar{nn136}\beq
\dif_0(\cA)\ar^\mathrm{in} \dif_1(\cA)\cdots
\ar^\mathrm{in}\dif_r(\cA)\ar\cdots \label{nn136}
\eeq
Its direct limit $\dif_\infty(\cA)$ (Remark \ref{nn3}) is a
$\mathbb Z_2$-graded module of all $\mathbb Z_2$-graded
differential operators on $\cA$.

Since the graded derivation module $\gd\cA$ is a Lie
$\cK$-superalgebra, let us consider the Chevalley--Eilenberg
complex $C^*[\gd\cA;\cA]$ where a $\mathbb Z_2$-graded commutative
ring $\cA$ is a regarded as a $\gd\cA$-module \cite{fuks,book12}.
It is a complex
\mar{ws85}\beq
0\to \cA\ar^d C^1[\gd\cA;\cA]\ar^d \cdots
C^k[\gd\cA;\cA]\ar^d\cdots \label{ws85}
\eeq
where
\be
C^k[\gd\cA;\cA]=\hm_\cK(\op\w^k \gd\cA,\cA)
\ee
are $\gd\cA$-modules of $\cK$-linear graded morphisms of the
graded exterior products $\op\w^k \gd\cA$ of a $\mathbb
Z_2$-graded $\cK$-module $\gd\cA$ to $\cA$. Let us bring
homogeneous elements of $\op\w^k \gd\cA$ into a form
\be
\ve_1\w\cdots\ve_r\w\e_{r+1}\w\cdots\w \e_k, \qquad
\ve_i\in\gd\cA_0, \quad \e_j\in\gd\cA_1.
\ee
Then the Chevalley--Eilenberg coboundary operator $d$ of the
complex (\ref{ws85}) is given by the expression
\ben
&& dc(\ve_1\w\cdots\w\ve_r\w\e_1\w\cdots\w\e_s)=
\label{ws86}\\
&&\op\sum_{i=1}^r (-1)^{i-1}\ve_i
c(\ve_1\w\cdots\wh\ve_i\cdots\w\ve_r\w\e_1\w\cdots\e_s)+
\nonumber \\
&& \op\sum_{j=1}^s (-1)^r\ve_i
c(\ve_1\w\cdots\w\ve_r\w\e_1\w\cdots\wh\e_j\cdots\w\e_s)
+\nonumber\\
&& \op\sum_{1\leq i<j\leq r} (-1)^{i+j}
c([\ve_i,\ve_j]\w\ve_1\w\cdots\wh\ve_i\cdots\wh\ve_j
\cdots\w\ve_r\w\e_1\w\cdots\w\e_s)+\nonumber\\
&&\op\sum_{1\leq i<j\leq s} c([\e_i,\e_j]\w\ve_1\w\cdots\w
\ve_r\w\e_1\w\cdots
\wh\e_i\cdots\wh\e_j\cdots\w\e_s)+\nonumber\\
&& \op\sum_{1\leq i<r,1\leq j\leq s} (-1)^{i+r+1}
c([\ve_i,\e_j]\w\ve_1\w\cdots\wh\ve_i\cdots\w\ve_r\w
\e_1\w\cdots\wh\e_j\cdots\w\e_s),\nonumber
\een
where the caret $\,\wh{}\,$ denotes omission. This operator is
called the graded Chevalley--Eilenberg coboundary operator.

Let us consider the extended Chevalley--Eilenberg complex
\be
0\to \cK\ar^\mathrm{in}C^*[\gd\cA;\cA].
\ee
It is easily justified that this complex contains a subcomplex
$\cO^*[\gd\cA]$ of $\cA$-linear graded morphisms. The $\mathbb
N$-graded module $\cO^*[\gd\cA]$ is provided with the structure of
a bigraded $\cA$-algebra  with respect to the graded exterior
product
\ben
&& \f\w\f'(u_1,...,u_{r+s})= \label{ws103'}\\
&& \qquad \op\sum_{i_1<\cdots<i_r;j_1<\cdots<j_s}
\mathrm{Sgn}^{i_1\cdots i_rj_1\cdots j_s}_{1\cdots r+s}
\f(u_{i_1},\ldots,
u_{i_r}) \f'(u_{j_1},\ldots,u_{j_s}), \nonumber \\
&& \f\in \cO^r[\gd\cA], \qquad \f'\in \cO^s[\gd\cA], \qquad u_k\in \gd\cA,
\nonumber
\een
where $u_1,\ldots, u_{r+s}$ are graded-homogeneous elements of
$\gd\cA$ and
\be
u_1\w\cdots \w u_{r+s}= \mathrm{Sgn}^{i_1\cdots i_rj_1\cdots
j_s}_{1\cdots r+s} u_{i_1}\w\cdots\w u_{i_r}\w u_{j_1}\w\cdots\w
u_{j_s}.
\ee
The graded Chevalley--Eilenberg coboundary operator $d$
(\ref{ws86}) and the graded exterior product $\w$ (\ref{ws103'})
bring $\cO^*[\gd\cA]$ into a differential bigraded ring whose
elements obey relations
\ben
&& \f\w \f'=(-1)^{|\f||\f'|+[\f][\f']}\f'\w\f, \label{ws45} \\
&& d(\f\w\f')= d\f\w\f' +(-1)^{|\f|}\f\w d\f'. \label{ws44}
\een
It is called the graded Chevalley--Eilenberg differential calculus
over a $\mathbb Z_2$-graded commutative $\cK$-ring $\cA$. In
particular, we have
\beq
\cO^1[\gd\cA]=\hm_\cA(\gd\cA,\cA)=\gd\cA^*. \label{ws47}
\eeq
One can extend this duality relation to the graded interior
product  of $u\in\gd\cA$ with any element $\f\in \cO^*[\gd\cA]$ by
the rules
\ben
&& u\rfloor(bda) =(-1)^{[u][b]}u(a),\qquad a,b \in\cA, \nonumber\\
&& u\rfloor(\f\w\f')=
(u\rfloor\f)\w\f'+(-1)^{|\f|+[\f][u]}\f\w(u\rfloor\f').
\label{ws46}
\een
As a consequence, any $\mathbb Z_2$-graded derivation $u\in\gd\cA$
of $\cA$ yields a graded derivation
\ben
&& \bL_u\f= u\rfloor d\f + d(u\rfloor\f), \qquad \f\in\cO^*, \qquad
u\in\gd\cA, \label{+117} \\
&& \bL_u(\f\w\f')=\bL_u(\f)\w\f' + (-1)^{[u][\f]}\f\w\bL_u(\f'), \nonumber
\een
termed the graded Lie derivative of a differential bigraded ring
$\cO^*[\gd\cA]$.

The minimal graded Chevalley--Eilenberg differential calculus
$\cO^*\cA\subset \cO^*[\gd\cA]$  over a $\mathbb Z_2$-graded
commutative ring $\cA$ consists of monomials $a_0da_1\w\cdots\w
da_k$, $a_i\in\cA$. The corresponding complex
\beq
0\to\cK\ar \cA\ar^d\cO^1\cA\ar^d \cdots  \cO^k\cA\ar^d \cdots
\label{t100}
\eeq
is called the bigraded de Rham complex  of a $\mathbb Z_2$-graded
commutative $\cK$-ring $\cA$.

Let us note that, if $\cA$ is a commutative ring, $\mathbb
Z_2$-graded differential operators and the graded
Chevalley--Eilenberg differential calculus defined above restart
the familiar commutative ones in Sections 2.2 and 2.4.

\begin{example} \label{nn133} \mar{nn133}
Let $\cA$ be a Grassmann algebra provided with an odd generating
basis $\{c^i\}$. Its $\mathbb Z_2$-graded derivations are defined
in full by their action on the generating elements $c^i$. Let us
consider odd derivations
\mar{nn130}\beq
\dr_i(c^j)=\dl^j_i, \qquad \dr_i\circ\dr_j=-\dr_j\circ\dr_i.
\label{nn130}
\eeq
Then any $\mathbb Z_2$-graded derivation of $\cA$ takes a form
\mar{nn131}\beq
u=u^i\dr_i, \qquad u_i\in\cA. \label{nn131}
\eeq
Graded derivations (\ref{nn131}) constitute the free $\mathbb
Z_2$-graded $\cA$-module $\gd\cA$ of finite rank. It also is a
finite dimensional Lie superalgebra over $\cK$ with respect to the
superbracket (\ref{ws14}). Any $\mathbb Z_2$-graded differential
operator on a Grassmann algebra $\cA$ is a composition of $\mathbb
Z_2$-graded derivations.
\end{example}

\subsection{$\mathbb Z_2$-Graded manifolds}

As was mentioned in Remark \ref{nn245}, the notion of a local ring
is extended to the $\mathbb Z_2$-graded ones (Definition
\ref{nn251}), and formalism of $\mathbb Z_2$-graded commutative
local-ringed spaces can be developed. Grassmann-graded rings in
Definition \ref{nn106} and Grassmann algebras in Definition
\ref{nn114} are local. Their maximal ideals consist of all
nilpotent elements. We restrict our consideration to real
Grassmann algebras $\La$. In this case Theorem \ref{nn1} remains
true, and we follow the notion of local-ringed spaces in Section 3
\cite{bart,book05,sard09a,ijgmmp13}.

A $\mathbb Z_2$-graded manifold  of dimension $(n,m)$ is
 defined as a local-ringed space $(Z,\gA)$
(Definition \ref{nn250}) where $Z$ is an $n$-dimensional smooth
manifold, and $\gA=\gA_0\oplus\gA_1$ is a sheaf of real Grassmann
algebras such that:

$\bullet$ there is the exact sequence of sheaves
\mar{cmp140}\beq
0\to \cR \to\gA \op\to^\si C^\infty_Z\to 0, \qquad
\cR=\gA_1+(\gA_1)^2,\label{cmp140}
\eeq
where $C^\infty_Z$ is the sheaf of smooth real functions on $Z$;

$\bullet$ $\cR/\cR^2$ is a locally free sheaf of
$C^\infty_Z$-modules of finite rank (with respect to pointwise
operations), and the sheaf $\gA$ is locally isomorphic to the
exterior product $\w_{C^\infty_Z}(\cR/\cR^2)$.

A sheaf $\gA$ is called the structure sheaf
 of a $\mathbb Z_2$-graded manifold
$(Z,\gA)$, and a manifold $Z$ is said to be the body
 of $(Z,\gA)$. Sections of the
sheaf $\gA$ are termed graded functions
 on a $\mathbb Z_2$-graded manifold $(Z,\gA)$. They make
up a $\mathbb Z_2$-graded commutative $C^\infty(Z)$-ring $\gA(Z)$
 called the structure ring of $(Z,\gA)$.

By virtue of the well-known Batchelor theorem
 \cite{bart,book12}, $\mathbb Z_2$-graded
manifolds possess the following structure.

\begin{theorem} \label{lmp1a} \mar{lmp1a}
Let $(Z,\gA)$ be a $\mathbb Z_2$-graded manifold. There exists a
vector bundle $E\to Z$ with an $m$-dimensional typical fibre $V$
such that the structure sheaf $\gA$ of $(Z,\gA)$ as a sheaf in
real rings is isomorphic to the structure sheaf $\gA_E=\w E^*_Z$
of germs of sections of the exterior bundle
\mar{nn6}\beq
\w E^*=(Z\times\mathbb R)\op\oplus_Z E^* \op\oplus_Z \w E^*
\op\oplus_Z \op\w^2 E^*\cdots,\label{nn6}
\eeq
whose typical fibre is the Grassmann algebra $\La=\w V^*$ in
Theorem \ref{nn117}.
\end{theorem}

It should be emphasized that Batchelor's isomorphism in Theorem
\ref{lmp1a} fails to be canonical. In applications, it however is
fixed from the beginning. Therefore, we restrict our consideration
to $\mathbb Z_2$-graded manifolds $(Z,\gA_E)$ whose structure
sheaf is the sheaf of germs of sections of some exterior bundle
$\w E^*$.

\begin{definition} \label{nn173} \mar{nn173}
We agree to call $(Z,\gA_E)$  the simple graded manifold modelled
over a vector bundle $E\to Z$,
 called its characteristic
vector bundle.
\end{definition}

Accordingly, the structure ring $\cA_E$  of a simple graded
manifold $(Z,\gA_E)$ is the structure module
\beq
\cA_E=\gA_E(Z)=\w E^*(Z) \label{33f1}
\eeq
of sections of the exterior bundle $\w E^*$. Automorphisms of a
simple graded manifold $(Z,\gA_E)$ are restricted to those induced
by automorphisms of its characteristic vector bundles $E\to Z$.

\begin{remark} \label{nn281} \mar{nn281}
In fact, the structure sheaf $\gA_E$ of a simple $\mathbb
Z_2$-graded manifold in Definition \ref{nn173} is a sheaf in
Grassmann-graded rings $\La^*$ whose $\mathbb N$-graded structure
is fixed. Therefore, it is an $\mathbb N$-graded manifold
(Definition \ref{nn306}).
\end{remark}

Combining Batchelor Theorem \ref{lmp1a} and classical Serre--Swan
Theorem \ref{sp60}, we come to the following Serre--Swan theorem
for $\mathbb Z_2$-graded manifolds \cite{jmp05a,book05}.

\begin{theorem} \label{vv0} \mar{vv0}
Let $Z$ be a smooth manifold. A $\mathbb Z_2$-graded commutative
$C^\infty(Z)$-ring $\cA$ is isomorphic to the structure ring of a
$\mathbb Z_2$-graded manifold with a body $Z$ iff it is the
exterior algebra of some projective $C^\infty(Z)$-module of finite
rank.
\end{theorem}

\begin{proof}  By virtue of the Batchelor theorem,
any $\mathbb Z_2$-graded manifold is isomorphic to a simple graded
manifold $(Z,\gA_E)$ modelled over some vector bundle $E\to Z$.
Its structure ring $\cA_E$ (\ref{33f1}) of graded functions
consists of sections of the exterior bundle $\w E^*$. The
classical Serre--Swan theorem states that a $C^\infty(Z)$-module
is isomorphic to the module of sections of a smooth vector bundle
over $Z$ iff it is a projective module of finite rank.
\end{proof}

\begin{remark} \label{triv}
One can treat a local-ringed space $(Z,\gA_0=C_Z^\infty)$ as a
trivial even $\mathbb Z_2$-graded manifold. It is a simple graded
manifold whose characteristic bundle is $E=Z\times\{0\}$. Its
structure module is a ring $C^\infty(Z)$ of smooth real functions
on $Z$.
\end{remark}

Given a simple graded manifold $(Z,\gA_E)$, every trivialization
chart $(U; z^A,y^a)$ of a vector bundle $E\to Z$ yields a
splitting domain $(U; z^A,c^a)$ of $(Z,\gA_E)$. Graded functions
on such a chart are $\La$-valued functions
\mar{z785}\beq
f=\op\sum_{k=0}^m \frac1{k!}f_{a_1\ldots a_k}(z)c^{a_1}\cdots
c^{a_k}, \label{z785}
\eeq
where $f_{a_1\cdots a_k}(z)$ are smooth functions on $U$ and
$\{c^a\}$ is the  fibre basis for $E^*$. In particular, the sheaf
epimorphism $\si$ in (\ref{cmp140}) is induced by the body map of
$\La$. One calls $\{z^A,c^a\}$ the local basis for a graded
manifold $(Z,\gA_E)$ \cite{bart}.  Transition functions
$y'^a=\rho^a_b(z^A)y^b$ of bundle coordinates on $E\to Z$ induce
the corresponding transformation
\beq
c'^a=\rho^a_b(z^A)c^b \label{+6}
\eeq
(cf. (\ref{nn302})) of the associated local basis for a simple
graded manifold $(Z,\gA_E)$ and the according coordinate
transformation law of graded functions (\ref{z785}).

\begin{remark}  \label{33r62}
Strictly speaking, elements $c^a$ of the local basis for a simple
graded manifold are locally constant sections $c^a$ of $E^*\to X$
such that $y_b\circ c^a=\dl^a_b$. Therefore, graded functions are
locally represented by $\La$-valued functions (\ref{z785}), but
they are not $\La$-valued functions on a manifold $Z$ because of
the transformation law (\ref{+6}).
\end{remark}

Given a $\mathbb Z_2$-graded manifold $(Z,\gA)$, by the sheaf
$\gd\gA$ of graded derivations  of $\gA$ is meant a subsheaf of
endomorphisms of the structure sheaf $\gA$ such that any section
$u\in \gd\gA(U)$ of $\gd\gA$ over an open subset $U\subset Z$ is a
graded derivation of the real $\mathbb Z_2$-graded commutative
algebra $\gA(U)$, i.e., $u\in\gd(\gA(U))$. Conversely, one can
show that, given open sets $U'\subset U$, there is a surjection of
the graded derivation modules $\gd(\gA(U))\to \gd(\gA(U'))$
\cite{bart}. It follows that any graded derivation of a local
$\mathbb Z_2$-graded algebra $\gA(U)$ also is a local section over
$U$ of a sheaf $\gd\gA$. Global sections of $\gd\gA$ are  called
graded vector fields  on a $\mathbb Z_2$-graded manifold
$(Z,\gA)$. They make up a graded derivation module $\gd\gA(Z)$ of
a real $\mathbb Z_2$-graded commutative ring $\gA(Z)$. This module
is a real Lie superalgebra with respect to the superbracket
(\ref{ws14}).

A key point is that graded vector fields $u\in\gd\cA_E$ on a
simple graded manifold $(Z,\gA_E)$ can be represented by sections
of some vector bundle as follows \cite{book05,book00,book12}.

Due to the canonical splitting $VE= E\times E$, the vertical
tangent bundle $VE$ of $E\to Z$ can be provided with fibre bases
$\{\dr/\dr c^a\}$, which are the duals of bases $\{c^a\}$. Then
graded vector fields on a splitting domain $(U;z^A,c^a)$ of
$(Z,\gA_E)$ read
\beq
u= u^A\dr_A + u^a\frac{\dr}{\dr c^a}, \label{hn14}
\eeq
where $u^\la, u^a$ are local graded functions on $U$. In
particular,
\be
\frac{\dr}{\dr c^a}\circ\frac{\dr}{\dr c^b} =-\frac{\dr}{\dr
c^b}\circ\frac{\dr}{\dr c^a}, \qquad \dr_A\circ\frac{\dr}{\dr
c^a}=\frac{\dr}{\dr c^a}\circ \dr_A.
\ee
The graded derivations (\ref{hn14}) act on graded functions
$f\in\gA_E(U)$ (\ref{z785}) by a rule
\beq
u(f_{a\ldots b}c^a\cdots c^b)=u^A\dr_A(f_{a\ldots b})c^a\cdots c^b
+u^k f_{a\ldots b}\frac{\dr}{\dr c^k}\rfloor (c^a\cdots c^b).
\label{cmp50a}
\eeq
This rule implies a corresponding coordinate transformation law
\be
u'^A =u^A, \qquad u'^a=\rho^a_ju^j +u^A\dr_A(\rho^a_j)c^j
\ee
of graded vector fields. It follows that graded vector fields
(\ref{hn14}) can be represented by sections of a vector bundle
\mar{nn253}\beq
\cV_E= \w E^*\op\ot_E TE\to Z. \label{nn253}
\eeq
Thus, the graded derivation module $\gd\cA_E$ is isomorphic to the
structure module $\cV_E(Z)$ of global sections of the vector
bundle $\cV_E\to Z$ (\ref{nn253}).

Given the structure ring $\cA_E$ of graded functions on a simple
graded manifold $(Z,\gA_E)$ and the real Lie superalgebra
$\gd\cA_E$ of its graded derivations, let us consider the graded
Chevalley--Eilenberg differential calculus
\beq
\cS^*[E;Z]=\cO^*[\gd\cA_E] \label{33f21}
\eeq
over  $\cA_E$ where $\cS^0[E;Z]=\cA_E$.

\begin{theorem} \label{mm1} \mar{mm1}
Since a graded derivation module $\gd\cA_E$ is isomorphic to the
structure module of sections of a vector bundle $\cV_E\to Z$,
elements of $\cS^*[E;Z]$ are represented by sections of the
exterior bundle $\w\ol\cV_E$ of the $\cA_E$-dual
\mar{nn254}\beq
\ol\cV_E=\w E^*\op\ot_E T^*E\to Z \label{nn254}
\eeq
of $\cV_E$.
\end{theorem}

With respect to the dual fibre bases $\{dz^A\}$ for $T^*Z$ and
$\{dc^b\}$ for $E^*$, sections of $\ol\cV_E$ (\ref{nn254}) take a
coordinate form
\be
\f=\f_A dz^A + \f_adc^a,
\ee
together with transition functions
\be
\f'_a=\rho^{-1}{}_a^b\f_b, \qquad \f'_A=\f_A
+\rho^{-1}{}_a^b\dr_A(\rho^a_j)\f_bc^j.
\ee
The duality isomorphism $\cS^1[E;Z]=\gd\cA_E^*$ (\ref{ws47}) is
given by  the graded interior product
\be
u\rfloor \f=u^A\f_A + (-1)^{\nw{\f_a}}u^a\f_a.
\ee
Elements of $\cS^*[E;Z]$ are called graded exterior forms on a
graded manifold $(Z,\gA_E)$.

Seen as an $\cA_E$-algebra, the differential bigraded ring
$\cS^*[E;Z]$ (\ref{33f21}) on a splitting domain $(z^A,c^a)$ is
locally generated by the graded one-forms $dz^A$, $dc^i$ such that
\be
dz^A\w dc^i=-dc^i\w dz^A, \qquad dc^i\w dc^j= dc^j\w dc^i.
\ee
Accordingly, the graded Chevalley--Eilenberg coboundary operator
$d$ (\ref{ws86}), termed the graded exterior differential, reads
\be
d\f= dz^A \w \dr_A\f +dc^a\w \frac{\dr}{\dr c^a}\f,
\ee
where  derivatives $\dr_\la$, $\dr/\dr c^a$ act on coefficients of
graded exterior forms by the formula (\ref{cmp50a}), and they are
graded commutative with the graded exterior forms $dz^A$ and
$dc^a$. The formulas (\ref{ws45}) -- (\ref{+117}) hold.

\begin{theorem} \label{v62} \mar{v62}
The differential bigraded ring $\cS^*[E;Z]$ (\ref{33f21}) is a
minimal differential calculus over $\cA_E$, i.e., it is generated
by elements $df$, $f\in \cA_E$.
\end{theorem}

\begin{proof}
Since $\gd\cA_E=\cV_E(Z)$, it is a projective $C^\infty(Z)$- and
$\cA_E$-module of finite rank, and so is its $\cA_E$-dual
$\cS^1[E;Z]$ (Theorem \ref{nn228}). Hence, $\gd\cA_E$ is the
$\cA_E$-dual of $\cS^1[E;Z]$ and, consequently, $\cS^1[E;Z]$ is
generated by elements $df$, $f\in \cA_E$.
\end{proof}

The bigraded de Rham complex (\ref{t100}) of a minimal graded
Chevalley--Eilenberg differential calculus $\cS^*[E;Z]$ reads
\beq
0\to \mathbb R\to \cA_E \ar^d \cS^1[E;Z]\ar^d\cdots
\cS^k[E;Z]\ar^d\cdots. \label{+137}
\eeq
Its cohomology $H^*(\cA_E)$  is called the de Rham cohomology of a
simple graded manifold $(Z,\gA_E)$.

In particular, given a differential graded ring $\cO^*(Z)$ of
exterior forms on $Z$, there exist a canonical monomorphism
\beq
\cO^*(Z)\to \cS^*[E;Z] \label{uut}
\eeq
and the body epimorphism $\cS^*[E;Z]\to \cO^*(Z)$ which are
cochain morphisms of the de Rham complexes (\ref{+137}) and
(\ref{t37}).

\begin{theorem} \label{33t3}
The de Rham cohomology of a simple graded manifold $(Z,\gA_E)$
equals the de Rham cohomology of its body $Z$.
\end{theorem}

\begin{proof}
Let $\gA_E^k$ denote the sheaf of germs of graded $k$-forms on
$(Z,\gA_E)$. Its structure module is $\cS^k[E;Z]$. These sheaves
constitute a complex
\beq
0\to\mathbb R\ar \gA_E \ar^d \gA_E^1\ar^d\cdots
\gA_E^k\ar^d\cdots. \label{1033}
\eeq
Its members $\gA_E^k$ are sheaves of $C^\infty_Z$-modules on $Z$
and, consequently, are fine and acyclic. Furthermore, the
Poincar\'e lemma for graded exterior forms holds \cite{bart}. It
follows that the complex (\ref{1033}) is a fine resolution of the
constant sheaf $\mathbb R$ on a manifold $Z$.  Then, by virtue of
Theorem \ref{spr230}, there is an isomorphism
\beq
H^*(\cA_E)=H^*(Z;\mathbb R)=H^*_{\mathrm{DR}}(Z) \label{+136}
\eeq
of the cohomology $H^*(\cA_E)$ to the de Rham cohomology
$H^*_{\mathrm{DR}}(Z)$ of a smooth manifold $Z$.
\end{proof}

\begin{theorem} \label{33c1}
The cohomology isomorphism (\ref{+136}) is accompanied by the
cochain monomorphism (\ref{uut}). Hence, any closed graded
exterior form is decomposed into a sum $\f=\si +d\xi$ where $\si$
is a closed exterior form on $Z$.
\end{theorem}

\section{$\mathbb N$-graded manifolds}

Let $\cK$ be a commutative ring without a divisor of zero, and let
$\Om=\Om^*$ (\ref{nn41}) be an $\mathbb N$-graded $\cK$-ring
(Definition \ref{nn40}).

As it was observed above, an $\mathbb N$-graded ring seen as a
$\cK$-ring can admit different $\mathbb N$-graded structures. In
the following case, all these structures are isomorphic
\cite{bell}.

\begin{definition} \label{nn210} \mar{nn210}
An $\mathbb N$-graded $\cK$-ring $\cA^*$ is said to be finitely
generated in degree 1 if the following hold:

$\bullet$ $\cA^0=\cK$,

$\bullet$ $\cA^1$ is a free $\cK$-module of finite rank,

$\bullet$ $\cA^*$ is generated by $\Om^1$, namely, if $R$ is an
ideal generated by $\cA^1$, then there are $\cK$-module
isomorphism $\Om/R=\cK$, $R/R^2=\cA^1$.
\end{definition}

\begin{theorem} \label{nn205} \mar{nn205}
Let $\cK$ be a field, and let $\cA^*$ and $\La^*$ be $\mathbb
N$-graded $\cK$-rings finitely generated in degree 1. If they are
isomorphic as $\cK$-rings, there exists their graded isomorphism
$\Phi: \cA^*\to \La^*$ so that $\Phi(\cA^i)=\La^i$ for all $i\in
\mathbb N$.
\end{theorem}

Let us mention the following particular types of $\mathbb
N$-graded rings.

$\bullet$ A commutative ring $\cA$ is the $\mathbb N$-graded ring
$\cA^*$ where $\cA^0=\cA$ and $\cA^{>0}=0$.

$\bullet$ It is a particular commutative $\mathbb N$-graded ring
$\cA^*$ where $\al\cdot\bt=\bt\cdot \al$ for all elements
$\al,\bt\in \cA^*$ (cf. (\ref{nn92})).

$\bullet$ A polynomial $\cK$-ring $\cP[Q]$ of a $\cK$-module $Q$
in Example \ref{ws40} exemplifies a commutative $\mathbb N$-graded
ring. If $\cK$ is a field and $Q$ is a free $\cK$-module of finite
rank, a polynomial $\cK$-ring $\cP[Q]$ is finitely generated in
degree 1 by virtue of Definition \ref{nn210}. If $\cK$ is a field,
all $\mathbb N$-graded structures of this ring are mutually
isomorphic in accordance with Theorem \ref{nn205}.

$\bullet$ We consider $\mathbb N$-graded commutative rings
(Definition \ref{nn44}). Any commutative $\mathbb N$-graded ring
$\cA^*$ can be regarded as an even $\mathbb N$-graded commutative
ring $\La^*$ such that $\La^{2i}=\cA^i$, $\La^{2i+1}=0$.

$\bullet$ Grassmann-graded rings (Definition \ref{nn106})
exemplify $\mathbb N$-graded commutative rings. They are finitely
generated in degree 1 (Definition \ref{nn210}).

$\bullet$ As was mentioned above, any $\mathbb N$-graded
commutative ring $\cA^*$ possesses a structure of a $\mathbb
Z_2$-graded commutative ring $\cA$ (Definition \ref{nn112}). In
particular, a Grassmann-graded ring is a Grassmann algebra
(Definition \ref{nn114}).

Hereafter, we restrict our consideration to $N$-graded commutative
rings (Definition \ref{nn44}).

Given an $\mathbb N$-graded commutative ring $\cA^*$, an $\mathbb
N$-graded $\cA^*$-module $Q$ is defined as a graded
$\cA^*$-bimodule which is an $\mathbb N$-graded $\cK$-module such
that
\be
qa = (-1)^{[a][q]}aq, \qquad [aq]=[a]+[q], \qquad a\in\cA^*,
\qquad q\in Q,
\ee
and it also is $\mathbb Z_2$-graded module.

A direct sum, a tensor product of $\mathbb N$-graded modules and
the exterior algebra $\w P$ of an $\mathbb N$-graded module are
defined similarly to those of $\mathbb Z_2$-graded modules
(Section 5.1), and they also are a direct sum, a tensor product
and an exterior algebra of associative $\mathbb Z_2$-graded
modules, respectively.

A morphism $\Phi:P\to Q$ of $\mathbb N$-graded $\cA^*$-modules
seen as $\cK$-modules is said to be homogeneous of degree $[\Phi]$
if $[\Phi(p)]=[p]+[\Phi]$ for all homogeneous elements $p\in P$
and the relations (\ref{nn212}) hold. A morphism $\Phi:P\to Q$ of
$\mathbb N$-graded $\cA^*$-modules as the $\cK$-ones is called a
$\mathbb N$-graded $\cA^*$-module morphism if it is represented by
a homogeneous morphisms. Therefore, a set $\hm_\cA(P,Q)$ of graded
morphisms of an $\mathbb N$-graded $\cA$-module $P$ to an $\mathbb
N$-graded $\cA^*$-module $Q$ is an $\mathbb N$-graded
$\cA^*$-module. An $\mathbb N$-graded $\cA^*$-module
$P^*=\hm_\cA(P,\cA)$ is called the dual of an $\mathbb N$-graded
$\cA$-module $P$. Certainly, an $\mathbb N$-graded $\cA^*$-module
morphism of $\mathbb N$-graded $\cA^*$-modules is their $\mathbb
Z_2$-graded $\cA$-module morphism as associative $\mathbb
Z_2$-graded modules, however the converse need not be true.

By automorphisms of an $\mathbb N$-graded ring $\cA^*$ are meant
automorphisms of a $\cK$-ring $\cA$ which preserve its $\mathbb
N$-gradation $\cA^*$. They also preserve the associated $\mathbb
Z_2$-structure of $\cA$. However, there exist automorphisms of a
$\cK$-ring $\cA$ which do not possess these properties.

For example, let $\cA^*$ be the Grassmann-graded ring
(\ref{z784}). As was mentioned above its automorphisms
(\ref{nn127}) where $b^i\neq 0$ and (\ref{nn280}) do not preserve
an $\mathbb N$-graded structure of $\cA$, and automorphisms
(\ref{nn280}) also do not keep its $\mathbb Z_2$ structure.
However, they are morphisms of $\cA^*$ as an $\mathbb N$-graded
module, because can be represented as a certain sum of homogeneous
morphisms.

The differential calculus on $\mathbb N$-graded modules over
$\mathbb N$-graded commutative rings is defined just as that over
$\mathbb Z_2$-graded commutative rings (Section 5.2), but
different from the differential calculus over non-commutative
rings in Section 9 (Remark \ref{nn121}).

However, it should be emphasized that an $\mathbb N$-graded
differential operator is an $\mathbb N$-graded $\cK$-module
homomorphism which obeys the conditions (\ref{nn300}), i.e., it is
a sum of homogeneous morphisms of fixed $\mathbb N$-degrees, but
not the $\mathbb Z_2$ ones. Therefore, any $\mathbb N$-graded
differential operator also is a $\mathbb Z_2$-graded differential
operator, but the converse might not be true (Remark \ref{nn301}).

In particular, $\mathbb N$-graded derivations of an $\mathbb
N$-graded commutative $\cK$-ring $\cA^*$ constitute a Lie
superalgebra $\gd\cA^*$ (Definition \ref{nn201}) over a
commutative ring $\cK$ with respect to the superbracket
(\ref{ws14}). It is a subalgebra of a Lie superalgebra $\gd\cA$ of
$\mathbb Z_2$-graded derivations of a $\mathbb Z_2$-graded
commutative $\cK$-ring $\cA$ in general. The Chevalley--Eilenberg
complex $C^*[\gd\cA^*;\cA^*]$ (\ref{ws85}) and the
Chevalley--Eilenberg differential calculus $\cO^*[\gd\cA^*]$ over
an $\mathbb N$-graded commutative ring $\cA^*$ are constructed
similarly to those over a $\mathbb Z_2$-graded commutative ring
$\cA^*$ in Section 5.2.

As was mentioned above, the notion of a local ring can be extended
to the non-commutative ones (Definition \ref{nn251}), and
formalism of $\mathbb N$-graded commutative local-ringed spaces
can be developed just as that of commutative local-ringed spaces
in Section 3.1 and $\mathbb Z_2$-graded manifolds in Section 5.3.

Hereafter, we restrict our consideration of $\mathbb N$-graded
commutative rings to Grassmann-graded rings and $\mathbb N$-graded
commutative rings regarded as the even graded ones.

Let $\cA^*$ be a Grassmann-graded $\cK$-ring (Definition
\ref{nn106}) whose associated $\mathbb Z_2$-graded commutative
ring is a Grassmann algebra $\cA$ (Definition \ref{nn114}). Seen
as a $\cK$-ring $\cA$, it admits different structures of an
$\mathbb N$-graded commutative ring and a Grassmann algebra, but
they are mutually isomorphic if $\cK$ is a field in accordance
with Theorem \ref{nn288}.

A Grassmann-graded ring is local because of a unique maximal ideal
$R$ of its nilpotent elements (Remark \ref{nn245}). Given an odd
generating basis $\{c^i\}$ for a $\cK$-module $\cA^1$, elements of
a Grassmann-graded ring $\cA^*$ take the form (\ref{z784}). As was
mentioned above, $\cK$-ring automorphisms of $\cA$ are
compositions of automorphisms (\ref{nn127}) and (\ref{nn280}), but
automorphisms of an $\mathbb N$-graded ring $\cA^*$ take a form
\mar{nn302}\beq
c^i\to c'^i=\rho^i_jc^j  \label{nn302}
\eeq
where $\rho$ is an automorphism of $\cK$-module $\cA^1$.

The differential calculus over a Grassmann-graded $\cK$-ring
$\cA^*$ is exactly the differential calculus over an associated
Grassmann algebra $\cA$ (Example \ref{nn133}). Namely, the
derivations (\ref{nn131}) of $\cA$ also are derivations of a
Grassmann-graded ring $\cA^*$, and any $\mathbb N$-graded
differential operator on $\cA^*$ is a composition of these
derivations.

Since Grassmann-graded rings are local, let us now consider
local-ringed spaces whose stalks are such kind $\mathbb N$-graded
commutative rings. Since Grassmann-graded rings also are Grassmann
algebras, we follow formalism of $\mathbb Z_2$-graded manifolds in
Section 5.3.

Let $Z$ be an $n$-dimensional real smooth manifold. Let $\cA$ be
real Grassmann-graded ring. By virtue of Theorem \ref{nn117}), it
is isomorphic to the exterior algebra $\w W$ of a real vector
space $W=\cA^1$. Therefore, we come to the following definition.

\begin{definition} \label{nn306} \mar{nn306} An
$N$-graded manifold is a simple $\mathbb Z_2$-graded manifold
$(Z,\gA_E)$ modelled over some vector bundle $E\to Z$ (Definition
\ref{nn173}).
\end{definition}

In accordance with Definition \ref{nn306}, an $\mathbb N$-graded
manifold is a local-ringed space.

In view of Definition \ref{nn306}, Serre--Swan Theorem  \ref{vv0}
for $\mathbb Z_2$-graded manifolds also can be formulated for the
$\mathbb N$-graded ones.

\begin{theorem} \label{vv00} \mar{vv00}
Let $Z$ be a smooth manifold. A $\mathbb N$-graded commutative
$C^\infty(Z)$-ring $\cA$ is isomorphic to the structure ring of a
$\mathbb N$-graded manifold with a body $Z$ iff it is the exterior
algebra of some projective $C^\infty(Z)$-module of finite rank.
\end{theorem}

In Section 5.3 on $\mathbb Z_2$-graded manifolds, we have
restricted our consideration to simple graded manifolds, this
Section, in fact, presents formalism of $\mathbb N$-graded
manifolds.

\begin{remark} \label{nn301} \mar{nn301}
Let emphasize the essential peculiarity of an $\mathbb N$-graded
manifold $(Z,\gA_E)$ in comparison with the $\mathbb Z_2$-graded
ones. Derivations of its structure module $\cA_E$ are represented
by sections of the vector bundle $\cV_E$ (\ref{nn253}). Due to
this fact the Chevalley--Eilenberg differential calculus
$\cS^*[E;Z]$ (\ref{33f21}) over $\cA_E$ is minimal, i.e., it is
generated by elements $df$, $f\in \cA_E$.
\end{remark}

\section{$\mathbb N$-Graded bundles}

An important example of commutative $N$-graded rings are the
polynomial ones (Example \ref{nn224}). In particular, let $\cA$ be
a commutative finitely generated ring over an algebraically closed
field $\cK$, and let it possess no nilpotent elements. Then it is
isomorphic to the quotient of some polynomial $\cK$-ring which is
the coordinate ring (\ref{nn260}) of a certain affine variety
(Theorem \ref{ws93}).

A problem is that a polynomial ring $\cP[Q]$ is not local. Any its
element $q-\la\bb$, $\la\in \cK$, generates a maximal ideal.
Therefore, a question is how to construct ringed spaces in
polynomial rings. Though there is a certain correspondence between
affine varieties and affine schemes (Remark \ref{nn261}).

In a different way, one can consider a subsheaf of polynomial
functions of a sheaf of smooth functions. Let $\pi: Y\to X$ be a
vector bundle and $C^\infty_Y$ a sheaf of smooth functions on $Y$.
Let $\cP_Y$ be a subsheaf of $C^\infty_Y$ of germs of functions
which are polynomial on fibres of $Y$. These functions are well
defined due to linear transition functions of a vector bundle
$Y\to X$. Let $\pi^*\cP_Y$ be the direct image of a sheaf $\cP_Y$
onto $X$ (Example \ref{nn310}). Its stalk $\pi^*\cP_x$ at a point
$x\in X$ is a polynomial ring $\cP[Y_x]$ of a fibre $Y_x$ of $Y$
over $x$ over a local ring $C^\infty_x$ which is a stalk
$C^\infty_x$ of a sheaf of smooth real functions on $X$ at a point
$x\in X$.

In a more general setting, let us consider $\mathbb N$-graded
bundles (Definition \ref{nn175}).

An epimorphism of $\mathbb Z_2$-graded manifolds $(Z,\gA) \to
(Z',\gA')$ where $Z\to Z'$ is a fibre bundle is called the graded
bundle \cite{hern,stavr}.  In this case, a sheaf monomorphism
$\wh\Phi$ induces a monomorphism of canonical presheaves $\ol
\gA'\to \ol \gA$, which associates to each open subset $U\subset
Z$ the ring of sections of $\gA'$ over $\phi(U)$. Accordingly,
there is a pull-back monomorphism of the structure rings
$\gA'(Z')\to\gA(Z)$ of graded functions on graded manifolds
$(Z',\gA')$ and $(Z,\gA)$.

In particular, let $(Y,\gA)$ be an $\mathbb N$-graded manifold
whose body $Z=Y$ is a fibre bundle $\pi:Y\to X$. Let us consider a
trivial graded manifold $(X,\gA_0=C^\infty_X)$ (Remark
\ref{triv}). Then we have a graded bundle
\beq
(Y,\gA) \to (X,C^\infty_X). \label{su3}
\eeq
Let us denote it by $(X,Y,\gA)$. Given a graded bundle
$(X,Y,\gA)$, the local basis for a graded manifold $(Y,\gA)$ can
be brought into a form $(x^\la, y^i, c^a)$ where $(x^\la, y^i)$
are bundle coordinates of $Y\to X$.

\begin{definition} \label{nn175}
We agree to call the graded bundle (\ref{su3}) over a trivial
graded manifold $(X, C^\infty_X)$ the graded bundle over a smooth
manifold \cite{sard14,sard15}.
\end{definition}

If $Y\to X$ is a vector bundle, the graded bundle (\ref{su3}) is a
particular case of graded fibre bundles in \cite{hern,mont} when
their base is a trivial graded manifold.

\begin{remark} \label{su20}
Let $Y\to X$ be a fibre bundle. Then a  trivial graded manifold
$(Y,C^\infty_Y)$ together with a real ring monomorphism
$C^\infty(X)\to C^\infty(Y)$ is a graded bundle $(X,Y,C^\infty_Y)$
of trivial graded manifolds
\be
(Y, C^\infty_Y) \to (X, C^\infty_X).
\ee
\end{remark}

\begin{remark} \label{su21}  A graded manifold $(X,\gA)$ itself can
be treated as the graded bundle $(X,X, \gA)$ (\ref{su3})
associated to the identity smooth bundle $X\to X$.
\end{remark}

Let $E\to Z$ and $E'\to Z'$ be vector bundles and $\Phi: E\to E'$
their bundle morphism over a morphism $\phi: Z\to Z'$. Then every
section $s^*$ of the dual bundle $E'^*\to Z'$ defines the
pull-back section $\Phi^*s^*$ of the dual bundle $E^*\to Z$ by the
law
\be
v_z\rfloor \Phi^*s^*(z)=\Phi(v_z)\rfloor s^*(\vf(z)), \qquad
v_z\in E_z.
\ee
It follows that a bundle morphism $(\Phi,\phi)$ yields a morphism
of $\mathbb N$-graded manifolds
\beq
(Z,\gA_E) \to (Z',\gA_{E'}). \label{w901}
\eeq
This is a pair $(\phi,\wh\Phi=\phi_*\circ\Phi^*)$ of a morphism
$\phi$ of  body manifolds and the composition $\phi_*\circ\Phi^*$
of the pull-back $\cA_{E'}\ni f\to \Phi^*f\in\cA_E$ of graded
functions and the direct image $\phi_*$ of a sheaf $\gA_E$ onto
$Z'$. Relative to local bases $(z^A,c^a)$ and $(z'^A,c'^a)$ for
$(Z,\gA_E)$ and $(Z',\gA_{E'})$, the morphism (\ref{w901}) of
$\mathbb N$-graded manifolds reads $z'=\phi(z)$,
$\wh\Phi(c'^a)=\Phi^a_b(z)c^b$.

The graded manifold morphism (\ref{w901}) is a monomorphism (resp.
epimorphism) if $\Phi$ is a bundle injection (resp. surjection).

In particular, the graded manifold morphism (\ref{w901}) is an
$\mathbb N$-graded bundle if $\Phi$ is a fibre bundle. Let
$\cA_{E'} \to \cA_E$ be the corresponding pull-back monomorphism
of the structure rings. By virtue of Theorem \ref{v62} it yields a
monomorphism of the differential bigraded rings
\beq
\cS^*[E';Z']\to \cS^*[E;Z]. \label{xxx}
\eeq

Let $(Y,\gA_F)$ be an $\mathbb N$-graded manifold modelled over a
vector bundle $F\to Y$. This is an $\mathbb N$-graded bundle
$(X,Y,\gA_F)$:
\beq
(Y, \gA_F) \to (X, C^\infty_X) \label{su11}
\eeq
 modelled over a composite bundle
\beq
F\to Y\to X.  \label{su5}
\eeq
The structure ring of graded functions on an $\mathbb N$-graded
manifold $(Y,\gA_F)$ is the graded commutative $C^\infty(X)$-ring
$\cA_F=\w F^*(Y)$ (\ref{33f1}). Let the composite bundle
(\ref{su5}) be provided with adapted bundle coordinates
$(x^\la,y^i,q^a)$ possessing transition functions
\be
x'^\la(x^\mu), \qquad y'^i(x^\m,y^j), \qquad
q'^a=\rho^a_b(x^\mu,y^j)q^b.
\ee
The corresponding local basis for an $\mathbb N$-graded manifold
$(Y,\gA_F)$ is $(x^\la,y^i,c^a)$ together with transition
functions
\be
x'^\la(x^\mu), \qquad y'^i(x^\m,y^j), \qquad
c'^a=\rho^a_b(x^\mu,j^j)c^b.
\ee
We call it the local basis for an $\mathbb N$-graded bundle
 $(X,Y,\gA_F)$.

 With respect to this basis, graded functions on
$(X,Y,\gA_F)$ take the form (\ref{z785}) where $f_{a_1\ldots
a_k}(x^\la,y^i)$ are local functions on $Y$. Let $\pi:Y\to X$ be a
vector bundle. Then one can consider graded functions $f$ whose
coefficients $f_{a_1\ldots a_k}(x^\la,y^i)$ are polynomial in
fibre coordinates $(y^i)$ of $Y$. Let $\cP$ be the sheaf of germs
of these functions on $Y$. Its direct image $\pi^*\cP$ is a sheaf
on $X$ whose stalk at $x\in X$ is an $\mathbb N$-graded ring of
polynomials both in even variables $y^i$ and the odd ones $c^a$
with coefficients in a stalk $C^\infty_x$ of smooth real functions
on $X$.

\section{Appendix. Cohomology}

For the sake of convenience of the reader, the relevant topics on
cohomology are compiled in this Appendix.

\subsection{Cohomology of complexes}

We start with cohomology of complexes of modules over a
commutative ring \cite{mcl,massey,book12}.

Let $\cK$ be a commutative ring.  A sequence
\mar{b3256}\beq
0\to B^0 \ar^{\dl^0} B^1 \ar^{\dl^1}\cdots B^p\ar^{\dl^p}\cdots
\label{b3256}
\eeq
of modules $B^p$ and their homomorphisms $\dl^p$ is said to be the
  cochain complex (henceforth, simply, a   complex) if
\be
\dl^{p+1}\circ \dl^p =0, \qquad  p\in \mathbb N,
\ee
i.e., $\im \dl^p\subset \Ker \dl^{p+1}$. Homomorphisms $\dl^p$ are
called   the coboundary operators. Elements of a module $B^p$ are
said to be the  $p$-cochains, whereas elements of its submodules
$\Ker \dl^p\subset B^p$ and  $\im \dl^{p-1}\subset \Ker \dl^p$ are
called the   $p$-cocycles and $p$-coboundaries, respectively. The
$p$-th cohomology group of the complex $B^*$ (\ref{b3256}) is the
factor module
\be
H^p(B^*)= \Ker \dl^p/\im \dl^{p-1}.
\ee
It is a $\cK$-module. In particular, $H^0(B^*)=\Ker \dl^0$.

The complex (\ref{b3256}) is said to be   exact at a term $B^p$ if
$H^p(B^*)=0$. It is an exact sequence if all cohomology groups are
trivial.

A complex $(B^*,\dl^*)$ is called   acyclic if its cohomology
groups $H^{p>0}(B^*)$ are trivial. A complex $(B^*,\dl^*)$ is said
to be the   resolution of a module $B$ if it is acyclic and
$H^0(B^*)=B$.

A   cochain morphism of complexes
\mar{spr32'}\beq
\g:B^*_1\to B^*_2 \label{spr32'}
\eeq
is defined as a family of degree-preserving homomorphisms
\be
\g^p: B^p_1\to B^p_2, \qquad p\in\mathbb N,
\ee
which commute with the coboundary operators, i.e.,
\be
\dl^p_2\circ\g^p=\g^{p+1}\circ\dl^p_1, \qquad p\in\mathbb N.
\ee
It follows that if $b^p\in B^p_1$ is a cocycle or a coboundary,
then $\g^p(b^p)\in B^p_2$ is so. Therefore, the cochain morphism
of complexes (\ref{spr32'}) yields an induced homomorphism of
their cohomology groups
\be
[\g]^*: H^*(B^*_1) \to H^*(B^*_2).
\ee

Let us consider a short sequence of complexes
\mar{spr34'}\beq
0\to C^*\ar^\g B^* \ar^\zeta F^*\to 0, \label{spr34'}
\eeq
represented by the commutative diagram
\be
\begin{array}{rcrccrccl}
& & & 0 & & & 0 & &\\
& & & \put(0,10){\vector(0,-1){20}} & & &
\put(0,10){\vector(0,-1){20}}
   & &\\
\cdots &\longrightarrow & & C^p & \op\longrightarrow^{\dl^p_C} & &
C^{p+1} & \longrightarrow &  \cdots \\
& & _{\g_p} & \put(0,10){\vector(0,-1){20}} & & _{\g_{p+1}}&
\put(0,10){\vector(0,-1){20}}
   &  &\\
\cdots &\longrightarrow & & B^p & \op\longrightarrow^{\dl^p_B} & &
B^{p+1} & \longrightarrow &  \cdots \\
& &_{\zeta_p} & \put(0,10){\vector(0,-1){20}} & & _{\zeta_{p+1}}
& \put(0,10){\vector(0,-1){20}} & & \\
\cdots &\longrightarrow & & F^p & \op\longrightarrow^{\dl^p_F} & &
F^{p+1} & \longrightarrow  & \cdots \\
& & & \put(0,10){\vector(0,-1){20}} & & &
\put(0,10){\vector(0,-1){20}}
   & & \\
& & & 0 & & & 0 & &
\end{array}
\ee
It is said to be exact if all columns of this diagram are exact,
i.e., $\g$ is a cochain monomorphism and $\zeta$ is a cochain
epimorphism onto the quotient $F^*=B^*/C^*$.

\begin{theorem} \label{spr36'} \mar{spr36'}
The short exact sequence  of complexes (\ref{spr34'}) yields a
long exact sequence of their cohomology groups
\be
&& 0\to H^0(C^*)\ar^{[\g]^0} H^0(B^*)\ar^{[\zeta]^0} H^0(F^*)\ar^{\tau^0}
H^1(C^*)\ar\cdots \\
&& \qquad \ar H^p(C^*)\ar^{[\g]^p} H^p(B^*)\ar^{[\zeta]^p}
H^p(F^*)\ar^{\tau^p} H^{p+1}(C^*)\ar\cdots.
\ee
\end{theorem}

\begin{theorem} \label{spr37'} \mar{spr37'}
A direct sequence of complexes
\be
B^*_0\ar B^*_1\ar\cdots B^*_k\ar^{\g^k_{k+1}} B^*_{k+1}\ar \cdots
\ee
admits a direct limit $B^*_\infty$ which is a complex whose
cohomology $H^*(B^*_\infty)$ is a direct limit of the direct
sequence of cohomology groups
\be
H^*(B^*_0)\ar H^*(B^*_1)\ar \cdots H^*(B^*_k)\ar^{[\g^k_{k+1}]}
H^*(B^*_{k+1})\ar \cdots.
\ee
This statement also is true for a direct system of complexes
indexed by an arbitrary directed set (Remark \ref{nn2}).
\end{theorem}

\subsection{Cohomology of Lie algebras}

One can associate the Chevalley--Eilenberg cochain complex
(\ref{spr997}) to an arbitrary Lie algebra \cite{fuks,book12}. In
this Section, $\cG$ denotes a Lie algebra (not necessarily
finite-dimensional) over a commutative ring $\cK$.

Let $P$ be a $\cK$-module, and let $\cG$ act on $P$ on the left by
endomorphisms
\be
&& \cG\times P\ni (\ve,p)\to \ve p\in P, \\
&& [\ve,\ve']p=(\ve\circ
\ve'-\ve'\circ \ve)p, \qquad \ve,\ve'\in\cG.
\ee
One says that $P$ is a   $\cG$-module.  A $\cK$-multilinear
skew-symmetric map
\be
c^k:\op\times^k\cG\to P
\ee
is called the $P$-valued $k$-cochain on a Lie algebra $\cG$. These
cochains form a $\cG$-module $C^k[\cG;P]$. Let us put
$C^0[\cG;P]=P$. We obtain the cochain complex
\mar{spr997}\beq
0\to P\ar^{\dl^0} C^1[\cG;P]\ar^{\dl^1} \cdots C^k[\cG;P]
\ar^{\dl^k} \cdots, \label{spr997}
\eeq
with respect to the   Chevalley--Eilenberg coboundary operators
\mar{spr132}\ben
&& \dl^kc^k (\ve_0,\ldots,\ve_k)=\op\sum_{i=0}^k(-1)^i\ve_ic^k(\ve_0,\ldots,
\wh\ve_i, \ldots, \ve_k)+ \label{spr132}\\
&& \qquad \op\sum_{1\leq i<j\leq k}
(-1)^{i+j}c^k([\ve_i,\ve_j], \ve_0,\ldots, \wh\ve_i, \ldots,
\wh\ve_j,\ldots, \ve_k), \nonumber
\een
where the caret $\,\wh{}\,$ denotes omission. The complex
(\ref{spr997}) is called the   Chevalley--Eilenberg complex with
coefficients in a module $P$. Cohomology $H^*(\cG;P)$ of the
complex $C^*[\cG;P]$ is called the Chevalley--Eilenberg cohomology
of a Lie algebra $\cG$ with coefficients in a module $P$.

\subsection{Sheaf cohomology}

In this Section, we follow the terminology of \cite{bred,hir}. All
presheaves and sheaves are considered on the same topological
space $X$.

A   sheaf on a topological space $X$ is a topological fibre bundle
$\pi:S\to X$ in modules over a commutative ring $\cK$, where a
surjection $\pi$ is a local homeomorphism and fibres $S_x$, $x\in
X$, called the   stalks, are provided with the discrete topology.
Global sections of a sheaf $S$ constitute a $\cK$-module $S(X)$,
called the   structure module of $S$.

Any sheaf is generated by a presheaf. A   presheaf $S_\sU$ on a
topological space $X$ is defined if a module $S_U$ over a
commutative ring $\cK$ is assigned to every open subset $U\subset
X$ $(S_\emptyset=0)$ and if, for any pair of open subsets
$V\subset U$, there exists a  restriction morphism
$r_V^U:S_U\rightarrow S_V$ such that
\be
r_U^U=\id S_U,\qquad  r_W^U=r_W^Vr_V^U, \qquad W\subset V\subset
U.
\ee

Every presheaf $S_\sU$ on a topological space $X$ yields a sheaf
on $X$ whose stalk $S_x$ at a point $x\in X$ is the direct limit
of modules $S_U,\,x\in U$, with respect to restriction morphisms
$r_V^U$. It means that, for each open neighborhood $U$ of a point
$x$, every element $s\in S_U$ determines an element $s_x\in S_x$,
called the   germ of $s$ at $x$. Two elements $s\in S_U$ and
$s'\in S_V$ belong to the same germ at $x$ iff there exists an
open neighborhood $W\subset U\cap V$ of $x$ such that
$r_W^Us=r_W^Vs'$.

\begin{example} \label{spr7} \mar{spr7}
Let $C^0_\sU$ be a presheaf of continuous real functions on a
topological space $X$. Two such functions $s$ and $s'$ define the
same germ $s_x$ if they coincide on an open neighborhood  of $x$.
Hence, we obtain a sheaf $C^0_X$ of continuous functions on $X$.
Similarly, a sheaf $C^\infty_X$ of smooth functions on a smooth
manifold $X$ is defined. Let us also mention a presheaf of real
functions which are constant on connected open subsets of $X$. It
generates the   constant sheaf on $X$
 denoted by $\mathbb R$.
\end{example}

\begin{example} \label{t1} \mar{t1}
Let $Y\to X$ be a smooth vector bundle. A sheaf of germs of its
sections is denoted $Y_X$. Its structure module is a
$C^\infty(X)$-ring $Y(X)$ of global sections of $Y\to X$.
\end{example}

Two different presheaves may generate the same sheaf. Conversely,
every sheaf $S$ defines a presheaf $S(\sU)$ of modules $S(U)$ of
its local sections. It is called the   canonical presheaf of a
sheaf $S$. Global sections of $S$ constitute the   structure
module $S(X)$ of $S$. If a sheaf $S$ is constructed from a
presheaf $S_\sU$, there are natural module morphisms
\be
S_U\ni s\to s(U)\in S(U), \qquad s(x)= s_x, \quad x\in U,
\ee
which are neither monomorphisms nor epimorphisms in general. For
instance, it may happen that a non-zero presheaf defines a zero
sheaf. The sheaf generated by the canonical presheaf of a sheaf
$S$ coincides with $S$.

A direct sum and a tensor product of presheaves (as families of
modules)  and sheaves (as fibre bundles in modules) are naturally
defined. By virtue of Theorem \ref{spr170}, a direct sum (resp. a
tensor product) of presheaves generates a direct sum (resp. a
tensor product) of the corresponding sheaves.

\begin{remark} \label{spr190'} \mar{spr190'}
In the terminology of \cite{tenn}, a sheaf is introduced as a
presheaf which satisfies the following additional axioms.

(S1) Suppose that $U\subset X$ is an open subset and $\{U_\al\}$
is its open cover. If $s,s'\in S_U$ obey a condition
$r^U_{U_\al}(s)=r^U_{U_\al}(s')$ for all $U_\al$, then $s=s'$.

(S2) Let $U$ and $\{U_\al\}$ be as in previous item. Suppose that
we are given a family of presheaf elements $\{s_\al\in
S_{U_\al}\}$ such that
\be
r^{U_\al}_{U_\al\cap U_\la}(s_\al)=r^{U_\la}_{U_\al\cap
U_\la}(s_\la)
\ee
for all $U_\al$, $U_\la$. Then there exists an element $s\in S_U$
such that $s_\al=r^U_{U_\al}(s)$.

Canonical presheaves are in one-to-one correspondence with
presheaves obeying these axioms. For instance, the presheaves of
continuous, smooth and locally constant functions in Example
\ref{spr7} satisfy the axioms (S1) -- (S2).
\end{remark}

A   morphism of a presheaf $S_\sU$ to a presheaf $S'_\sU$ on a
topological space $X$ is defined as a set of module morphisms
$\g_U:S_U\to S'_U$ which commute with restriction morphisms. A
morphism of presheaves yields a   morphism of sheaves generated by
these presheaves. This is a bundle morphism over $X$ such that
$\g_x: S_x\to S'_x$ is the direct limit of morphisms $\g_U$, $x\in
U$. Conversely, any morphism of sheaves $S\to S'$ on a topological
space $X$ yields a morphism of canonical presheaves of local
sections of these sheaves. Let $\hm(S|_U,S'|_U)$ be a commutative
group of sheaf morphisms $S|_U\to S'|_U$ for any open subset
$U\subset X$. These groups are assembled into a presheaf, and
define the sheaf $\hm(S,S')$ on $X$. There is a monomorphism
\mar{+212}\beq
\hm(S,S')(U)\to \hm(S(U),S'(U)), \label{+212}
\eeq
which need not be an isomorphism.

By virtue of Theorem \ref{dlim1}, if a presheaf morphism is a
monomorphism or an epimorphism, so is the corresponding sheaf
morphism. Furthermore, the following holds.

\begin{theorem} \label{spr29} \mar{spr29}
A short exact sequence
\mar{spr208}\beq
0\to S'_\sU\to S_\sU\to S''_\sU\to 0 \label{spr208}
\eeq
of presheaves on a topological space $X$ yields a short exact
sequence of sheaves generated by these presheaves
\mar{ms0102}\beq
0\to S'\to S\to S''\to 0, \label{ms0102}
\eeq
where the   factor sheaf $S''=S/S'$ is isomorphic to that
generated by the factor presheaf $S''_\sU=S_\sU/S'_\sU$. If the
exact sequence of presheaves (\ref{spr208}) is split, i.e.,
\be
S_\sU\cong S'_\sU\oplus S''_\sU,
\ee
the corresponding splitting $S\cong S'\oplus S''$ of the exact
sequence of sheaves (\ref{ms0102}) holds.
\end{theorem}

The converse is more intricate. A sheaf morphism induces a
morphism of the corresponding canonical presheaves. If $S\to S'$
is a monomorphism, $S(\sU)\to S'(\sU)$ also is a monomorphism.
However, if $S\to S'$ is an epimorphism, $S(\sU)\to S'(\sU)$ need
not be so. Therefore, the short exact sequence (\ref{ms0102}) of
sheaves yields the exact sequence of the canonical presheaves
\mar{ms0103'}\beq
0\to S'(\sU)\to S(\sU)\to S''(\sU), \label{ms0103'}
\eeq
where $S(\sU)\to S''(\sU)$ is not necessarily an epimorphism. At
the same time, there is the short exact sequence of presheaves
\mar{ms0103}\beq
0\to S'(\sU)\to S(\sU)\to S''_\sU \to 0, \label{ms0103}
\eeq
where the factor presheaf
\be
S''_\sU=S(\sU)/S'(\sU)
\ee
generates the factor sheaf $S''=S/S'$, but need not be its
canonical presheaf.

\begin{theorem} \label{spr30} \mar{spr30}
Let the exact sequence of sheaves (\ref{ms0102}) be split. Then
\be
S(\sU)\cong S'(\sU) \oplus S''(\sU),
\ee
and the canonical presheaves make up the short exact sequence
\be
0\to S'(\sU)\to S(\sU)\to S''(\sU)\to 0.
\ee
\end{theorem}

Let us turn now to sheaf cohomology. We follow its definition in
\cite{hir}.

Let $S_\sU$ be a presheaf of modules on a topological space $X$,
and let $\gU=\{U_i\}_{i\in I}$ be an open cover of $X$. One
constructs a cochain complex where a $p$-cochain is defined as a
function $s^p$ which associates an element
\be
s^p(i_0,\ldots,i_p)\in S_{U_{i_0}\cap\cdots\cap U_{i_p}}
\ee
to each $(p+1)$-tuple $(i_0,\ldots,i_p)$ of indices in $I$. These
$p$-cochains are assembled into a module $C^p(\gU,S_\sU)$. Let us
introduce a coboundary operator
\mar{spr180}\ben
&& \delta^p:C^p(\gU,S_\sU)\to C^{p+1}(\gU,S_\sU), \nonumber\\
&& \dl^ps^p(i_0,\ldots,i_{p+1})=\op\sum_{k=0}^{p+1}(-1)^kr_W^
{W_k}s^p(i_0,\ldots,\wh i_k,\ldots,i_{p+1}), \label{spr180}\\
&&
W=U_{i_0}\cap\ldots\cap U_{i_{p+1}},\qquad
W_k=U_{i_0}\cap\cdots\cap\wh U_{i_k}\cap\cdots\cap
U_{i_{p+1}}.\nonumber
\een
It is easily justified that $\delta^{p+1}\circ\delta^p=0$. Thus,
we obtain a cochain complex of modules
\mar{spr181}\beq
0\to C^0(\gU,S_\sU)\ar^{\dl^0}\cdots C^p(\gU,S_\sU)\ar^{\dl^p}
C^{p+1}(\gU,S_\sU)\ar\cdots. \label{spr181}
\eeq
Its cohomology groups
\be
H^p(\gU;S_\sU)=\Ker\dl^p/\im\dl^{p-1}
\ee
are modules. Certainly, they depend on a cover $\gU$ of a
topological space $X$.

\begin{remark}
Only   proper covers throughout are considered, i.e., $U_i\neq
U_j$ if $i\neq j$. A cover $\gU'$ is said to be the refinement of
a cover $\gU$ if, for each $U'\in\gU'$, there exists $U\in\gU$
such that $U'\subset U$.
\end{remark}

If $\gU'$ is a refinement of a cover $\gU$, there is a morphism of
cohomology groups
\be
H^*(\gU;S_\sU)\rightarrow H^*(\gU';S_\sU).
\ee
Let us take the direct limit of cohomology groups $H^*(\gU;S_\sU)$
with respect to these morphisms, where $\gU$ runs through all open
covers of $X$. This limit $H^*(X;S_\sU)$ is called the cohomology
of $X$ with coefficients  in a presheaf $S_\sU$.

Let $S$ be a sheaf on a topological space $X$.   Cohomology of $X$
with coefficients in $S$ or, simply,   sheaf cohomology is defined
as cohomology
\be
H^*(X;S)=H^*(X;S(\sU))
\ee
with coefficients in the canonical presheaf $S(\sU)$ of a sheaf
$S$.

In this case, a $p$-cochain $s^p\in C^p(\gU,S(\sU))$ is a family
$s^p=\{s^p(i_0,\ldots,i_p)\}$ of local sections
$s^p(i_0,\ldots,i_p)$ of a sheaf $S$ over $U_{i_0}\cap\cdots\cap
U_{i_p}$ for each $(p+1)$-tuple $(U_{i_0},\ldots,U_{i_p})$ of
elements of a cover $\gU$. The coboundary operator (\ref{spr180})
reads
\be
\dl^ps^p(i_0,\ldots,i_{p+1})=\op\sum_{k=0}^{p+1}(-1)^k
s^p(i_0,\ldots,\wh i_k,\ldots,i_{p+1})|_{U_{i_0}\cap\cdots\cap
U_{i_{p+1}}}.
\ee
For instance,
\mar{spr188,9}\ben
&& \dl^0s^0(i,j)=[s^0(j) -s^0(i)]|_{U_i\cap U_j},
\label{spr188}\\
&& \dl^1s^1(i,j,k)=[s^1(j,k)-s^1(i,k)
   +s^1(i,j)]|_{U_i\cap U_j\cap U_k}. \label{spr189}
\een
A glance at the expression (\ref{spr188}) shows that a
zero-cocycle is a collection $s=\{s(i)\}_I$ of local sections of a
sheaf $S$ over $U_i\in\gU$ such that $s(i)=s(j)$ on $U_i\cap U_j$.
It follows from the axiom (S2) in Remark \ref{spr190'} that $s$ is
a global section of a sheaf $S$, while each $s(i)$ is its
restriction $s|_{U_i}$ to $U_i$. Consequently, the cohomology
group $H^0(\gU,S(\sU))$ is isomorphic to the structure module
$S(X)$ of global sections of a sheaf $S$. A one-cocycle is a
collection $\{s(i,j)\}$ of local sections of a sheaf $S$ over
overlaps $U_i\cap U_j$ which satisfy the   cocycle condition
\be
[s(j,k)-s(i,k) +s(i,j)]|_{U_i\cap U_j\cap U_k}=0.
\ee

If $X$ is a paracompact space, the study of its sheaf cohomology
is essentially simplified due to the following fact \cite{hir}.

\begin{theorem} \label{spr225} \mar{spr225}
Cohomology of a paracompact space $X$ with coefficients in a sheaf
$S$ coincides with cohomology of $X$ with coefficients in any
presheaf generating a sheaf $S$.
\end{theorem}

\begin{remark} \label{spr200} \mar{spr200}
We follow the definition of a   paracompact topological space in
\cite{hir} as a Hausdorff space such that any its open cover
admits a   locally finite open refinement, i.e., any point has an
open neighborhood which intersects only a finite number of
elements of this refinement. A topological space $X$ is
paracompact iff any cover $\{U_\xi\}$ of $X$ admits the
subordinate   partition of unity $\{f_\xi\}$, i.e.:

(i) $f_\xi$ are real positive continuous functions on $X$;

(ii) supp$\,f_\xi\subset U_\xi$;

(iii) each point $x\in X$ has an open neighborhood which
intersects only a finite number of the sets supp$\,f_\xi$;

(iv) $\op\sum_\xi f_\xi(x)=1$ for all $x\in X$.
\end{remark}

A key point of the analysis of sheaf cohomology is that short
exact sequences of presheaves and sheaves yield long exact
sequences of sheaf cohomology groups.

Let $S_\sU$ and $S'_\sU$ be presheaves on a topological space $X$.
It is readily observed that, given an open cover $\gU$ of $X$, any
morphism $S_\sU\to S'_\sU$  yields a cochain morphism of complexes
\be
C^*(\gU,S_\sU)\to C^*(\gU,S'_\sU)
\ee
and the corresponding morphism
\be
H^*(\gU,S_\sU)\to H^*(\gU,S'_\sU)
\ee
of cohomology groups of these complexes. Passing to the direct
limit through all refinements of $\gU$, we come to a morphism of
the cohomology groups
\be
H^*(X,S_\sU)\to H^*(X,S'_\sU)
\ee
of $X$ with coefficients in the presheaves $S_\sU$ and $S'_\sU$.
In particular, any sheaf morphism $S\to S'$ yields a morphism of
canonical presheaves $S(\{U\})\to S'(\{U\})$ and the corresponding
cohomology morphism
\be
H^*(X,S)\to H^*(X,S').
\ee

By virtue of Theorems \ref{spr36'} and \ref{spr37'}, every short
exact sequence
\be
0\to S'_\sU\ar S_\sU\ar S''_\sU\to 0
\ee
of presheaves on a topological space $X$ and the corresponding
exact sequence of complexes (\ref{spr181}) yield the long exact
sequence
\be
&& 0\to H^0(X;S'_\sU)\ar H^0(X;S_\sU)\ar H^0(X;S''_\sU)\ar \\
&& \qquad
H^1(X;S'_\sU) \ar\cdots  H^p(X;S'_\sU)\ar H^p(X;S_\sU)\ar \\
&& \qquad  H^p(X;S''_\sU)\ar
H^{p+1}(X;S'_\sU) \ar\cdots
\ee
of the cohomology groups of $X$ with coefficients in these
presheaves. This result however is not extended to an exact
sequence of sheaves, unless $X$ is a paracompact space. Let
\be
0\to S'\ar S\ar S'' \to 0
\ee
be a short  exact sequence of sheaves on $X$. It yields the short
exact sequence of presheaves (\ref{ms0103}) where the presheaf
$S''_\sU$ generates the sheaf $S''$. If $X$ is paracompact,
\be
H^*(X;S''_\sU)=H^*(X;S'')
\ee
in accordance with Theorem \ref{spr225}, and we have the exact
sequence of sheaf cohomology
\mar{spr227}\ben
&& 0\to H^0(X;S')\ar H^0(X;S)\ar H^0(X;S'')\ar
\label{spr227}\\
&& \qquad H^1(X;S') \ar\cdots   H^p(X;S')\ar H^p(X;S)\ar \nonumber
\\
&& \qquad H^p(X;S'')\ar H^{p+1}(X;S') \ar\cdots\,. \nonumber
\een

Let us point out the following isomorphism between the sheaf
cohomology and the singular (\v Chech and Alexandery) cohomology
of a paracompact space \cite{bred,span}.

\begin{theorem} \label{spr257} \mar{spr257}
The sheaf cohomology $H^*(X;\mathbb Z)$ (resp. $H^*(X;\mathbb R)$)
of a paracompact topological space $X$ with coefficients in the
constant sheaf $\mathbb Z$ (resp. $\mathbb R$) is isomorphic to
the singular cohomology of $X$  with coefficients in a ring
$\mathbb Z$ (resp. $\mathbb R$).
\end{theorem}

Since singular cohomology is a   topological invariant (i.e.,
homotopic topological spaces have the same singular cohomology)
\cite{span}, the sheaf cohomology groups $H^*(X;\mathbb Z)$ and
$H^*(X;\mathbb R)$ of a paracompact space also are topological
invariants.

Let us turn now to the abstract de Rham theorem which provides a
powerful tool of studying algebraic systems on paracompact spaces.

Let us consider an exact sequence of sheaves
\mar{spr228}\beq
0\to S\ar^h S_0\ar^{h^0} S_1\ar^{h^1}\cdots S_p\ar^{h^p}\cdots.
\label{spr228}
\eeq
It is said to be the   resolution of a sheaf $S$ if each sheaf
$S_{p\geq 0}$ is   acyclic, i.e., its cohomology groups
$H^{k>0}(X;S_p)$ vanish.

Any exact sequence of sheaves (\ref{spr228}) yields a sequence of
their structure modules
\mar{spr229}\beq
0\to S(X)\ar^{h_*} S_0(X)\ar^{h^0_*} S_1(X)\ar^{h^1_*}\cdots
S_p(X)\ar^{h^p_*}\cdots \label{spr229}
\eeq
which always is exact at terms $S(X)$ and $S_0(X)$ (see the exact
sequence (\ref{ms0103'})). The sequence (\ref{spr229}) is a
cochain complex because $h^{p+1}_*\circ h^p_*=0$. If $X$ is a
paracompact space and the exact sequence (\ref{spr228}) is a
resolution of $S$, the   abstract de Rham theorem establishes an
isomorphism of cohomology of the complex (\ref{spr229}) to
cohomology of $X$ with coefficients in a sheaf $S$ as follows
\cite{hir}.

\begin{theorem} \label{spr230} \mar{spr230}
Given a resolution (\ref{spr228}) of a sheaf $S$ on a paracompact
topological  space $X$ and the induced complex (\ref{spr229}),
there are isomorphisms
\mar{spr231}\beq
H^0(X;S)=\Ker h^0_*, \qquad H^q(X;S)=\Ker h^q_*/\im h^{q-1}_*,
\qquad q>0. \label{spr231}
\eeq
\end{theorem}

We will also refer to the following minor modification of Theorem
\ref{spr230} \cite{jmp01}.

\begin{theorem} \label{spr232} \mar{spr232}
Let
\be
0\to S\ar^h S_0\ar^{h^0} S_1\ar^{h^1}\cdots\ar^{h^{p-1}}
S_p\ar^{h^p} S_{p+1}, \qquad p>1,
\ee
be an exact sequence of sheaves on a paracompact topological space
$X$, where the sheaves $S_q$, $0\leq q<p$, are acyclic, and let
\be
0\to S(X)\ar^{h_*} S_0(X)\ar^{h^0_*}
S_1(X)\ar^{h^1_*}\cdots\ar^{h^{p-1}_*} S_p(X)\ar^{h^p_*}
S_{p+1}(X)
\ee
be the corresponding cochain complex of structure modules of these
sheaves. Then the isomorphisms (\ref{spr231}) hold for $0\leq
q\leq p$.
\end{theorem}

Let us mention a fine resolution of sheaves, i.e., a resolution by
fine sheaves. A sheaf $S$  on a paracompact space $X$ is called
fine if, for each locally finite open cover $\gU =\{U_i\}_{i\in
I}$ of $X$, there exists a system $\{h_i\}$ of endomorphisms
$h_i:S\to S$ such that:

(i) there is a closed subset $V_i\subset U_i$ and $h_i(S_x)=0$ if
$x\not\in V_i$,

(ii) $\op\sum_{i\in I}h_i$ is the identity map of $S$.

A fine sheaf on a paracompact space is acyclic. There are the
following important examples of fine sheaves \cite{hir}.

\begin{proposition}  \label{spr256} \mar{spr256}
Let $X$ be a paracompact topological  space which admits the
partition of unity performed by elements of the structure module
$\gA(X)$ of some sheaf $\gA$ of real functions on $X$. Then any
sheaf $S$ of $\gA$-modules on $X$, including $\gA$ itself, is
fine.
\end{proposition}

In particular, a sheaf $C^0_X$ of continuous functions on a
paracompact topological space is fine, and so is any sheaf of
$C^0_X$-modules. A smooth manifold $X$ admits the partition of
unity performed by smooth real functions. It follows that a sheaf
$C^\infty_X$ of smooth real functions on $X$ is fine, and so is
any sheaf of $C^\infty_X$-modules, e.g., the sheaves of sections
of smooth vector bundles over $X$.

\section{Appendix. Non-commutative differential calculus}

The notion of the graded differential calculus (Definition
\ref{nn51}) has been formulated for any $\cK$-ring. One can
generalize the Chevalley--Eilenberg differential calculus over a
commutative ring in Section 2.4 to a non-commutative $\cK$-ring
$\cA$ \cite{dub01,book05,book12}. As was mentioned above, the
extension of the notion of a differential operator to modules over
a non-commutative ring however meets difficulties
\cite{book05,book12}. A key point is that a multiplication in a
non-commutative ring is not a zero-order differential operator
(Remark \ref{nn103}). One overcomes this difficulty in a case of
graded commutative rings by means of a reformulation (Definition
\ref{nn122}) of the notion of differential operators (Definition
\ref{ws120}).

Namely, let $\cA$ further be an arbitrary non-commutative
$\cK$-ring. Its derivation $u\in\gd\cA$, by definition, obeys the
Leibniz rule
\mar{ws100'}\beq
u(ab)=u(a)b+au(b), \qquad a,b\in\cA. \label{ws100'}
\eeq
However, if $\cA$ is a graded commutative ring, its Leibniz  rule
(\ref{nn124}) differs from the above one (\ref{ws100'}).

By virtue of the relation (\ref{ws100'}), derivations preserve the
center $\cZ_\cA$ of $\cA$. A set of derivations $\gd\cA$ is both a
$\cZ_\cA$-bimodule and a Lie $\cK$-algebra with respect to the Lie
bracket (\ref{nn60}).

Given a Lie $\cK$-algebra $\gd\cA$, let us consider its
Chevalley--Eilenberg complex (\ref{ws102}) with coefficients in a
ring $\cA$, regarded as a $\gd\cA$-module. This complex contains a
subcomplex $\cO^*[\gd\cA]$ of $\cZ_\cA$-multilinear skew-symmetric
maps (\ref{+840'}) with respect to the Chevalley--Eilenberg
coboundary operator $d$ (\ref{+840}). Its terms $\cO^k[\gd\cA]$
are $\cA$-bimodules. The graded module $\cO^*[\gd\cA]$ is provided
with the product (\ref{ws103}) which obeys the relation
(\ref{ws98}) and brings $\cO^*[\gd\cA]$ into a differential graded
ring. Let us note that, if $\cA$ is not commutative, there is
nothing like the graded commutativity of forms (\ref{ws99}).
Therefore, a differential graded ring $\cO^*[\gd\cA]$ fails to be
$\mathbb N$-graded commutative.

The minimal Chevalley--Eilenberg differential calculus $\cO^*\cA$
over $\cA$ consists of monomials $a_0 da_1\w\cdots \w da_k$,
$a_i\in\cA$, whose product $\w$ (\ref{ws103}) obeys a
juxtaposition rule
\be
(a_0d a_1)\w (b_0d b_1)=a_0d (a_1b_0)\w d b_1- a_0a_1d b_0\w d
b_1, \qquad a_i,b_i\in\cA.
\ee
For instance, it follows from the product (\ref{ws103}) that, if
$a,a'\in\cZ_\cA$, then
\beq
da\w da'=-da'\w da, \qquad ada'=(da')a. \label{w123}
\eeq

\begin{remark} \label{nn102} \mar{nn102}
Let us mention a different graded differential calculus over a
non-commutative ring which often is used in non-commutative
geometry \cite{conn,land,book12}. Given a $\cK$-ring  $\cA$, let
us consider a tensor product $\cA\op\ot_\cK\cA$ of $\cK$-modules.
It is brought into an $\cA$-bimodule with respect to the
multiplication
\be
b(a\ot a')c=(ba)\ot (a'c), \qquad a,a',b,c\in\cA.
\ee
Let us consider its submodule $\Om^1(\cA)$ generated by the
elements $\bb\ot a-a\ot\bb$, $a\in\cA$. It is readily observed
that
\beq
d:\cA\ni a\to \bb\ot a -a\ot\bb\in\Om^1(\cA) \label{w110}
\eeq
is a $\Om^1(\cA)$-valued derivation of $\cA$. Thus, $\Om^1(\cA)$
is an $\cA$-bimodule generated by elements $da$, $a\in\cA$, such
that the relation
\beq
(da)b=d(ab)-adb, \qquad a,b\in \cA, \label{w265}
\eeq
holds. Let us consider the tensor algebra
$\Om^*(\cA)=\ot\Om^1(\cA)$ (\ref{spr620}) of an $\cA$-bimodule
$\Om^1(\cA)$ (Example \ref{ws40}). It consists of the monomials
\beq
a_0da_1\cdots da_k, \qquad a_i\in\cA, \label{w1005}
\eeq
whose product obeys the juxtaposition rule
\be
(a_0d a_1)(b_0d b_1)=a_0d (a_1b_0)d b_1- a_0a_1d b_0 b_1, \qquad
a_i,b_i\in \cA,
\ee
because of the relation (\ref{w265}). The operator $d$
(\ref{w110}) is extended to $\Om^*(\cA)$ by the law
\beq
d(a_0da_1\cdots da_k)=da_0da_1\cdots da_k, \label{w169}
\eeq
that makes $\Om^*(\cA)$ into a differential graded ring. Of
course, $\Om^*(\cA)$ is a minimal differential calculus. One calls
it the universal differential calculus. It differs from the
Chevalley--Eilenberg differential calculus. For instance, the
monomials $da$, $a\in\cZ_\cA$, of the universal differential
calculus do not satisfy the relations (\ref{w123}).
\end{remark}

It seems natural to regard the derivations (\ref{ws100'}) of a
non-commutative $\cK$-ring $\cA$ and the Chevalley--Eilenberg
coboundary operator $d$ (\ref{+840}) as particular differential
operators over $\cA$. Definition \ref{ws131} provides a standard
notion of differential operators on modules over a commutative
ring. However, there exist its different generalizations to
modules over a non-commutative ring
\cite{bor97,dublmp,dub01,book05,lunts,book12}.

Let $P$ and $Q$ be $\cA$-bimodules over a non-commutative
$\cK$-ring $\cA$. The $\cK$-module $\hm_\cK(P,Q)$ of $\cK$-linear
homomorphisms $\Phi:P\to Q$ can be provided with the left $\cA$-
and $\cA^\bll$-module structures (\ref{5.29}) and the similar
right module structures
\beq
(\Phi a)(p)=\Phi(p)a, \qquad (a\bll\Phi)(p)=\Phi(pa), \quad
a\in\cA, \qquad p\in\ P. \label{ws105}
\eeq
For the sake of convenience, we will refer to the module
structures (\ref{5.29}) and (\ref{ws105}) as the left and right
$\cA-\cA^\bll$ structures, respectively. Let us put
\mar{ws133a,3}\ben
&& \dl_a\Phi= a\Phi -\Phi\bll a, \qquad a\in\cA, \label{ws133a} \\
&& \ol\dl_a\Phi=\Phi a-a\bll\Phi, \qquad a\in\cA, \qquad \Phi\in
\hm_\cK(P,Q). \label{ws133}
\een
It is readily observed that
\be
\dl_a\circ\ol\dl_b=\ol\dl_b\circ\dl_a, \qquad a,b\in\cA.
\ee

The left $\cA$-module homomorphisms $\Delta: P\to Q$ obey
conditions $\dl_a\Delta=0$, for all $a\in\cA$ and, consequently,
they can be regarded as left zero-order $Q$-valued differential
operators on $P$. Similarly, right zero-order differential
operators are defined.

\begin{remark} \label{nn103} \mar{nn103}
However, a left (resp., right) multiplication fails to be a left
(resp., right) zero-order differential operator.
\end{remark}

Utilizing the condition (\ref{ws106}) as a definition of a
first-order differential operator, one meets difficulties too. If
$P=\cA$ and $\Delta(\bb)=0$, the condition (\ref{ws106}) does not
lead to the Leibniz rule (\ref{+a20}), i.e., derivations of a
$\cK$-ring $\cA$ are not first-order differential operators. In
order to overcome these difficulties, one can replace the
condition (\ref{ws106}) with the following one \cite{dublmp}.

\begin{definition} \label{ws120} \mar{ws120}
An element $\Delta\in \hm_\cK(P,Q)$ is called the first-order
differential operator of a  bimodule $P$ over a non-commutative
ring $\cA$ if it obeys a condition
\be
&& \dl_a\circ\ol\dl_b\Delta=\ol\dl_b\circ\dl_a\Delta=0,
\qquad  a,b\in\cA, \nonumber \\
&& a\Delta(p)b -a\Delta(pb) -\Delta(ap)b +\Delta(apb)=0, \qquad p\in P.
\ee
\end{definition}

If $P$ is a commutative bimodule over a commutative ring $\cA$,
then $\dl_a=\ol\dl_a$ and Definition \ref{ws120} comes to
Definition \ref{ws131} for first-order differential operators.

First-order $Q$-valued differential operators on a module $P$
constitute a $\cZ_\cA$-module $\dif_1(P,Q)$.

Derivations of $\cA$ are first-order differential operators in
accordance with Definition \ref{ws120}. The Chevalley--Eilenberg
coboundary operator $d$ (\ref{+840}) also is well.

If derivations of a non-commutative $\cK$-ring $\cA$ are
first-order differential operators on $\cA$, it seems natural to
think of their compositions as being particular higher order
differential operators on $\cA$.

By analogy with Definition \ref{ws131}, one may try to generalize
Definition \ref{ws120} by means of the maps $\dl_a$ (\ref{ws133a})
and $\ol\dl_a$ (\ref{ws133}). A problem lies in the fact that, if
$P=Q=\cA$, the compositions $\dl_a\circ\dl_b$ and
$\ol\dl_a\circ\ol\dl_b$ do not imply the Leibniz rule and, as a
consequence, compositions of derivations of $\cA$ fail to be
differential operators \cite{book05,book12}.

This problem can be solved if $P$ and $Q$ are regarded as left
$\cA$-modules \cite{lunts}. Let us consider a $\cK$-module
$\hm_\cK (P,Q)$ provided with the left $\cA-\cA^\bll$ module
structure (\ref{5.29}). We denote by $\cZ_0$ its center, i.e.,
$\dl_a\Phi=0$ for all $\Phi\in\cZ_0$ and $a\in\cA$. Let $\cI_0=\ol
\cZ_0$ be an $\cA-\cA^\bll$ submodule of $\hm_\cK (P,Q)$ generated
by $\cZ_0$. Let us consider:

(i) the quotient $\hm_\cK (P,Q)/\cI_0$,

(ii) its center $\cZ_1$,

(iii) an $\cA-\cA^\bll$ submodule $\ol \cZ_1$ of $\hm_\cK
(P,Q)/\cI_0$ generated by $\cZ_1$,

(iv) an $\cA-\cA^\bll$ submodule $\cI_1$ of $\hm_\cK (P,Q)$ given
by the relation $\cI_1/\cI_0=\ol \cZ_1$.

\noindent Then we define the $\cA-\cA^\bll$ submodules $\cI_r$,
$r=2,\ldots$, of $\hm_\cK (P,Q)$ by induction as
$\cI_r/\cI_{r-1}=\ol \cZ_r$, where $\ol \cZ_r$ is the
$\cA-\cA^\bll$ module generated by the center $\cZ_r$ of the
quotient $\hm_\cK (P,Q)/\cI_{r-1}$.

\begin{definition} \label{ws135} \ref{ws135}
Elements of the submodule $\cI_r$ of $\hm_\cK (P,Q)$ are said to
be left $r$-order $Q$-valued  differential operators of an
$\cA$-bimodule $P$ \cite{lunts}.
\end{definition}

If $\cA$ is a commutative ring, Definition \ref{ws135} comes to
Definition \ref{ws131}.

Let $P=Q=\cA$. Any zero-order differential operator on $\cA$ in
accordance with Definition \ref{ws135} takes a form $a\to cac'$
for some $c,c'\in\cA$. Any derivation $u\in\gd\cA$ of a $\cK$-ring
$\cA$ is a first-order differential operator in accordance with
Definition \ref{ws135}. Indeed, it is readily observed that
\be
(\dl_au)(b)= au(b)-u(ab)=-u(a)b, \qquad b\in\cA,
\ee
is a zero-order differential operator for all $a\in\cA$. The
compositions $au$, $u\bll a$ (\ref{5.29}), $ua$, $a\bll u$
(\ref{ws105}) for any $u\in\gd\cA$, $a\in\cA$ and the compositions
of derivations $u_1\circ\cdots\circ u_r$ also are differential
operators on $\cA$ in accordance with Definition \ref{ws135}.

By analogy with Definition \ref{ws135}, one can define
differential operators on right $\cA$-modules as follows.

\begin{definition} \label{ws151} \mar{ws151}
Let $P$, $Q$ be seen as right $\cA$-modules over a non-commutative
$\cK$-ring $\cA$. An element $\Delta\in\hm_\cK(P,Q)$ is said to be
the right zero-order $Q$-valued differential operator on $P$ if it
is a finite sum $\Delta=\Phi^i b_i$, $b_i\in\cA$, where
$\ol\dl_a\Phi^i=0$ for all $a\in\cA$. An element
$\Delta\in\hm_\cK(P,Q)$ is called the right differential operator
of order $r>0$ on $P$ if it is a finite sum
\be
\Delta(p)=\Phi^i(p)b_i +\Delta_{r-1}(p), \qquad b_i\in\cA,
\ee
where $\Delta_{r-1}$ and $\ol\dl_a\Phi^i$ for all $a\in\cA$ are
right $(r-1)$-order differential operators.
\end{definition}

Definition \ref{ws135} and Definition \ref{ws151} of left and
right differential operators on $\cA$-bimodules are not
equivalent, but one can combine them as follows.

\begin{definition} \label{ws152} \mar{ws152}
Let $P$ and $Q$ be bimodules over a non-commutative $\cK$-ring
$\cA$. An element $\Delta\in\hm_\cK(P,Q)$ is a two-sided
zero-order $Q$-valued differential operator on $P$ if it is either
a left or right zero-order differential operator. An element
$\Delta\in\hm_\cK(P,Q)$ is said to be the two-sided differential
operator of order $r>0$ on $P$ if it is brought both into the form
\be
&& \Delta=b_i\Phi^i +\Delta_{r-1},\qquad b_i\in\cA,\\
&& \Delta=\ol\Phi^i\ol b_i +\ol\Delta_{r-1}, \qquad \ol b_i\in\cA,
\ee
where $\Delta_{r-1}$, $\ol\Delta_{r-1}$ and $\dl_a\Phi^i$,
$\ol\dl_a\ol\Phi^i$ for all $a\in\cA$
  are two-sided $(r-1)$-order differential operators.
\end{definition}

One can think of this definition as a generalization of Definition
\ref{ws120} to higher order differential operators.

It is readily observed that two-sided differential operators
described by Definition \ref{ws152} need not be left or right
differential operators, and \textit{vice versa}. At the same time,
derivations of a $\cK$-ring $\cA$ and their compositions obey
Definition \ref{ws152}.

\end{document}